\documentclass[11pt]{article}

\usepackage[utf8]{inputenc}
\usepackage{a4wide}
\usepackage{graphicx}
\usepackage[hmarginratio={1:1},     
  vmarginratio={1:1},     
  textwidth=17cm,        
  textheight=24cm,
  heightrounded,]{geometry}
  
\usepackage{amsmath,accents}
\usepackage{amsthm}
\usepackage{amssymb}

\usepackage{mathrsfs, mathtools}

\usepackage{tikz}
\usepackage{collectbox}
\usepackage{dsfont}
\usetikzlibrary{matrix}

\makeatletter

\let\emph\undefined
\newcommand{\emph}[1]{\textsl{#1}}

\usepackage{hyperref}
\usepackage{dsfont}

\numberwithin{equation}{section}
\usepackage{mathtools}
\mathtoolsset{showonlyrefs}
\numberwithin{equation}{section}
\usepackage[percent]{overpic} 
\newtheoremstyle{style1}
  {13pt}
  {13pt}
  {}
  {}
  {\normalfont\bfseries}
  {.}
  {.5em}
  {}
\theoremstyle{style1}

\newtheorem{definition}[equation]{Definition}
\newtheorem{example}[equation]{Example}
\newtheorem{remark}[equation]{Remark}

\usepackage{tocloft}

\newtheoremstyle{style2}
  {13pt}
  {13pt}
  {\slshape}
  {}
  {\normalfont\bfseries}
  {.}
  {.5em}
  {}

\theoremstyle{style2}

\newtheorem{lemma}[equation]{Lemma}
\newtheorem{theorem}[equation]{Theorem}
\newtheorem{proposition}[equation]{Proposition}

\usepackage{tikz}
\usetikzlibrary{matrix,arrows,decorations.pathmorphing}
\usepackage{tikz-cd}

\newcommand{\R}{\mathbb{R}}
\newcommand{\C}{\mathbb{C}}
\newcommand{\Z}{\mathbb{Z}}

\newcommand{\Ca}{\mathcal{C}}

\newcommand{\Gscr}{\mathscr{G}}

\newcommand{\End}{\mathsf{End}}
\newcommand{\Aut}{\mathsf{Aut}}
\newcommand{\Hom}{\mathsf{Hom}}
\newcommand{\Obj}{\mathsf{Obj}}
\newcommand{\im}{\mathrm{im}}

\newcommand{\id}{\text{id}}

\newcommand{\tr}{\text{tr}}
\newcommand{\pr}{\text{pr}}
\newcommand{\Ds}{/\hspace{-0.12cm}/}
\newcommand{\DS}{\text{/\hspace{-0.12cm}/}}
\newcommand{\iu}{\mathrm{i}}
\newcommand{\dd}{\mathrm{d}}

\newcommand{\BunG}{\mathsf{Bun}_G}
\newcommand{\BunD}{\mathsf{Bun}_D}
\newcommand{\fvs}{{\mathsf{Vect}}}
\newcommand{\Tvs}{{\mathsf{2Vect}}}
\newcommand{\Cob}{{\mathsf{Cob}}}
\newcommand{\CobF}{{{\mathsf{Cob}}^{\mathscr{F}}_{n}}}

\newcommand{\Hilb}{{\mathsf{Hilb}}}
\newcommand{\Fund}{\mathsf{Fund}}
\newcommand{\GCob}{G\text{-}\Cob_n}
\newcommand{\EGCob}{G\text{-}\Cob_{n,n-1,n-2}}
\newcommand{\ECob}{\Cob_{n,n-1,n-2}}
\newcommand{\DCob}{D\text{-}\Cob_n}
\newcommand{\EDCob}{D\text{-}\Cob_{n,n-1,n-2}}
\newcommand{\DTFT}{D\text{-}\mathsf{TFT}_n}
\newcommand{\Cat}{\mathsf{Cat}}

\newcommand{\Grp}{{\mathsf{Grp}}}
\newcommand{\Grpd}{{\mathsf{Grpd}}}

\DeclareMathSymbol{\Phiit}{\mathalpha}{letters}{"08} 
\DeclareMathSymbol{\Psiit}{\mathalpha}{letters}{"09}
\DeclareMathSymbol{\Sigmait}{\mathalpha}{letters}{"06}
\DeclareMathSymbol{\Xiit}{\mathalpha}{letters}{"04}
\DeclareMathSymbol{\Piit}{\mathalpha}{letters}{"05}
\DeclareMathSymbol{\Gammait}{\mathalpha}{letters}{"00}
\DeclareMathSymbol{\Omegait}{\mathalpha}{letters}{"0A}
\DeclareMathSymbol{\Upsilonit}{\mathalpha}{letters}{"07}
\DeclareMathSymbol{\thetait}{\mathalpha}{letters}{"02}
\let\Phi\undefined\newcommand{\Phi}{\Phiit}
\let\Sigma\undefined\newcommand{\Sigma}{\Sigmait}
\let\Psi\undefined\newcommand{\Psi}{\Psiit}

\begin{document}

\begin{flushright}
\small
{\sf EMPG--18--23}
\end{flushright}

\vspace{10mm}

\begin{center}
	\textbf{\LARGE{'t Hooft anomalies of discrete gauge theories \\[4pt] and non-abelian group cohomology}}\\
	\vspace{1cm}
	{\large Lukas Müller} \ \ and \ \ {\large Richard J. Szabo}

\vspace{5mm}

{\em Department of Mathematics\\
Heriot-Watt University\\
Colin Maclaurin Building, Riccarton, Edinburgh EH14 4AS, U.K.}\\
and {\em Maxwell Institute for Mathematical Sciences, Edinburgh, U.K.}\\
and {\em The Higgs Centre for Theoretical Physics, Edinburgh, U.K.}\\
Email: \ {\tt lm78@hw.ac.uk \ , \ R.J.Szabo@hw.ac.uk}

\vspace{2cm}

\end{center}

\begin{abstract}
\noindent
We study discrete symmetries of Dijkgraaf-Witten theories and their
gauging in the framework of (extended) functorial quantum field theory.
Non-abelian group cohomology is used to describe discrete symmetries
and we derive concrete conditions for such a symmetry to admit 't~Hooft 
anomalies in terms of the Lyndon-Hochschild-Serre spectral sequence.  
We give an explicit realization of a discrete gauge theory with 't~Hooft anomaly as
a state on the boundary of a higher-dimensional Dijkgraaf-Witten theory.
This allows us to calculate the 2-cocycle twisting the projective
representation of physical symmetries via transgression.    
We present a general discussion of the bulk-boundary correspondence at the level of 
partition functions and state spaces, 
which we make explicit for discrete gauge theories.     
\end{abstract}

\newpage

\tableofcontents

\section{Introduction and overview}

A Dijkgraaf-Witten theory \cite{DijkgraafWitten} is a topological 
gauge theory with a finite gauge group $D$.
These theories are mathematically well-defined 
and hence provide an interesting toy model for the
mathematical study of quantum gauge theory. 
Recently they have received a lot of attention 
in the physics literature
due to their relevance to topological
phases, see e.g.~\cite{Wang:2014oya,PhysRevB.92.045101,Yoshida:2017xqa,He:2016xpi,Cong2017,Cong:2017ffh,Delcamp2017,Tiwari2017,Wen2017}. In this context 
properties of topologically protected 
boundary states can be described by quantum field 
theories with anomalies, which are compensated by 
an anomaly inflow from the bulk. 
Anomalies in discrete gauge theories show 
interesting new features that are absent in the continuous
case~\cite{Kapustin:Symmetries,Tachikawa:2017gyf}.

The purpose of the present paper is two-fold: 
\begin{itemize}
\item 
We present a 
detailed study of 't Hooft anomalies of Dijkgraaf-Witten theories
in the framework of functorial quantum field theory.
This is based on the framework originally discussed in~\cite{Kapustin:Symmetries}. 
Our formulation further enables the
consideration of non-abelian gauge groups in terms of 
non-abelian group cohomology.     

\item 
We explicitly describe discrete gauge theories with anomalies 
as boundary phases in the language of extended functorial 
quantum field theory~\cite{FreedAnomalies,MonnierHamiltionianAnomalies}.
This provides a mathematically rigorous toy model for this geometric approach to anomalies.
These boundary theories can be regarded as functorial reformulations
and extensions of the gapped boundary states originally 
proposed in~\cite{Witten:2016cio}. 
\end{itemize}      
We begin with some physical background and motivation. 
 
\subsection{Anomalies and symmetry-protected topological phases}

At the end of the last century it was realized that quantum phases of matter
exist which cannot be described by Landau's theory of symmetry breaking. 
Instead these phases can be distinguished by `topological order'
parameters which prevent them from being deformed to a trivial system
whose ground state is a factorized state. 
Since their discovery over 30 years ago, immense progress in understanding and 
classifying these topological phases has been made. 
For instance, there exists a classification for non-interacting gapped
systems in 
terms of twisted equivariant K-theory~\cite{FreedMoore} (see also~\cite{Bunk2017}). 

A fruitful approach to the study of gapped interacting systems is to 
consider the effective low-energy (long-range) continuum theory of 
a lattice Hamiltonian model. Usually these field theories are topological.
A famous example is the effective 
description of the integer quantum Hall effect in terms of 
Chern-Simons gauge theory. In this sense a gapped quantum phase may be thought of as a path-connected
component of the moduli stack of topological quantum field theories.
However, in the interacting case no complete classification exists. For this reason 
one usually restricts to tractable subclasses. In this paper we focus on 
`short-range entangled' phases \cite{Chen:2010gda}. 
A gapped phase $\Psi$ is short-range entangled if there exists a 
phase $\Psi^{-1}$ such that the `stacked' phase $\Psi \otimes \Psi^{-1}$
can be deformed by an adiabatic transformation of the Hamiltonian to a
trivial product state without closing the energy gap between the
ground state and the first excited state. 
The stacking operation of topological phases corresponds to the tensor product   
of their low-energy effective topological field theories. 
A topological field theory is called invertible if it admits an inverse with respect to 
the tensor product. This observation motivates a classification of
short-range entangled phases in terms of invertible topological field 
theories~\cite{Freed:2016rqq}. 

Let $G$ be a group.
In the case of an additional global $G$-symmetry, a non-trivial short-range entangled
phase may be trivial when the symmetry is ignored. 
Such a phase is called 
`$G$-symmetry-protected'~\cite{Chen:2010gda,Chen:2011pg}.   
A $G$-symmetry-protected phase can be understood by studying its topological
response to non-trivial background $G$-gauge fields, which is called
`gauging' the $G$-symmetry. 
For a finite symmetry group $G$, the low-energy effective field theories 
are $G$-equivariant topological field theories~\cite{Kapustin2017}. 
Classical Dijkgraaf-Witten theories provide a particularly tractable
class of invertible $G$-equivariant topological field theories. The corresponding lattice Hamiltonian
models have been constructed in e.g.~\cite{PhysRevB.92.045101,Cong2017,Cong:2017ffh,Wang2018}. 
They are classified by group cohomology. The corresponding classification 
of topological phases is called group cohomological classification~\cite{Levin2012}.
However, this is not a 
complete classification of symmetry-protected phases and more refined
classifications have been proposed, see e.g.~\cite{Kapustin2015,Freed:2016rqq,Gaiotto:2017zba, PhysRevX.8.011055}.

An essential feature of symmetry-protected topological phases is that they exhibit
`topologically protected' boundary states. These boundary states can 
be effectively described by anomalous quantum field theories. 
Under the bulk-boundary correspondence the anomaly is cancelled
by an anomaly inflow from the bulk theory. 
The presence of non-trivial global anomalies forces the boundary theory 
to be non-trivial and topologically protected.
The modern geometric point of view on field theories with 
anomaly is that they should be considered as a theory `relative to' a 
higher-dimensional invertible field theory, as is naturally suggested by their
appearance in condensed matter physics and quantum information theory. 
Reversing this logic, it follows that $n{+}1$-dimensional invertible field 
theories should classify the possible anomalies in $n$ dimensions~\cite{Wen:2013oza}.
A class of gapped boundary states for the topological phases 
described by group cohomology are quantum Dijkgraaf-Witten theories
based on a different gauge group $D$ with an anomalous $G$-symmetry~\cite{Kapustin2014}.
Anomaly in this context means 't~Hooft anomaly, i.e. an obstruction 
to gauging the $G$-symmetry~\cite{tHooft1979}, and the bulk-boundary
correspondence implements the 't~Hooft anomaly matching conditions.
     
The purpose of this paper is to study the appearance of 't Hooft anomalies in the
mathematically rigorous framework of functorial field theories and to realize 
anomalous field theories as gauge theories relative to higher-dimensional topological
phases. 
Before outlining precisely what we do, let us
first informally review some of the main mathematical background.              

\subsection{Anomalies in functorial quantum field theory}

We formulate our results in the framework of functorial field theories. 
The idea is to give an axiomatic framework for the partition function of a
quantum field theory. In physics, the partition function $Z(M)$ on an 
$n$-dimensional manifold $M$ is calculated by
the Feynman path integral of an exponentiated action functional over the space of
dynamical field configurations on $M$; so far there is no mathematically 
well-defined theory of such path integration in general. 
The axioms of functorial field theories are derived from the properties that
such an integration would satisfy in the case that the action
functional is an integral of a local Lagrangian density on $M$.\footnote{In the case of discrete gauge theory there exists
a well-defined integration 
theory (see for example~\cite[Appendix A]{OFK}) which satisfies the axioms of a functorial field theory.}
A quantum field theory should also assign 
a Hilbert space of states $Z(\Sigma)$ 
to every $n{-}1$-dimensional manifold $\Sigma$. They satisfy
$Z(\Sigma\sqcup\Sigma')\cong Z(\Sigma)\otimes_\C Z(\Sigma')$, i.e. the
state space of non-interacting systems is given by the tensor product
of the corresponding Hilbert spaces.
In any quantum field theory there exists a time evolution 
operator (propagator)
\begin{align}
Z([t_0,t_1]\times \Sigma)\colon Z(\Sigma) \longrightarrow Z(\Sigma)
\end{align}         
from time $t_0$ to $t_1$. We think of this operator as associated to the 
cylinder $[t_0,t_1]\times \Sigma$. They satisfy 
$Z([t_1,t_2]\times \Sigma) \circ Z([t_0,t_1]\times \Sigma)=Z([t_0,t_2]\times \Sigma)$. 
The path integral should also allow for the construction of a more general operator 
$Z(M)\colon Z(\Sigma_-)\longrightarrow Z(\Sigma_+)$ for every manifold $M$
with boundary $\Sigma_-\sqcup \Sigma_+$, such that the gluing of
manifolds corresponds to the composition of linear maps. Such a manifold is called a cobordism 
from $\Sigma_-$ to $\Sigma_+$.

These considerations motivate the definition of a functorial field theory,
generalising Atiyah's definition of topological field 
theories~\cite{Atiyah1988} and Segal's definition of conformal field
theories~\cite{Segal1988}, as a symmetric monoidal functor 
\begin{align}
Z \colon \Cob_n^\mathscr{F} \longrightarrow \Hilb \ ,
\end{align}    
where $\Cob_n^\mathscr{F}$ is a category modelling physical spacetimes with
non-dynamical background fields $\mathscr{F}$ and $\Hilb$ is the category of
complex Hilbert spaces. 
Roughly speaking, $\Cob_n^\mathscr{F}$ contains closed $n{-}1$-dimensional manifolds 
with background fields as objects, 
and as morphisms the $n$-dimensional cobordisms as well as additional limit morphisms corresponding
to diffeomorphisms which are compatible with the background fields.  
The additional morphisms encode symmetries. 
Evaluating $Z$ on a closed $n$-dimensional manifold $M$ gives rise to a linear
map $\C \cong Z(\varnothing)\longrightarrow Z(\varnothing) \cong \C$ which 
can be identified with a complex number $Z(M)$, the partition function of 
$Z$ on $M$. This definition can be thought of as a prescription for computing a manifold invariant $Z(M)$ by cutting manifolds into simpler pieces and studying the quantum field theory
on these pieces.

We now turn our attention to the description of anomalies.
The partition function of an $n{-}1$-dimensional 
quantum field theory $Z$ with anomaly
described by an invertible field theory $L\colon \Cob_n^\mathscr{F} \longrightarrow \Hilb$ 
on an $n{-}1$-dimensional manifold $\Sigma$ takes values in the one-dimensional vector space $L(\Sigma)$, instead of $\C$. 
It is possible to pick a non-canonical isomorphism $L(\Sigma)\cong \C$
to identify the partition function with a complex number. The group of
symmetries acts non-trivially on $L(\Sigma)$ encoding the breaking of 
the symmetry in the quantum field theory $Z$. 

To also incorporate the description of the state space of $Z$ on an 
$n{-}2$-dimensional manifold $S$ we need to promote $L$ to an extended field
theory $E$ which assigns $\C$-linear categories to $n{-}2$-dimensional manifolds
such that $Z(S)$ can be considered as an object of $E(S)$. 
In other words, $E$ should be an extended functorial field theory, i.e. a 
symmetric monoidal 2-functor 
\begin{align}
E \colon \Cob_{n,n-1,n-2}^\mathscr{F} \longrightarrow \Tvs \ ,
\end{align}  
where $\Cob_{n,n-1,n-2}^\mathscr{F}$ is a suitable extension of
$\Cob_{n}^\mathscr{F}$; see Section \ref{Sec: DW} for details. 
There are different possible choices for the target bicategory. For simplicity we 
restrict ourselves to Kapranov-Voevodsky 2-vector spaces \cite{KV94}. 

Requiring that $E$ is an invertible field theory implies that there is 
a non-canonical equivalence of categories $E(S)\cong \fvs$ which
allows one
to identify the state space of the anomalous theory with a vector space. 
We can subsume the ideas sketched above in the following concise definition:
A {quantum field theory with anomaly} is a natural symmetric monoidal 2-transformation
\begin{align}
Z\colon 1 \Longrightarrow \tr\, E
\end{align}
between a trivial field theory and a certain truncation of $E$; see 
Section~\ref{Sec: Boundary theories} for details.
This formalism allows one to compute the 2-cocycle twisting the projective
representation of the symmetry group on the state space completely in 
terms of the extended field theory $E$~\cite{MonnierHamiltionianAnomalies}.
Anomalous theories formulated in this way are a special case of relative field
theories~\cite{RelativeQFT} and are closely related to twisted quantum field 
theories~\cite{Stolz:2011zj,Johnson-Freyd:2017ykw}. 
The present paper describes 't Hooft anomalies of discrete gauge theories 
as relative field theories. We shall now give an overview of our constructions
and findings. 

\subsection{Summary of results and outline}

One of the main achievements of this paper is to give a mathematical description of 
symmetries of Dijkgraaf-Witten theories and their gauging 
in the framework of functorial 
field theory which is motivated by physical considerations. 
Let $D$ be a finite group and $n$ a natural number.
The possible topological actions for $n$-dimensional Dijkgraaf-Witten
theories with gauge group $D$ are classified by the group cohomology
of $D$ or equivalently by the singular cohomology of 
the classifying space $BD$ with coefficients 
in $U(1)$ \cite{DijkgraafWitten,FreedQuinn}. Let $\omega \in Z^n(BD;U(1))$ be an $n$-cocycle
and $M$ an $n$-dimensional manifold. 
Let $P$ be a $D$-gauge field on $M$ with classifying map 
$\psi_P\colon M \longrightarrow BD$.
The action of the Dijkgraaf-Witten theory $L_\omega$ evaluated at $P$
is given by
\begin{align}
\exp(2\pi \,\iu\, S_{\text{DW}})\coloneqq \int_M\, \psi_P^*\, \omega \ .
\end{align} 
The quantum theory can be defined by appropriately summing 
over isomorphism classes of $D$-bundles. We review Dijkgraaf-Witten theories
in detail in Section \ref{Sec: DW}. 

In general, a physical symmetry group $G$ acts on gauge fields only up to gauge transformations. 
Since for finite gauge groups, gauge transformations can be naturally identified with 
homotopies of classifying maps, we define such an action as a
homotopy coherent action of $G$ on 
$BD$ (Definition~\ref{Def: homotopy Coherent action}). We show that,
up to equivalence, homotopy 
coherent actions on $BD$ are described by non-abelian group 2-cocycles. If $D$ is abelian,
this description agrees with the description in \cite{Kapustin:Symmetries}. 
Non-abelian 2-cocycles classify extensions of $G$ by $D$:
\begin{align}
1\longrightarrow D \overset{\iota}{\longrightarrow} \widehat{G} \overset{\lambda}{\longrightarrow} G \longrightarrow 1 \ ,
\end{align} 
as we review in Section \ref{Sec: Non-abelian group cohomology}. 
This extension has a natural physical interpretation: It describes how to 
combine $D$- and $G$-gauge fields into a single $\widehat{G}$-gauge field. When the extension is non-trivial, i.e. $\widehat{G}$ is not a product group $D\times G$, one says that the $G$-symmetry is `fractionalized'~\cite{Wang:2017loc}.

A homotopy coherent action on $BD$ induces a homotopy coherent action on the 
collection of classical $D$-gauge theories. 
Homotopy fixed points of this action are defined to be classical 
field theories with $G$-symmetry 
(Definition~\ref{Def: Field theory with k. symmetry}).
An essential feature of homotopy fixed points is that they are a
structure, not a property.  
In Proposition~\ref{Prop: Classical symmetry}
we show that if the topological action is preserved by the action of $G$ 
(Definition~\ref{Def: Preserved}), 
then the corresponding Dijkgraaf-Witten theory can be equipped with 
a homotopy fixed point structure.

An internal symmetry of a quantum field theory acts on its Hilbert space
of states. This motivates the definition of a functorial quantum
  field theory with {internal $G$-symmetry} as a functor 
\begin{align}
\Cob_n^\mathscr{F} \longrightarrow G\text{-}\mathsf{Rep}
\end{align}  
to the category $G\text{-}\mathsf{Rep}$ of representations 
of $G$.   
We show in Proposition \ref{Prop: symmetry extends to quantum} that 
classical symmetries of Dijkgraaf-Witten theories induce internal symmetries 
of the quantized theory. 
This shows that discrete gauge theories are anomaly-free in the sense that all symmetries 
extend to the quantum level. 
A discussion of symmetries of discrete classical and quantum 
gauge theories is the content of 
Section~\ref{Sec: Finite symmetries}.  

Anomalies appear 
as an obstruction to gauging the $G$-symmetry, 
i.e. to coupling it to 
non-trivial background gauge fields (Definition \ref{Def: Gauging}). Anomalies of this 
type are called {'t Hooft anomalies}.   
Gauging the $G$-symmetry can be achieved by finding a 
topological action $\widehat{\omega}$ for a $\widehat{G}$-gauge theory which restricts to $\omega$ and 
performing a path integral over
$D$-gauge fields. 
Mathematically, this can be described by       
the partial orbifold (pushforward) construction of \cite{OFK}. 
In Theorem \ref{Theorem: Gauging} we prove that 
the partial orbifold construction of the classical
gauge theory corresponding to $\widehat{\omega}$ gauges
the $G$-symmetry. We discuss the gauging of discrete symmetries
in Section \ref{Sec: Gauging}.

However, in general it might be impossible to find a
topological action which restricts correctly. In this 
case we say that the corresponding symmetry has a 
't Hooft anomaly. 
The obstructions for $\widehat{\omega}$ to exist are encoded
in the Lyndon-Hochschild-Serre spectral sequence.
For an $n$-dimensional field theory there are $n$ obstructions
which need to vanish. 
In Proposition \ref{Prop: Obstructions} we show that if all 
obstructions except the last one vanish then there exists an $n{+}1$-dimensional     
topological action $\theta $ for a discrete $G$-gauge theory,
together with an $n$-cochain $\omega'$ in $C^n(B\widehat{G}; U(1))$ satisfying
$\iota^* \omega' = \omega$ and $\delta \omega' = \lambda^* \theta$;
physically, this is interpreted as saying that the corresponding
symmetry-protected topological state becomes trivial when its fields
are regarded as $n$-dimensional $\widehat{G}$-gauge fields rather than
as $n{+}1$-dimensional $G$-gauge fields.
These obstructions are studied in Section \ref{Sec: Obstructions}.

Based on this result we construct a boundary 
quantum field theory $Z_{\omega'}$ encoding the anomaly in Section~\ref{Sec: Realisation on boundary}. 
Let us give an informal description of $Z_{\omega'}$ here.
The fact that $\omega'$ is not closed implies that 
$\int_M\, \psi_{\widehat P}^*\,{\omega'}$
is not gauge-invariant for a general $\widehat{G}$-bundle $\widehat P $ on $M$. 
Under a gauge transformation 
$\widehat{h}\colon \widehat P \longrightarrow \widehat P'$
the value of $\int_M\, \psi_{\widehat P}^*\,{\omega'}$
changes by multiplication with\footnote{This integral is not actually
well-defined as a complex number, see Section \ref{Sec: DW} for details.
We ignore this subtlety in the present section.} 
$\int_{[0,1]\times M}\, \widehat{h}^* \delta {\omega'}$, 
where we consider $\widehat{h}$ as a homotopy $[0,1]\times M \longrightarrow
B\widehat G$.
We can rewrite this integral as $\int_{[0,1]\times M}\, (\lambda_*\widehat{h}\,)^* {\theta}$.
This is exactly the value of $L_\theta$ evaluated on $\lambda_*\widehat{h}$,
which shows that the anomaly is controlled by the bulk classical gauge theory 
$L_\theta$. 
The rough idea for the construction of $Z_{\omega'}$ is to modify the partial 
orbifold construction used in Section \ref{Sec: Gauging} in a way suited
to the construction of boundary states. 
Let us fix a $G$-bundle $P$ on $M$. To define the partition function 
we want to perform an integration over the preimage of $P$ under $\lambda_*$.
However, in the presence of gauge transformations, requiring two bundles to be
the same is not natural. Hence we use the homotopy fibre 
$\lambda_*^{-1}[P]$ as a groupoid with objects the pairs $(\widehat P, h)$ of a 
$\widehat G$-bundle $\widehat P$ and a gauge transformation 
$h\colon \lambda_*\widehat{P} \longrightarrow P $. Morphisms are gauge 
transformations $\widehat{h} \colon \widehat{P} \longrightarrow \widehat{P}'$
which are compatible
with $h$ and $h'$.
We show that 
\[
L_{\omega'}(M)\coloneqq \int_M\, \psi_{\widehat{P}}^*\,\omega' \ \int_{[0,1]\times M}\, h^*\theta
\]           
is gauge-invariant with respect to morphisms in $\lambda_*^{-1}[P]$. 
Hence we define the partition function of $Z_{\omega'}$ on $M$ as\footnote{We recall integration over essentially finite groupoids in Section \ref{Sec: DW}.} 
\[
Z_{\omega'}(M) := \int_{\lambda_*^{-1}[P]}\, L_{\omega'}(M) \ .
\]  
Similar adaptations of the partial orbifold construction lead to 
the definition of state space in Section~\ref{Sec: State space}. 
The groupoid of symmetries acts only projectively on this state space.
Using a result of~\cite{EHQFT} we show
that the 2-cocycle twisting this
projective representation is the transgression of $\theta$ to the groupoid
of $G$-bundles. With this construction we provide an explicit
demonstration of the anomaly inflow mechanism at the level of both
partition functions and state spaces, which renders the composite bulk-boundary 
field theory free from anomalies.

\subsubsection*{Related work}

Section \ref{Sec:Finite symmetries and 't Hooft anomalies} can be regarded as 
a formulation of the main ideas of \cite{Kapustin:Symmetries} in the 
framework of functorial field theories and homotopy theory. 
The abstract formulation allows us to also treat non-abelian gauge groups 
$D$, which are not considered in \cite{Kapustin:Symmetries}. 
A large part of our discussion should generalise and provide
functorial descriptions of higher-form
symmetries, as in
e.g.~\cite{Kapustin2013,Kapustin:2014gua,Gaiotto2014,Thorngren2015,Tachikawa:2017gyf,Delcamp2018,Benini2018},
and of invertible topological sigma-models, as in~\cite{Thorngren:2018wwt,Thorngren:2017vzn}. It should also be straightforward to include time-reversal symmetries using the techniques of~\cite{Young2018}.

The relative field theory constructed in Section 
\ref{Sec: Realisation on boundary} can be regarded as a formulation of the 
gapped boundary states constructed in \cite{Witten:2016cio}, together with the
explicit lattice gauge theory and Hamiltonian constructions of~\cite{Wang:2017loc}, in the framework 
of relative field theories. We extend this construction to give the state space 
of the quantum field theory explicitly; our boundary field theory $Z_{\omega'}$
formally realises the new boundary degrees of freedom of the
$\widehat{G}$-symmetry extended boundary states from~\cite{Wang:2017loc} in this language. 
Proposition \ref{Prop: Obstructions} provides a clear relation between the 
works~\cite{Kapustin:Symmetries} and \cite{Witten:2016cio}:
We explicitly show that if all obstructions except the last one vanish
in the spectral sequence, then
one is in the set-up of \cite{Witten:2016cio}.
Similar considerations appear in~\cite{Thorngren2015}, see
also~\cite{Wang:2017loc,Tachikawa:2017gyf} for different perspectives.

Recently the two-dimensional Ising model 
has been formulated as a field theory relative to a three-dimensional discrete 
gauge theory with trivial topological action~\cite{Freed:2018cec}. It is 
conjectured that the low-energy effective field theory of the Ising model
is topological, and topological boundary field theories are constructed using the Cobordism Hypothesis. Our formalism should provide an 
explicit construction of these field theories. 

\subsubsection*{Conventions and notation}

For the convenience of the reader, we summarise here our notation and
conventions which are used throughout this paper.
For 2-categories we use the definitions outlined in~\cite[Appendix~B]{Parity}. For integration over groupoids we use~\cite[Appendix~A]{OFK}.

\begin{itemize}
	
\item Let $G$ be a group.
We denote by $B G$ the classifying space of $G$, which for $G$ finite
is an Eilenberg-MacLane space $K(G,1)$, i.e. its only non-trivial
homotopy group is $\pi_1(BG)=G$.
Let $P$ be a principal $G$-bundle on a manifold $M$. 
We denote by $\psi_P\colon M \longrightarrow BG$ the corresponding
classifying map. 

\item Let $T$ be a topological space, $n$ a positive integer 
and $A$ an abelian group. 
We denote the pairing of chains and cochains on $T$ by 
$\langle \,\cdot\, , \,\cdot\, \rangle \colon C^n(T;A)\times C_n(T) \longrightarrow A$.

\item 
Let $G$ be a group.
The groupoid of $G$-bundles on a manifold $M$ is denoted by $\BunG(M)$. 

\item For most constructions in this paper we fix a positive integer $n$. 
After fixing such an integer, we use $M$, $\Sigma$ and $S$ to denote manifolds of dimensions $n$, $n-1$ and $n-2$, respectively.

\item 
Let $M$ be an oriented manifold. We denote by $-M$ the same manifold
equipped with the opposite orientation. 
 	
\item If $\Ca$ is a category we define $\underline{\Ca}$ to be the (strict)
$2$-category obtained by adding identity 2-morphisms to $\Ca$.

\item 
Let $\Ca$ be a monoidal category. We denote by $*\Ds \Ca$ the bicategory 
with one object and $\Ca$ as endomorphisms. 

\item 
Let $F \colon \mathcal{G}\longrightarrow \mathcal{G}'$ be a functor 
between groupoids and $g'\in \mathcal{G}'$. We denote by $F^{-1}[g'\,]$
the homotopy fibre of $g'$. Explicitly, the groupoid $F^{-1}[g'\,]$ has as objects the pairs $(g,h')$, with $g\in\mathcal{G}$ and $h':F(g)\longrightarrow g'$ an isomorphism, and morphisms $m:(g_1,h_1')\longrightarrow(g_2,h_2')$ comprising of a morphism $m:g_1\longrightarrow g_2$ such that the diagram
\begin{equation}
\begin{tikzcd}
F(g_1) \ar[rr,"F(m)"]\ar[dr, "h_1'", swap] & & F(g_2) \ar[dl,"h_2'"] \\
 & g' & 
\end{tikzcd}
\end{equation} 
commutes.

\item 
Let $\lambda \colon G\longrightarrow G'$ be a homomorphism of groups. 
We denote the induced maps $BG\longrightarrow BG'$ and $*\Ds G \longrightarrow *\Ds G'$ again by $\lambda$. 
 
\item Let $F\colon \Ca \longrightarrow \Ca'$ be a functor. We write the 
limit of $F$ as an end $\int_\Ca\, F$. 

\item Let $G$ be a finite group. We denote by $\GCob$ the category and by 
$\EGCob$ the bicategory of cobordisms equipped with maps into $BG$.	

\item We denote by $\Grp$ the category of groups. 

\item We denote by $\fvs$ the category of finite-dimensional $\C$-vector spaces. 

\item We denote by $\Tvs$ the bicategory of Kapranov-Voevodsky 2-vector spaces.
	
\item We denote by $\Cat$ the bicategory of (small) categories.
	
\end{itemize}

\subsubsection*{Acknowledgments}

We thank Simon Lentner, Ehud Meir, Christoph Schweigert and Lukas Woike for helpful discussions.
This work was supported by the COST Action MP1405 QSPACE, funded by the
European Cooperation in Science and Technology (COST). The work of L.M. was
supported by the Doctoral Training Grant ST/N509099/1 from the UK Science and Technology
Facilities Council (STFC).
The work of
R.J.S. was supported in part by the STFC
Consolidated Grant ST/P000363/1. 

\section{Dijkgraaf-Witten theory}\label{Sec: DW}

Dijkgraaf-Witten theories are topological gauge theories with finite gauge group $D$. In this section we introduce the framework of (extended) functorial quantum field theory, and subsequently construct classical and quantum Dijkgraaf-Witten theories. 

\subsection{Equivariant functorial field theories}

The idea of describing field theories as functors from a geometric category to a category of vector spaces goes back at least to Atiyah's definition of topological quantum field theories~\cite{Atiyah1988} and Segal's definition of conformal field theories~\cite{Segal1988}.   
The general framework allows for arbitrary background fields. We denote the collection of all background fields by\footnote{For concreteness one can model background fields as a ($\infty$-)stack $\mathscr{F}$ on the category of smooth $n$-dimensional manifolds (with corners). The language of stacks is not used in this paper.}  $\mathscr{F}$.
Typical examples of background fields in physics are metrics, spin or
spin$^c$ structures, framings and principal bundles with connection. If the collection of background fields does not contain a metric, the corresponding field theory is called {topological}.   
We only consider theories involving principal bundles with finite gauge group $D$ as background fields in this paper.

An $n$-dimensional functorial field theory is a symmetric monoidal functor from the geometric category $\CobF$, modelling spacetimes with background fields, to the category of vector spaces. The category
$\CobF$ is roughly defined as follows: Objects are $n{-}1$-dimensional manifolds $\Sigma$ equipped with $\mathscr{F}$-background fields. A morphism $\Sigma_1 \longrightarrow \Sigma_2$ is an $n$-dimensional cobordism, i.e. a manifold $M$ with boundary $\Sigma_1 \sqcup \Sigma_2$ equipped with compatible $\mathscr{F}$-background fields. Composition is defined by gluing along common boundaries (Figure \ref{Fig: Composition}). The disjoint union of manifolds makes $\CobF$ into a symmetric monoidal category.

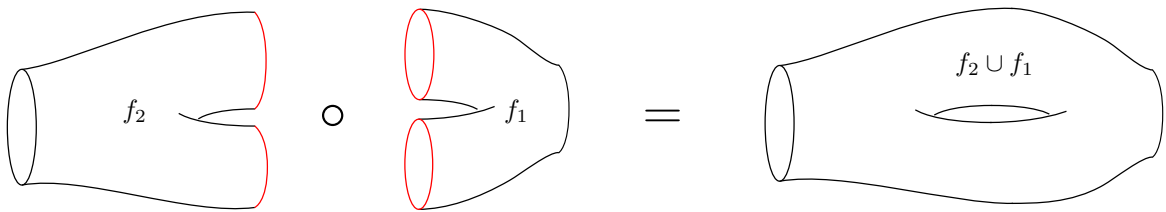
\begin{figure}[hbt]
\small
\begin{center}
\begin{overpic}[scale=0.3]{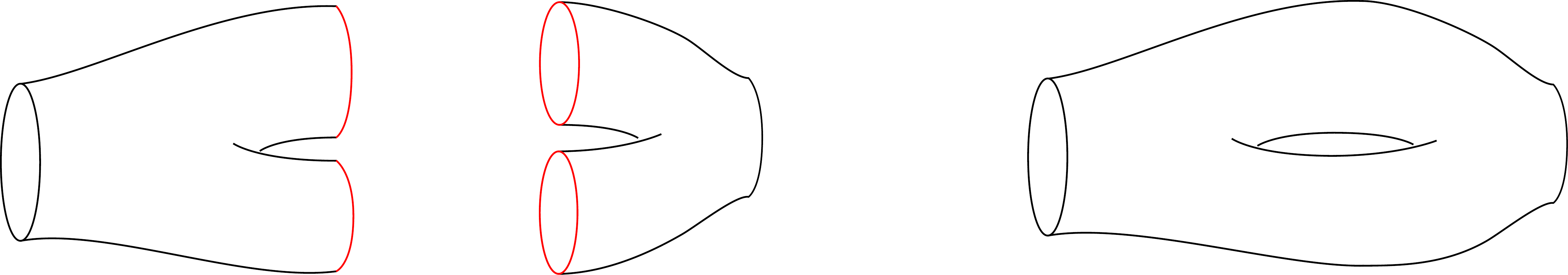}
\put(10,8){$f_2$}
\put(43,8){$f_1$}
\put(27,7){\huge$\circ$} 
\put(55,7){\huge$=$}
\put(82,12){$f_2 \cup f_1$}
\end{overpic}
\end{center}
\caption{\small Composition of two cobordisms equipped with background fields $f_1$ and $f_2$.}
\normalsize
\label{Fig: Composition}
\end{figure}

An {$n$-dimensional functorial field theory with $\mathscr{F}$-background fields} is a symmetric monoidal functor 
\begin{align}
Z \colon \CobF \longrightarrow \fvs_{\rm top}
\end{align}   
with target the category of topological vector spaces. All vector spaces appearing in this paper will be finite-dimensional and hence have a unique topology. We denote by $\fvs$ the category of finite-dimensional $\C$-vector spaces.
 
This definition ensures that $Z(\varnothing)\cong \C$ for every
functorial field theory $Z$. For this reason, $Z$ assigns to a closed
$n$-dimensional manifold $M$ with background fields $f$ a linear map
$Z(M,f)\colon Z(\varnothing)\cong\C\longrightarrow \C \cong
Z(\varnothing)$, which can be identified with a complex number, the
partition function of $Z$ on $M$.
The vector space associated to an $n{-}1$-dimensional manifold $\Sigma$ can be interpreted as the state space of the quantum field theory $Z$ on $\Sigma$.  
Making the definition precise is quite involved, see e.g. \cite{Stolz:2011zj}. A detailed discussion of this approach to quantum field theory can be found in \cite{FelixKleinSegal}. 

\begin{remark}
The functorial framework does not only describe quantum field theories, but also classical field theories. The invariant assigned to a closed manifold with background fields by a classical field theory is the exponentiated action. 
A classical field theory has no Hilbert space of states, however the exponentiated action on a manifold with boundary might not be well-defined without making additional choices on the boundary. In this case it is more natural to consider the exponentiated action not as a complex number, but as an element in a complex line associated to the additional choices. 
A famous example is Chern-Simons theory \cite{FreedCS}, where the
action functional on a manifold with boundary is not gauge-invariant
and the gauge variation is controlled by the Wess-Zumino term. The
$n$-dimensional functorial field theory corresponding to a classical field theory assigns these complex lines to $n{-}1$-dimensional manifolds.  
\end{remark}  

All theories appearing in this paper are homotopy quantum field theories \cite{turaev2010homotopy} with aspherical target. For these the background fields consist of principal bundles with finite gauge group $D$ and an orientation. 
We describe principal $D$-bundles on a manifold $M$ by their classifying map $M\longrightarrow BD$. 
However, by a slight abuse of notation we will freely switch between classifying maps and principal bundles without explicitly mentioning so. 
In particular, we will pull back cochains on $BD$ along bundles, which should always be understood as a pullback along the classifying map.        
We denote the 
corresponding cobordism category by $\DCob$, for which we can make the
sketch above precise~\cite{turaev2010homotopy}. 

\begin{definition}\label{Def: D-Cob}
For an oriented manifold $M$, we denote by $-M$ the same manifold
equipped with the opposite orientation. 
Let $D$ be a finite group. We denote by $\DCob$ the symmetric monoidal category defined by:
\begin{itemize}
\item[(a)] An object is a closed oriented $n{-}1$-dimensional manifold $\Sigma$ equipped with a continuous map $\varphi \colon \Sigma \longrightarrow BD$.

\item[(b)] A morphism $(\Sigma_1,\varphi_1)\longrightarrow (\Sigma_2,\varphi_2)$ is an equivalence class of triples $(M,\psi,\phi)$, where $M$ is a compact oriented manifold with boundary, $\phi\colon-\Sigma_1 \sqcup\Sigma_2 \longrightarrow \partial M$ is an orientation-preserving diffeomorphism, and $\psi \colon M \longrightarrow BD$ is a continuous map such that $\varphi_i= \psi \circ \phi|_{\Sigma_i}$ for $i=1,2$. 

Two such triples are equivalent if the underlying manifolds differ by a diffeomorphism $ M \longrightarrow M'$ relative to the boundary such that $M\overset{\psi}{\longrightarrow} BG$ and $M{\longrightarrow} M' \overset{\psi'}{\longrightarrow} BD$ differ by a homotopy relative to the boundary.    

\end{itemize}
The composition is defined by gluing manifolds and continuous maps along boundaries. The monoidal structure is given by the disjoint union of manifolds.
\end{definition}

\begin{definition}
An {$n$-dimensional $D$-equivariant field theory} is a symmetric monoidal functor 
\[
\DCob \longrightarrow \fvs \ . 
\]
\end{definition} 

We denote by $\DTFT$ the category of $n$-dimensional $D$-equivariant
field theories and natural symmetric monoidal transformations. Let $\Cob_n$ be the category $e\text{-}\Cob_n$, where $e$ is the unique group with one element. Objects of this category are oriented spacetimes without background gauge fields. 

Later on we will also need extended field theories which assign Kapranov-Voevodsky 2-vector spaces, i.e.~semi-simple $\C$-linear categories with finitely-many simple objects~\cite{KV94}, to closed $n{-}2$-dimensional manifolds $S$ equipped with maps $\xi \colon S\longrightarrow BD$, functors to $n{-}1$-dimensional cobordisms equipped with compatible maps into $BD$, and natural transformations to $n$-dimensional cobordisms between cobordisms (manifolds with corners) equipped with compatible maps into $BD$.
To define extended field theories properly one has to introduce a symmetric monoidal bicategory $\EDCob$ similar to the category of Definition~\ref{Def: D-Cob}. For a concrete definition we refer to~\cite{EOFK}. An {$n$-dimensional extended $D$-equivariant field theory} is a symmetric monoidal 2-functor
\begin{align}
\EDCob \longrightarrow \Tvs
\end{align}
where $\Tvs$ is the symmetric monoidal 2-category of Kapranov-Voevodsky 2-vector spaces. The only extended field theory we need in this paper is the classical Dijkgraaf-Witten theory explicitly described in Section~\ref{Sec: Calssical gauge theory}.

The symmetric monoidal structure of $\Tvs$, i.e. the Deligne tensor product $\boxtimes$, defines a tensor product of field theories. A field theory which has an inverse with respect to this tensor product is called {invertible}. Classical field theories regarded as functorial field theories are always invertible. 
Invertible field theories are physically interesting, since they describe the low-energy effective field theories of short-range entangled topological phases and field theories with anomalies can be described as boundary states of invertible field theories (see Section~\ref{Sec: Realisation on boundary}). 

\subsection{Classical Dijkgraaf-Witten theories}\label{Sec: Calssical gauge theory}

We provide a brief summary of the construction of an invertible extended field theory $E_\omega$ depending on an $n$-cocycle $\omega \in Z^n(BD;U(1))$ which is given in \cite{EHQFT}, to which we refer for more details and all proofs. We denote the corresponding unextended field theory by $L_\omega$.

The possible actions of $n$-dimensional topological gauge theories with finite gauge group $D$ are classified by the $n$-th cohomology group of the classifying space $BD$ with coefficients in $\R / \Z $ \cite{DijkgraafWitten}. For a fixed representative $\omega'\in Z^n(BD;\R / \Z )$ of a cohomology class, the action for a $D$-bundle $ \psi \colon M \longrightarrow BD$ on an oriented closed $n$-dimensional manifold $M$ is given by 
\begin{align}
S_{\omega'}(M,\psi)= \int_M \, \psi^* \omega' \ .
\end{align}
As a real number, this action is only well-defined modulo $\Z$ by definition. The quantity with physical relevance is the exponentiated action $\exp(2\pi \,\iu\, S_{\omega'}(M,\psi))$ which takes values in $U(1)$. 
For simplicity we work from now on with a cocycle $\omega \in Z^n(BD;U(1))$ and interpret its integration over the manifold as the exponentiated action. 
Formally, we can define integration as the evaluation of $\psi^* \omega$ on a representative $\sigma_M$ of the fundamental class $[\sigma_M] \in H_n(M)$ of $M$.
This defines the partition function
\begin{align}
L_\omega (M,\psi) \coloneqq \langle \psi^*\omega , \sigma_M  \rangle \ .
\end{align}
Here the brackets denote the evaluation of cochains on chains.

On manifolds with boundary the exponentiated action is not well-defined as a complex number:
To naively extend the exponentiated action to an $n$-dimensional manifold $M$ with boundary $\partial M=\Sigma$ and $D$-bundle $\psi \colon M \longrightarrow BD$ we would pick a representative $\sigma_M \in C_n(M)$ of the fundamental class $[\sigma_M]\in H_n(M,\Sigma)$ and evaluate $\psi^*\omega$ on this representative. 
The problem is that $[\sigma_M]$ is an element in relative homology, but $\psi^*\omega$ is an element of ordinary cohomology. As a representative of a relative homology class, $\sigma_M$ has a boundary $\partial \sigma_M \in C_{n-1}(\Sigma)$ and $\langle \psi^* \omega , \sigma_M \rangle$ depends on $\partial \sigma_M$ in general. Hence the exponentiated action is not well-defined as a complex number, but rather it is an element of a complex line $L_\omega(\Sigma, \psi|_{\Sigma})$. In \cite{FreedQuinn} it is shown that the vector spaces $L_\omega (\,\cdot\,,\,\cdot\,)$ are part of a functorial field theory $L_\omega \colon \DCob \longrightarrow \fvs$. This field theory is the classical Dijkgraaf-Witten theory corresponding to the topological action $\omega$.

Now we sketch how to promote $L_\omega$ to an extended field theory $E_\omega \colon \EDCob \longrightarrow \Tvs$. 
It is convenient to introduce, for an arbitrary oriented manifold $S$ of dimension $k$, the groupoid $\Fund(S)$ whose objects are cycles $\sigma_S \in Z_k(S)$ representing the fundamental class of $S$ and morphisms $\tau \colon \sigma_S \longrightarrow \sigma'_S$ are given by chains $\tau \in C_{k+1}(S)$ satisfying $\partial \tau = \sigma_S'-\sigma_S$. 
To an object of $\EDCob$, i.e. an oriented $n{-}2$-dimensional manifold $S$ equipped with a continuous map $\xi \colon S \longrightarrow BD$, we assign the 2-vector space $E_\omega(S,\xi)$ whose objects are finite formal sums $\sum_{i=1}^p\, V_i * \sigma_i$ of representatives $\sigma_i$ of the fundamental class of $S$ and finite-dimensional vector spaces $V_i$. We will abbreviate $\C * \sigma$ with $\sigma$. The space of morphisms between $\sigma,\sigma' \in E_\omega(S,\xi)$ is given by
\begin{align} 
\Hom_{E_\omega(S,\xi)}(\sigma,\sigma') := \frac{\mathbb{C}[\Hom_{\Fund(S)}(\sigma,\sigma')]}{\sim_\omega} \ ,      \label{defmorphismon0cells}
\end{align} 
where $\mathbb{C}[\Hom_{\Fund(S)}(\sigma,\sigma')]$ is the complex vector space generated by the set $\Hom_{\Fund(S)}(\sigma,\sigma')$, and for two morphisms $\tau,\widetilde \tau : \sigma \longrightarrow \sigma'$ we make the identification
\begin{align} 
\widetilde \tau \sim_\omega \langle \xi^* \omega,\Lambda\rangle \, \tau \ ,  \label{defmorphismon0cells2} 
\end{align} 
whenever there exists $\Lambda \in C_n(S)$ such that $\widetilde \tau-\tau=\partial \Lambda$.  
In \eqref{defmorphismon0cells2} the particular choice of $\Lambda$ is immaterial. 
In order to obtain the morphism spaces between arbitrary objects in $E_\omega(S,\xi)$, we extend \eqref{defmorphismon0cells} bilinearly:
\begin{align}
\Hom_{E_\omega(S,\xi)} \Big( \sum\limits_{i=1}^p \, V_i * \sigma_i
  \ ,\ \sum\limits_{j=1}^m \, V_j * \sigma_j \Big) = \bigoplus_{i=1}^p \ \bigoplus_{j=1}^m \, \Hom_\fvs(V_i,V_j) \otimes_{\C} \Hom_{E_\omega(S,\xi)} (\sigma_i, \sigma_j) \ .
\end{align} 
Composition is defined by matrix multiplication together with composition in $\fvs$ and $\Fund (S)$. 
The category $E_\omega (S,\xi)$ carries a natural $\fvs$-module structure given by 
\begin{align}
\fvs \times E_\omega (S,\xi) & \longrightarrow E_\omega (S,\xi) \\
V\times (W*\sigma) & \longmapsto (V\otimes_\C W)*\sigma \ .
\end{align}

Given a 1-morphism $(\Sigma,\varphi) \colon (S_1,\xi_1)\longrightarrow (S_2,\xi_2)$ in $\EDCob $, together with representatives $\sigma_1$ and $\sigma_2$ for the fundamental classes of $S_1$ and $S_2$, respectively, we define the vector space 
\begin{align}
\Sigma^\varphi(\sigma_2,\sigma_1) = \frac{\C[\Fund^{\sigma_2}_{\sigma_1}(\Sigma)]}{\sim_\varphi}\ ,
\end{align}  
where $\Fund^{\sigma_2}_{\sigma_1}(\Sigma)$ is the set of cycles relative to $\partial \Sigma$ representing the fundamental class of $\Sigma$ with boundary $\sigma_2-\sigma_1$. The equivalence relation is defined by $\zeta' \sim_\varphi \langle \varphi^* \omega   ,\Lambda  \rangle \, \zeta $ for arbitrary $\Lambda\in C_n(\Sigma)$ with boundary $\zeta-\zeta'$. This equivalence relation is different from \eqref{defmorphismon0cells2}.  
With this notation we define $E_\omega (\Sigma,\varphi)$ via the coend 
\begin{align}
E_\omega (\Sigma,\varphi) (\sigma_1) = \int^{\sigma_2\in \Fund(S_2)} \, \Sigma^\varphi(\sigma_2,\sigma_1)*\sigma_2 
\end{align}
and linear extensions. 
These functors are compatible with the composition of cobordisms only up to natural isomorphisms. For 1-morphisms $(\Sigma_a,\varphi_a) \colon (S_1,\xi_1)\longrightarrow (S_2,\xi_2)$ and $(\Sigma_b,\varphi_b) \colon (S_2,\xi_2)\longrightarrow (S_3,\xi_3)$ the natural isomorphism is induced by the linear map
\begin{align}
 \Sigma_b^{\varphi_b}(\sigma_3,\sigma_2) \otimes_\C \Sigma_a^{\varphi_a}(\sigma_2,\sigma_1) &\longrightarrow \left( \Sigma_b\circ\Sigma_a \right)^{\varphi_b\cup \varphi_a}(\sigma_3,\sigma_1) \\ 
 \sigma_{\Sigma_b} \otimes_\C \sigma_{\Sigma_a} &\longmapsto \sigma_{\Sigma_b}+\sigma_{\Sigma_a} \ .
\end{align}
In~\cite{EHQFT} it is shown that these coherence morphisms are closely related to the transgression of $\omega$.
They encode interesting physical properties as we will see in Section \ref{Sec: Boundary theories}. 

There are further coherence isomorphisms
\begin{align}\label{eqnnatisocohuni}
\Phiit_{(S,\xi)}\colon \id_{E_\omega(S,\xi)} \Longrightarrow E_\omega(\id_{(S,\xi)}) 
\end{align}
encoding the compatibility with identities for all objects $(S,\xi)$ of $\EDCob$. They are defined as follows.
Using the enriched co-Yoneda Lemma we can write the identity as the coend
\begin{align}
\id_{E_\omega(S,\xi)}(\,\cdot\,) \cong \int^{\sigma \in {E_\omega(S,\xi)}} \, \Hom_{E_\omega(S,\xi)} (\sigma , \,\cdot\,)* \sigma \ .
\end{align} 
Without loss of generality we can evaluate this at a representative $\sigma_0$ of the fundamental class of $S$:
\begin{align}
\sigma_0 \cong \int^{\sigma \in {E_\omega(S,\xi)}} \, \Hom _{E_\omega(S,\xi)} ( \sigma , \sigma_0)* \sigma \cong \int^{\sigma \in {\Fund(S)}} \, \Hom _{E_\omega(S,\xi)}( \sigma , \sigma_0)* \sigma \ .\label{cohid1eqn}
\end{align}  
On the other hand we have
\begin{align}
E_\omega(\id_{(S,\xi)})(\sigma_0) = \int ^{\sigma \in {\Fund(S)}} \, (S\times [0,1])^{ \xi \times \id_{[0,1]}} (\sigma , \sigma_0)* \sigma \ . \label{cohid2eqn}
\end{align}
There is a natural isomorphism
\begin{align}
(S\times [0,1])^{ \xi \times \id_{[0,1]}} (\sigma , \sigma_0) & \longrightarrow \Hom _{E_\omega(S,\xi)}( \sigma , \sigma_0) \\
 \mu & \longmapsto - {p_S}_* \mu\label{Eq: nat iso coh uni on integrands} 
\end{align}
defined by the projection $p_S:S \times [0,1]\longrightarrow S$, which induces the natural transformation~\eqref{eqnnatisocohuni}.

To a 2-morphism $(M,\psi)\colon (\Sigmait_a,\varphi_a) \Longrightarrow (\Sigma_b,\varphi_b)$ between 1-morphisms $(\Sigma_a,\varphi_a)$ and $(\Sigma_b,\varphi_b) $ from $(S_1,\xi_1)$ to $(S_2,\xi_2)$ in $\EDCob$, we assign the natural transformation
\begin{align}
E_\omega(M,\psi) \colon E_\omega(\Sigmait_a,\varphi_a) \Longrightarrow E_\omega(\Sigmait_b,\varphi_b) 
\end{align}
 between the functors $E_\omega(\Sigmait_a,\varphi_a), E_\omega(\Sigmait_b,\varphi_b) : E_\omega(S_1,\xi_1) \longrightarrow E_\omega(S_2,\xi_2)$ consisting of the natural maps
\begin{align}
E_\omega(\Sigmait_a,\varphi_a) (\sigma_1) \longrightarrow E_\omega(\Sigmait_b,\varphi_b) (\sigma_1)
\end{align} 
for $\sigma_1 \in \Fund(S_1)$, which are the maps between the respective coends induced by the linear maps
\begin{align}
E_\omega(M,\psi)_{\sigma_2,\sigma_1} : \Sigma_a^{\varphi_a}(\sigma_2,\sigma_1)  \longrightarrow \Sigma_b^{\varphi_b}(\sigma_2,\sigma_1) \label{Eq: Map between coends induce}  
\end{align} 
defined as follows: For $[\mu_a] \in \Fund_{\sigma_1}^{\sigma_2}(\Sigmait_a)$ and $[\mu_b] \in \Fund_{\sigma_1}^{\sigma_2}(\Sigmait_b)$ we can find a fundamental cycle 
$\nu$ of $M$ which is compatible with these fundamental cycles on the boundary~\cite[eq.~(3.11)]{EHQFT}:
\begin{align}
\partial\nu = \mu_b-\mu_a+(-1)^{n-2} \, \big(\sigma_1\times[0,1] - \sigma_2 \times[0,1]\big) \ .
\end{align}
Mapping $[\mu_a]$ to $\langle \psi^* \omega,\nu\rangle \, [\mu_b]$ yields a well-defined linear map
$\mathbb{C}[\Fund_{\sigma_1}^{\sigma_2}(\Sigmait_a)] \longrightarrow \Sigma_b^{\varphi_b}(\sigma_2,\sigma_1)$, which descends to $\Sigma_a^{\varphi_a}(\sigma_2,\sigma_1)$ and induces the map~\eqref{Eq: Map between coends induce}. 

The construction sketched above fits into an extended field theory~\cite[Theorem 3.19]{EHQFT}, the classical Dijkgraaf-Witten theory with action $\omega$. Restricting $E_\omega$ to the endomorphism category of the empty set reproduces the ordinary field theory 
\begin{align}
L_\omega \colon \End_{\EDCob}(\varnothing)= \DCob \longrightarrow \fvs = \End_\Tvs(\fvs)
\end{align}
mentioned above. The functor
$L_\omega$ admits the following concrete description:
\begin{itemize}
\item 
To a closed $n{-}1$-dimensional manifold $\Sigma$ equipped with a map $\varphi : \Sigma \longrightarrow BD$ it assigns the vector space $L_\omega(\Sigma, \varphi)=\Sigma^\varphi(\varnothing,\varnothing)= \C[\Fund(\Sigma)]/{\sim_\varphi}$.
\item 
To a morphism $(M,\psi): (\Sigma_a,\varphi_a) \longrightarrow (\Sigma_b,\varphi_b)$ it assigns the linear map 
\begin{align}
L_\omega(M,\psi): L_\omega(\Sigma_a,\varphi_a) & \longrightarrow L_\omega (\Sigma_b,\varphi_b) \\
[\sigma_{a}]& \longmapsto \langle \psi^*\omega, \sigma_M \rangle \, [\sigma_{b}] \ ,
\end{align}
with $\sigma_M \in \Fund^{\sigma_b}_{\sigma_a}(M)$.
\end{itemize}

\subsection{Quantum Dijkgraaf-Witten theories}\label{sec:quantumDWtheories}

Let $\omega$ be an $n$-cochain on $BD$ with values in $U(1)$.
We define the quantum Dijkgraaf-Witten theory $Z_\omega$ with topological action $\omega$ by performing a `path integral' of the classical field theory $L_\omega$ over the space of field configurations $\BunD(M)$.
One of the problems with finding a rigorous definition of general quantum field theories is the lack of a measure on the space of field configurations. 
In the case of Dijkgraaf-Witten theories the space of field
configurations is the essentially finite groupoid\footnote{A groupoid
  is {essentially finite} if all morphism sets are finite and
  there exists an equivalence to a groupoid with a finite number of
  objects. The field configurations form a groupoid, since $\BunD$ is a stack.} $\BunD(M)$ of $D$-bundles on $M$ and a well-defined integration theory exists. 
The groupoid cardinality~\cite{groupoidfication} induces a natural
counting measure on $\BunD(M)$. The {integral} of a gauge-invariant function $f\colon \Obj(\BunD(M))\longrightarrow \C$ over $\BunD(M)$ is
\begin{align}
\int_{\BunD(M)} \, f \coloneqq \sum_{\varphi \in \pi_0(\BunD(M))} \, \frac{f(\varphi)}{\big|\Aut_{\BunD(M)}(\varphi)\big|} \ ,
\end{align}  
where $ \pi_0(\BunD(M))$ is the set of isomorphism classes of
principal $D$-bundles on $M$, each of which is weighted by the inverse
of the order of its automorphism group. 

The existence of a well-defined path integral makes quantization in principal straightforward. A slight subtlety arises, since the action on a manifold with boundary is not a complex number, but rather an element of a complex line as we have seen in Section \ref{Sec: Calssical gauge theory}.
For this reason, we have to assign to an $n{-}1$-dimensional manifold $\Sigma$ the vector space \cite{FreedQuinn}
\begin{align}
\label{Definition on Objects}
Z_\omega (\Sigma) = \lim_{\BunD(\Sigma)} \, L_\omega = \int_{\BunD(\Sigma)} \, L_\omega \ ,
\end{align}  
writing the limit as an end in the second step. The functor $L_\omega$
can be evaluated on the groupoid $\BunD(\Sigma)$ by considering a gauge
transformation $h$ as a homotopy $h\colon [0,1]\times \Sigma
\longrightarrow BD$ between classifying maps, i.e. as a morphism in $\DCob$. The correspondence between homotopies of classifying maps and gauge transformations only holds for discrete gauge groups. 
A useful realisation of the limit is given by the vector space of
parallel sections of $L_\omega(\Sigma, \,\cdot\,)$. A {parallel
  section} $f$ of $L_\omega(\Sigma, \,\cdot\,)$ consists of an element
$f(\varphi)\in L_{\omega}(\Sigma,\varphi)$ for all $\varphi \in
\BunD(\Sigma)$ such that $L_\omega ([0,1]\times \Sigma, h)
\big(f(\varphi) \big)= f(\varphi')$ for all gauge transformations
$h\colon \varphi \longrightarrow \varphi'$. The space of parallel
sections can be regarded as the space of gauge-invariant functions on
the set of classical gauge field configurations.
For this reason the definition can be interpreted as an implementation of the Gauss Law in quantum gauge theory, which requires that physical states must be gauge-invariant.

Now consider a cobordism $M\colon \Sigmait_1 \longrightarrow \Sigmait_2$. To define $Z_\omega(M)$ we introduce for $\varphi_2 \in \BunD(\Sigma_2)$ the {homotopy fibre} $\BunD(M)|_{\varphi_2}$  as the homotopy pullback
\begin{equation}
\begin{tikzcd}
\BunD(M)|_{\varphi_2} \ar[r] \ar[d] & \BunD(M) \ar[d, "{\pr_{\Sigmait_2}}"] \\
* \ar[r, swap,"\varphi_2"] & \BunD(\Sigmait_2) 
\end{tikzcd} 
\end{equation} 
where $\pr_{\Sigmait_2}$ is the pullback functor induced by the inclusion
$\Sigmait_2\hookrightarrow M$. The groupoid
$\BunD(M)|_{\varphi_2}=\pr_{\Sigma_2}^{-1}[\varphi_2]$ of gauge fields on $M$ restricting, up to gauge transformations, to $\varphi_2$ on $\Sigma_2$ can be concretely described as follows:
Objects are pairs $(\psi,h)$ where $\psi$ is a bundle over $M$ and
$h\colon \psi|_{\Sigmait_2}\longrightarrow \varphi_2$ is a gauge
transformation. Morphisms $(\psi,h)\longrightarrow (\psi',h')$ are
gauge transformations $\psi \longrightarrow \psi'$ such that the
diagram 
\begin{equation}
\begin{tikzcd}
\psi|_{\Sigmait_2} \ar[rr]\ar[dr, "h", swap] & & \psi'|_{\Sigmait_2} \ar[dl,"h'"] \\
 & \varphi_2 & 
\end{tikzcd}
\end{equation} 
commutes. We fix representatives $\sigma_1$ and $\sigma_2$ of the fundamental classes of $\Sigmait_1$ and $\Sigmait_2$, respectively. This allows us to express the value of a parallel section $f\in Z_\omega(\Sigma_1)$ on a principal $D$-bundle $\varphi_1\in \BunD(\Sigma_1)$ as $f(\varphi_1)=\mathbf{f}(\varphi_1)\, [\sigma_1]$ with $\mathbf{f}(\varphi_1) \in \C$. Then we can define 
\begin{align}
\label{Definition on Morphisms}
Z_\omega (M)(f)(\varphi_2)= \Big( \int_{(\psi,h)\in
  \BunD(M)|_{\varphi_2}} \, \langle  h^* \omega, [0,1] \times
  \sigma_{2} \rangle \, \langle \psi^* \omega , \sigma_M  \rangle \,
  \mathbf{f}(\psi|_{\Sigmait_1}) \Big) \ [\sigma_2] \ ,
\end{align}  
where $\sigma_M$ is a representative for the fundamental class of $M$ satisfying $\partial \sigma_M = \sigma_2-\sigma_1$ and we consider the gauge transformation $h$ as a homotopy $h \colon [0,1]\times \Sigma_2 \longrightarrow BD$. This definition is independent of all choices involved.

The assignments \eqref{Definition on Objects} and \eqref{Definition on Morphisms} define a functorial field theory~\cite{FreedQuinn}
\begin{align}
Z_\omega \colon \Cob_n \longrightarrow \fvs \ .
\end{align}

\begin{remark}
A more systematic way of constructing $Z_\omega$ is given by applying
the orbifold construction of~\cite{OFK} to the classical field theory
$L_\omega$. In \cite{EHQFT} the extended orbifold construction
of~\cite{EOFK} is used to construct Dijkgraaf-Witten theories as
extended field theories. Three-dimensional extended Dijkgraaf-Witten
theories are also constructed in \cite{Morton}.  
\end{remark}

We conclude this section with a few examples, which we will consider throughout this paper.

\begin{example}\label{Ex: 2D coycles}
In two dimensions, for any finite group $D$ and a 2-cocycle $\omega\in
Z^2(BD;U(1))$, the vector space $Z_\omega(\mathbb{S}^1)$ associated to
the circle $\mathbb{S}^1$ is the space of $\omega$-twisted
characters on $D$, because the groupoid of principal $D$-bundles on
$\mathbb{S}^1$ is equivalent to the action groupoid $D\Ds D$ for the
adjoint action of $D$ on itself, and hence a parallel section is just a map $f:D\longrightarrow \C$ satisfying $f(d\,d'\,d^{-1}) = \omega(d,d')\, f(d')$ for all $d,d'\in D$. The invariant
\begin{align}
Z_\omega(\mathbb{T}^2) = \frac1{|D|} \ \sum_{\stackrel{\scriptstyle (d,d')\in D\times D}{\scriptstyle d\,d'=d'\,d}} \, \frac{\omega(d,d')}{\omega(d',d)}
\end{align}
associated to the two-dimensional torus $\mathbb{T}^2$ is the number of irreducible $\omega$-twisted representations of $D$.
We give a few concrete examples of 2-cocycles:
\begin{itemize}
\item[(a)] The group cohomology $H^2(\Z_N\times \Z_N;U(1))$ is $\Z_N$.
If we write the cyclic group $\Z_N$ additively then the non-trivial 2-cocycle corresponding to 
$k\in \{0,1,\dots , N-1\}$ is 
\begin{align}\label{EQ:Def 2 cocycle}
\omega_k\big((a_1,b_1)\,,\,(a_2,b_2)\big)= \exp\Big(\frac{2\pi \,\iu\, k}{N}\,a_1\, b_2 \Big)
\end{align}  
with $(a_1,b_1),(a_2,b_2)\in \Z_N\times \Z_N$.
For $N=2$, the partition function $Z_{\omega_1}(\mathbb{T}^2)$ on $\mathbb{T}^2$ for the non-trivial $\Z_2\times \Z_2$ cocycle is $1$ corresponding to the fact that there exists only one $\omega_1$-twisted irreducible representation of $\Z_2\times \Z_2$~\cite{TwistedDWandGerbs}. 

\item[(b)] The degree~2 group cohomology of the dihedral group $D_8=\langle a,b \mid a^4=b^2=1 \
  , \ b\,a\,b^{-1}=a^{-1}\rangle$ with values in $U(1)$ is $\Z_2$. The non-trivial 2-cocycle is given by~\cite[Section 3.7]{ProjectiveRep}
\begin{align}
\omega\big(a^i\, b^j, a^{i'}\,b^{j'}\big)=
\begin{cases}
1 &, \quad \mbox{$j=0$} \ , \\
\exp\big( \frac{2\pi \,\iu\, }{4} \, i' \big) &, \quad \mbox{$j=1$} \ .
\end{cases}
\end{align} 
\end{itemize}
\end{example}

\begin{example} 
In three dimensions, for any finite group $D$ and a 3-cocycle $\omega\in
Z^3(BD;U(1))$, the invariant $Z_\omega(\mathbb{T}^3)$ associated to
the three-dimensional torus $\mathbb{T}^3$ is the number of
irreducible representations of the $\omega$-twisted Drinfeld double of
the group algebra $\C[D]$. 
The cohomology group $H^3(\Z_N;U(1))$ is $\Z_N$. The 3-cocycles have the
concrete form~\cite[Proposition~2.3]{Huang2014}  
\begin{align}\label{Eq: Zn 3-cocycle}
\omega_k(a,b,c)= \exp\Big(\frac{2\pi \,\iu\, k}{N} \, a \, \Big\lfloor\frac{b+c}{N}\Big\rfloor \Big) 
\end{align}  
for $a,b,c,k\in \Z_N=\lbrace 0,1,\dots, N-1 \rbrace$, where $\lfloor
r\rfloor$ denotes the integer part of the real number $r\in \R$,
i.e.~the largest integer less than or equal to $r$. These theories are
studied in~\cite{Kapustin:2014gua}. They have been extended to a product
of an arbitrary number of cyclic groups $\Z_{N_i}$ (i.e. a generic
finite abelian group) in~\cite{Chen:2011pg,Wang:2014tia}.
\end{example} 

\section{Discrete symmetries and 't Hooft anomalies}\label{Sec:Finite
  symmetries and 't Hooft anomalies}

In this section we study the action of a finite symmetry group $G$ on a classical Dijkgraaf-Witten theory $L_\omega\colon \DCob \longrightarrow \fvs$ with gauge group $D$ and topological action $\omega\in Z^n(BD;U(1))$. General symmetries of abelian quantum Dijkgraaf-Witten theories are discussed in \cite{BrauerGroup}. 
We only consider symmetries arising from an action of $G$ on $D$-gauge
fields which preserve $\omega \in Z^n(BD; U(1))$. Following~\cite{BrauerGroup} we call these {kinematical} symmetries. 
We show that they extend to the quantum theory and study their gauging. 
Gauging these symmetries is not always possible and the obstructions are
encoded in a spectral sequence.

\subsection{Discrete symmetries of Dijkgraaf-Witten theories}\label{Sec: Finite symmetries}

For a symmetry to be compatible with cutting and gluing of manifolds,
we describe it as an endofunctor of $\DCob$ acting by pullback along
the inverse on a field theory. There is a natural way to construct
endofunctors of $\DCob$ from homeomorphisms of $BD$ which is described by a 2-functor 
\begin{align}\label{End(BG)=>End(D-Cob)}
 \mathcal{R} \colon * \DS \Piit_1 [BD,BD] \longrightarrow * \DS  \End (\DCob) \ , 
\end{align}
where $\Piit_1[BD,BD]$ is the category with continuous maps $BD\longrightarrow BD$ as objects and equivalence classes of homotopies as morphisms. 
Concretely, $\mathcal{R}$ sends a continuous map $\chi \colon BD \longrightarrow BD$ to the endofunctor
\begin{align}
\mathcal{R}(\chi) \colon \DCob &\longrightarrow \DCob \\
(\Sigma,\varphi \colon \Sigma \longrightarrow BD) & \longmapsto (\Sigma ,\chi \circ \varphi\colon \Sigma \longrightarrow BD) \\
\big( (M,\psi) \colon (\Sigma_1,\varphi_1) \longrightarrow 
  (\Sigma_2,\varphi_2) \big) & \longmapsto \big( (M,\chi \circ\psi) \colon (\Sigma_1,\chi \circ\varphi_1)\longrightarrow (\Sigma_2,\chi \circ \varphi_2) \big)
\end{align}
and a homotopy $h \colon \chi_1 \longrightarrow \chi_2$ to the natural transformation $\mathcal{R}(h)\colon \mathcal{R}(\chi_1)\Longrightarrow \mathcal{R}(\chi_2)$ with components
\begin{align}
\mathcal{R}(h)_{(\Sigma,\varphi)} = \left( [0,1] \times \Sigma , h \circ \id_\varphi \right)
\end{align}
where $h \circ \id_\varphi$ denotes the horizontal composition of homotopies. The naturality of $\mathcal{R}(h)$ follows from the proof of~\cite[Proposition 2.9]{OFK}.
By the bicategorical Yoneda Lemma, automorphisms of $BD$ correspond to automorphisms of the stack of principal $D$-bundles (which is represented by $BD$). Hence a symmetry corresponding to a homeomorphism of $BD$ acts on the space of field configurations. 
 
To define a symmetry, the group $G$ only has to act up to gauge transformations. For finite groups, gauge transformations and homotopies between classifying maps are in one-to-one correspondence. For this reason we expect a symmetry for every action of $G$ on $BD$ up to a `homotopy' which preserves $\omega$. Since $BD$ is a homotopy 1-type we can work with the following concrete description.

\begin{definition}\label{Def: homotopy Coherent action}
An {action of $G$ on $BD$ up to (coherent) homotopy} is a 2-functor 
\begin{align}
\alpha \colon \underline{* \DS G} \longrightarrow * \DS \Piit_1[BD,BD] \ ,
\end{align}
where $\underline{* \DS G}$ is the 2-category with one object, the group $G$ as 1-morphisms and only identity 2-morphisms. 
\end{definition}

\begin{remark}
To unpack this compact definition note that the category   
$\Piit_1[BD,BD]$ is equivalent to the action groupoid
\begin{align}
[\pi_1(BD),\pi_1(BD)]\DS D = \End_{\Grp}(D)\DS D \ ,
\end{align}
where the action of $D$ on a group homomorphism is by conjugation. Every action of $G$ up to homotopy takes values in the full subgroupoid $\Aut_{\Grp}(D)\DS D$ of automorphisms of $D$.
An arbitrary 2-functor $\underline{* \DS G} \longrightarrow *\DS (\Aut_{\Grp}(D)\DS D)$ is called a {non-abelian group cocycle} \cite{nCat=Cohomology,Blanco2005}. Hence homotopy coherent actions on $BD$ are classified by non-abelian group cocycles. 
Non-abelian cocycles also appear in the construction of equivariant Dijkgraaf-Witten theories~\cite{NMS} under the name weak 2-cocycles.
We discuss non-abelian group cohomology in more detail in Section~\ref{Sec: Non-abelian group cohomology}.  

If $D$ is abelian there are no morphisms between different objects in $\Aut_{\Grp}(D)\DS D$. This implies that an action up to homotopy of $G$ on $BD$ is given by a proper action of $G$ on $D$ and a group 2-cocycle in $H^2(BG;D)$ describing the coherence isomorphisms of the corresponding 2-functor. This agrees with the physical description in \cite{Kapustin:Symmetries}. 
\end{remark}
For every action $\alpha \colon \underline{* \DS G} \longrightarrow *\DS \Piit_1[BD,BD]$ up to homotopy the 2-functor \eqref{End(BG)=>End(D-Cob)} induces via pullbacks a 2-functor
\begin{align}\label{EQ: Action of kinematic symmetries on TFTs}
\rho \colon \underline{* \DS G}&\longrightarrow \DTFT\DS \End_{\Cat}(\DTFT)\hookrightarrow \Cat \\
g &\longmapsto \mathcal{R}\big(\alpha(g^{-1})\big)^* \ ,
\end{align}  
where we denote by $\Cat$ the 2-category of categories.

The (exponentiated) action of a gauge theory can be considered as a gauge-invariant map from the space of field configurations on an $n$-dimensional manifold $M$, in our case $\BunD(M)$, to $U(1)$. An action of $G$ on the space of field configurations induces an action via pullbacks on the set of gauge-invariant functions from the space of field configurations to $U(1)$.
A theory admits the symmetry $G$ if its (exponentiated) action is invariant under this action, i.e. it is a fixed point. 
By categorification we arrive at the following description.

\begin{definition}\label{Def: Field theory with k. symmetry}  
A $D$-equivariant field theory with {kinematical symmetry} described by  
\[
\rho\colon \underline{* \DS G}\longrightarrow \DTFT \DS \End_{\Cat}(\DTFT) \ ,
\]
as in~\eqref{EQ: Action of kinematic symmetries on TFTs}, is a homotopy fixed point of $\rho$, i.e. a natural 2-transformation $Z\colon 1 \Longrightarrow \rho$, 
where $1$ is the unique 2-functor sending $*$ to the category with one object and only identity morphisms. 
\end{definition}
\begin{remark}
Unpacking the definition, a $D$-equivariant field theory with kinematical symmetry consists of
\begin{itemize}
\item[(a)]
A functor $Z \colon 1 \longrightarrow \DTFT$; and

\item[(b)]
Natural transformations $\Upsilon_g \colon \rho(g) [Z] \Longrightarrow Z$ for all $g\in G$;
\end{itemize} 
satisfying natural coherence conditions. Since $1$ represents the identity 2-functor on $\Cat$ this is the same as a field theory $Z \in \DTFT$, together with coherent natural symmetric monoidal transformations $\Upsilon_g \colon  \mathcal{R}(\alpha(g^{-1}))^* Z \Longrightarrow Z$ for $g\in G$.
\end{remark}

An arbitrary Dijkgraaf-Witten theory with topological action $\omega\in Z^n(BD;U(1))$ does not admit a kinematical symmetry in general. On the other hand, there may be different ways to equip a given field theory with the structure of a homotopy fixed point. We give a sufficient condition for a kinematical symmetry to exist.
For this we need to introduce the following notion.

\begin{definition}\label{Def: Preserved}
A $n$-cocycle $\omega\in Z^n(BD;U(1))$ is {preserved by the action $\alpha$} if it can be equipped with the structure of a homotopy fixed point for the induced action of $G$ via the pullback along $\alpha(g^{-1}) $ on the category $Z^n(BD;U(1))$ whose morphisms are $n{-}1$-cochains up to coboundaries.  
\end{definition} 
In general there are non-isomorphic choices for the fixed point structure. A necessary condition for such a fixed point to exist is $\alpha(g)^*[\omega]=[\omega]$ for all $g\in G$. 

\begin{remark}\label{Rem: Preserved}
Concretely, the additional structure consists of an equivalence class of cochains $\Phi_{g}\in C^{n-1}(BD;U(1))$ up to coboundary satisfying\footnote{Throughout we switch freely between the additive and multiplicative notation for $U(1)$-valued cocycles.} $\delta \Phi_{g} = \omega- \alpha(g^{-1})^* \omega$. These cochains have to satisfy the coherence relations
\begin{align}
\Phi_{g_1}+\alpha(g_1^{-1})^*\Phi_{g_2} = \Phi_{g_1\,g_2}+\sigma_{g_1,g_2}[\omega]  \ , 
\end{align}
where $\sigma_{g_1,g_2}[\omega]$ is the $n{-}1$-cochain induced by the homotopy $\sigma_{g_1,g_2}\colon \alpha(g_2^{-1})\circ \alpha(g_1^{-1}) \longrightarrow \alpha(g_2^{-1}\,g_1^{-1})$. The difference between two homotopy fixed point structures can be described by a group homomorphism $G \longrightarrow H^{n-1}(BD;U(1))$.  
\end{remark}

\begin{proposition}\label{Prop: Classical symmetry}
Let $\omega \in Z^n(BD;U(1))$ be a topological action and $\alpha\colon \underline{* \DS G} \longrightarrow * \DS \Piit_1[BD,BD]$ a homotopy coherent action of $G$ on $BD$. If $\alpha$ preserves $\omega$, then the classical Dijkgraaf-Witten theory $L_\omega \colon \DCob \longrightarrow \fvs$ admits a kinematical symmetry described by $\alpha$.
\end{proposition}
\begin{proof}
We set $Z=L_{\omega}$ and define a family of natural transformations $\Upsilon_g \colon  L_{\alpha(g^{-1})^* \omega} \Longrightarrow L_{\omega}$ by
\begin{align*}
{\Upsilon_g}_{\,(\Sigma, \varphi)}\colon L_{\alpha(g^{-1})^* \omega}(\Sigma, \varphi)&\longrightarrow L_{ \omega}(\Sigma, \varphi) \\
 [\sigma_\Sigma] &\longmapsto \left\langle \varphi^* \Phi_g, \sigma_\Sigma \right\rangle \, [\sigma_\Sigma] \ ,
\end{align*}  
where $\Phi_g$ is the $n{-}1$-cochain of Remark~\ref{Rem: Preserved} satisfying $\delta \Phi_{g} = \omega-\alpha(g^{-1})^* \omega $. For a morphism $(M,\psi) \colon (\Sigma, \varphi) \longrightarrow (\Sigma', \varphi')$ we have to show that the diagram
\begin{equation}
\begin{tikzcd}
L_{\alpha(g^{-1})^* \omega}(\Sigma, \varphi) \ar[rr,"{\Upsilon_g}_{\,(\Sigma, \varphi)}"]\ar[dd,"{L_{\alpha(g^{-1})^* \omega}(M,\psi)}",swap] & & L_{ \omega}(\Sigma, \varphi) \ar[dd,"{L_{\omega}(M,\psi)}"]\\
 & & \\
L_{\alpha(g^{-1})^* \omega}(\Sigma', \varphi') \ar[rr,swap,"{\Upsilon_g}_{\,(\Sigma', \varphi')}"] & & L_{ \omega}(\Sigma', \varphi')
\end{tikzcd}
\end{equation}
commutes.
We fix $\sigma_M \in \Fund(M)_{\sigma_\Sigma}^{\sigma_{\Sigma'}}$ and compute 
\begin{align*}
\left\langle  \psi^* \big(\omega-\alpha(g^{-1})^* \omega \big), \sigma_M  \right\rangle =
\left\langle  \psi^* (\delta \Phi_g ), \sigma_M  \right\rangle = \left\langle  \psi^*  \Phi_g, \partial \sigma_M  \right\rangle = \left\langle   \psi^*  \Phi_g ,\sigma_{\Sigma'} \right\rangle - \left\langle  \psi^*  \Phi_g , \sigma_{\Sigma} \right\rangle \ .
\end{align*}
This shows that the diagram commutes. 
The coherence conditions follow from the fact that the collection $\Phi_g$ corresponds to a homotopy fixed point structure.
\end{proof}

We describe the symmetries of quantum Dijkgraaf-Witten theory following~\cite[Section~2.4]{Freed:2014eja} by the following notion.
\begin{definition}\label{Def: Internal symmetry}
Let $G$ be a finite group and denote by $G\text{-}\mathsf{Rep}$ the category of finite-dimensional $G$-representations.
Let $Z\colon \Cob_n \longrightarrow \fvs$ be a topological field theory. An {internal $G$-symmetry} of $Z$ is a lift 
\begin{equation}
\begin{tikzcd}
 & G\text{-}\mathsf{Rep} \ar[dr] & \\
\Cob_n \ar[rr,swap,"Z"] \ar[ur,"Z_G"] & & \fvs
\end{tikzcd}
\end{equation}
of $Z$, where $G\text{-}\mathsf{Rep} \longrightarrow \fvs$ is the
forgetful functor. 
\end{definition}

\begin{remark}
This definition is equivalent to fixing a group homomorphism $G\longrightarrow \Aut_\otimes(Z)$ to the group of symmetric monoidal natural automorphisms of $Z$. 
\end{remark}

Kinematical symmetries of classical Dijkgraaf-Witten theories extend to the quantum theory: 
For a fixed manifold $\Sigma$, $\Upsilon_g$ induces a natural isomorphism $L_\omega \circ \mathcal{R}(\alpha(g^{-1}))|_{\BunD(\Sigma)}\longrightarrow L_\omega|_{\BunD(\Sigma)}$, which induces a linear map 
\[
 \int_{\BunD(\Sigma)} \, L_\omega \circ
 \mathcal{R}\big(\alpha(g^{-1})\big) \longrightarrow
 \int_{\BunD(\Sigma)} \, L_\omega = Z_\omega (\Sigma)\ .
\] 
The equivalence $\mathcal{R}(\alpha(g^{-1}))$ induces a morphism 
\begin{align}
Z_\omega(\Sigma) = \int_{\BunD(\Sigma)} \, L_\omega \longrightarrow
  \int_{\BunD(\Sigma)} \, L_\omega \circ \mathcal{R}\big(\alpha(g^{-1})\big) \ .
\end{align}
The action of $G$ consisting of $\varphi_g \colon Z_\omega(\Sigma) \longrightarrow Z_\omega(\Sigma) $ is defined as the composition of these two maps. 
Fixing a fundamental class $\sigma_\Sigma$ of $\Sigma$ as in \eqref{Definition on Morphisms} and a parallel section $f(\,\cdot\,)= \mathbf{f}(\,\cdot\,)\,[\sigma_\Sigma] \in Z_\omega(\Sigma) $, the action of $G$ on $Z_\omega(\Sigma)$ takes the concrete form  
\begin{align}\label{Eq: Concrete form of the quantum symmetry}
g\triangleright \mathbf{f}(\Sigma,
  \varphi)=\mathbf{f}\big(\mathcal{R}(\alpha(g^{-1}))[\Sigma,
  \varphi]\big) \, \langle \varphi^* \Phi_{g} , \sigma_\Sigma \rangle
\end{align} 
with $\delta \Phi_{g}= \omega-\alpha(g^{-1})^* \omega$ as in the proof of Proposition~\ref{Prop: Classical symmetry}. 
\begin{proposition}\label{Prop: symmetry extends to quantum}
The collection $\varphi_g$ define a representation of $G$ on $Z_\omega (\Sigma)$ such that $Z_\omega$ is a functor into the category $G\text{-}\mathsf{Rep}$ of finite-dimensional $G$-representations. 
\end{proposition}
\begin{proof}
This is a direct consequence of the functoriality of the orbifold construction \cite[Remark~3.43]{OFK} and the coherence conditions for the homotopy fixed point. 
\end{proof}

\begin{example}
The trivial action of $G$ on $BD$ is always an internal $G$-symmetry. 
Any action of $G$ is an internal $G$-symmetry for a theory with trivial topological Lagrangian. 
We will provide some more profound examples in Sections~\ref{Sec: Non-abelian group cohomology} and~\ref{Sec: Gauging}.
\end{example}

\subsection{Non-abelian group cohomology}\label{Sec: Non-abelian group
  cohomology}

Following \cite{Blanco2005} we review non-abelian group 2-cocycles and show how they classify extensions. For simplicity we only discuss groups, which is enough for the study of anomalies in Dijkgraaf-Witten theories. The generalisation to groupoids is straightforward. Let $G$ and $D$ be finite groups. Recall from Section~\ref{Sec: Finite symmetries} that a non-abelian 2-cocycle on $G$ with coefficients in $D$ is a 2-functor $\alpha \colon \underline{* \DS G} \longrightarrow * \DS (\Aut_{\Grp}(D)\DS D) \subset \Grpd$, where $\Grpd$ is the 2-category of groupoids. 
The 2-category $* \DS (\Aut_{\Grp}(D)\DS D)$ can be considered as a sub-2-category of $\Grpd$ by sending the only object to $* \DS D$. We assume without loss of generality that $\alpha$ preserves identities strictly.  
Spelling out the definition, we see that $\alpha$ consists of  maps of sets $\alpha \colon G \longrightarrow \Aut_{\Grp}(D)$ and $\sigma_\alpha \colon G\times G \longrightarrow D$ satisfying
\begin{align}
\alpha(1)&=\id_D \ , \\[4pt]
\sigma_\alpha (1,1)&=1 \ , \\[4pt]
\alpha(g_1\,g_2)[d] &= \sigma_\alpha(g_1,g_2)^{-1} \, \alpha (g_1)
                      \big[\alpha(g_2) [d]\big] \,
                      \sigma_\alpha(g_1,g_2) \ , \\[4pt]
\sigma_\alpha(g_1,g_2) \, \sigma_\alpha(g_1\, g_2,g_3) & = \alpha
                                                         (g_1)
                                                         \big[\sigma_\alpha(g_2,g_3)\big]
                                                         \,
                                                         \sigma_\alpha(g_1,
                                                         g_2 \, g_3) \ . \label{Eq: Non abelian cocycle} 
\end{align}

Using the Grothendieck construction, the 2-functor $\alpha \colon \underline{BG}\longrightarrow \Grpd$ induces an op-fibration of groupoids
\begin{align}
\int \alpha \longrightarrow *\DS G \ 
\end{align}  
having the following concrete description: There is only one object which we denote by $*$, endomorphisms are given by pairs $(g,d)\in G\times D$ and composition is defined by 
\begin{align}
(g_2,d_2)\,(g_1,d_1)= \big(g_2\,g_1\,,\, d_2 \, \alpha (g_2)[d_1] \, \sigma_\alpha(g_2,g_1)\big) \ . 
\end{align} 
The op-fibration corresponds to an extension of $G$ by $D$, i.e. an exact sequence 
\begin{align} \label{Eq: Extension}
1\longrightarrow D \longrightarrow \widehat G \longrightarrow G
  \longrightarrow 1 \ ,
\end{align} 
with $\widehat{G}=\End_{\int \alpha}(*)$. 
Since the Grothendieck construction induces an equivalence of
categories, it is natural to believe that extensions of $G$ by $D$ are
classified by non-abelian 2-cocycles on $G$ with coefficients in
$D$. A proof of this expectation can be found in \cite{Blanco2005}. 

Let us conclude this discussion with a few explicit examples.
\begin{example}
For every pair of groups $(G,D)$ there is a trivial non-abelian 2-cocycle corresponding to the constant 2-functor $\underline{* \DS G}\longrightarrow * \DS (\Aut_{\Grp}( D) \DS D)$. The corresponding extension is
\begin{align*}
1\longrightarrow D \longrightarrow D\times G \longrightarrow G \longrightarrow 1 \ .
\end{align*}
\end{example}
\begin{example}\label{Ex: Abelian cohomology}
If $D$ is abelian and $\alpha \colon G\longrightarrow \Aut_{\Grp}(D)$ is trivial, then a non-abelian 2-cocycle reduces to an ordinary 2-cocycle $\sigma\in H^2(G;D)$ and the corresponding extensions are the usual central extensions classified by the abelian 2-cocycle. 
From a physical point of view such a 2-cocycle can appear if the $G$-action on matter fields only closes up to a $D$-gauge transformation.
We give two concrete examples for later use. Let $N$ and $M$ be positive integers.
Identifying the cyclic groups $\Z_N$ and $\Z_M$ with $\{0,1, \dots , N-1 \} $ and $\{0,1, \dots , M-1 \} $ we define the 2-cocycle 
\begin{align}
\sigma \colon \Z_M\times \Z_M &\longrightarrow \Z_N \\
(a, b)& \longmapsto \Big\lfloor \frac{a+b}{M} \Big\rfloor \text{ mod }N \ .
\end{align}  
The corresponding central extension is 
\begin{align}
0 \longrightarrow \Z_N \xrightarrow{M\,\cdot}  \Z_{N\,M}
  \longrightarrow \Z_M \longrightarrow 0 \ , 
\end{align}
where the first map is multiplication by $M$ and the second map is
reduction modulo $N$.
This example can be adapted to an arbitrary number of copies of $\Z_N$ and $\Z_M$. An example is the abelian 2-cocycle corresponding to
\begin{align}
(0,0) \longrightarrow \Z_N\times \Z_N \xrightarrow{(M ,M)\,
  \cdot}\Z_{N\,M} \times \Z_{N\,M} \longrightarrow \Z_M \times \Z_{M}
  \longrightarrow (0,0)
\end{align}
which is given by 
\begin{align}
 (\Z_M\times \Z_M)^2 &\longrightarrow \Z_N\times \Z_N \\
\big( (a_1, b_1)\,,\,(a_2,b_2)\big)& \longmapsto \Big(\Big\lfloor
                                     \frac{a_1+a_2}{M} \Big\rfloor
                                     \text{ mod }N \,,\, \Big\lfloor
                                     \frac{b_1+b_2}{M} \Big\rfloor
                                     \text{ mod }N \Big) \ .
\end{align}
\end{example}
\begin{example}
Given a group homomorphism $\alpha\colon G\longrightarrow
\Aut_{\Grp}(D)$, we can consider $\alpha$ as a non-abelian 2-cocycle
with trivial map $\sigma_\alpha$. The corresponding extension is the semi-direct product
\begin{align}
1 \longrightarrow D \longrightarrow G \ltimes_\alpha D \longrightarrow G \longrightarrow 1 \ .
\end{align}
\end{example}

\subsection{Gauging discrete symmetries and 't Hooft anomalies}\label{Sec: Gauging}

There is an inclusion of categories $i\colon \Cob_n \hookrightarrow G\text{-}\Cob_n$ for every group $G$ by equipping every manifold with the trivial $G$-bundle. The pullback $i^* Z_G \colon \Cob_n \longrightarrow \fvs$ of a $G$-equivariant field theory $Z_G:G\text{-}\Cob_n\longrightarrow\fvs$ carries additional structure: 
By evaluating $Z_G$ on gauge transformations of the trivial $G$-bundle
on an $n{-}1$-dimensional manifold $\Sigma$ we get a representation of
$G$ on $i^* Z_G(\Sigma)$ which is compatible with the definition on
cobordisms. Hence $i^* Z_G$ is a quantum field theory with internal
$G$-symmetry in the sense of Definition~\ref{Def: Internal symmetry}, i.e. a symmetric monoidal functor
\begin{align}
i^* Z_G \colon \Cob_n \longrightarrow G\text{-}\mathsf{Rep} \ .
\end{align}   
Recall that we considered a $G$-equivariant field theory as a field theory coupled to classical $G$-gauge fields. Given a field theory $Z \colon \Cob_n \longrightarrow G\text{-}\mathsf{Rep}$ with internal $G$-symmetry we can ask if the symmetry can be gauged.

\begin{definition}\label{Def: Gauging}
Let $Z\colon \Cob_n \longrightarrow G\text{-}\mathsf{Rep}$ be a topological quantum field theory with internal $G$-symmetry. A $G$-equivariant field theory $Z_G\colon \GCob \longrightarrow \fvs$ {gauges} the internal $G$-symmetry if $i^*Z_G=Z$ as functors $\Cob_n \longrightarrow G\text{-}\mathsf{Rep} $.
\end{definition}
In general it may be impossible to gauge a given symmetry due to cohomological obstructions. In this case we say that the symmetry has a 't Hooft anomaly. In the following we will study under which conditions the symmetries discussed in Section \ref{Sec: Finite symmetries} have 't Hooft anomalies.

\begin{remark}
In three dimensions the question of whether a given field theory can be gauged is related to an interesting algebraic problem.
A three-dimensional extended topological quantum field theory is described by a modular tensor category $\mathsf{M}$~\cite{bartlett2015modular}. An internal $G$-symmetry corresponds to a homotopy coherent action of $G$ on $\mathsf{M}$ via braided autoequivalences. The group of braided autoequivalences up to natural isomorphism is known as the Brauer-Picard group.   
The modular tensor category corresponding to the Dijkgraaf-Witten theory with gauge group $D$ and topological action $\omega\in Z^3(BD;U(1))$ is the category of finite-dimensional modules over the $\omega$-twisted Drinfeld double of the group algebra $\C[D]$ defined in~\cite{Twisted_Drinfeld_double}.
The corresponding Brauer-Picard group for $\omega=0$ is studied in detail in~\cite{Lentner:2015pla}. The more general case of the representation category of Hopf algebras which includes the case of non-trivial $\omega$ is studied in~\cite{2017arXiv170205133L}. The kinematical symmetries studied in this paper correspond to the subgroup of classical symmetries in~\cite{2017arXiv170205133L}.       
Three-dimensional $G$-equivariant extended field theories correspond to $G$-modular categories~\cite{TV14,EOFK}. 
The symmetry corresponding to a homotopy coherent action of $G$ on a modular tensor category $\mathsf{M}$ can be gauged if there exists a $G$-modular category $\mathsf{M}_G= \bigoplus_{g\in G}\, \mathsf{M}_g$ such that $\mathsf{M}_1=\mathsf{M}$ in a compatible way. The question of under which conditions such an extension exists is answered in~\cite{2009arXiv0909.3140E}, whereby the case relevant for Dijkgraaf-Witten theories is discussed in their appendix. 

In this algebraic framework the gauging of more complicated symmetries
of arbitrary three-dimensional extended topological field theories can
be addressed using the cobordism hypothesis and representation
theoretic techniques. A detailed study of this would be interesting. However, we refrain from doing so in this paper and focus instead on a largely dimension-independent discussion.   
\end{remark}

The non-abelian group 2-cocycle describing the action of $G$ on a Dijkgraaf-Witten theory with gauge group $D$ and topological action $\omega \in Z^n(BD;U(1))$ determines an (not necessarily central) extension 
\begin{align}
\label{Group extension}
1\longrightarrow D \overset{\iota}{\longrightarrow} \widehat{G} \overset{\lambda }{\longrightarrow} G \longrightarrow 1\ .
\end{align} 
The short exact sequence should be understood as a way to combine $D$- and $G$-gauge fields into a single $\widehat{G}$-gauge field. 
If there exists $\widehat{\omega} \in H^n(B\widehat{G};U(1))$ such that $\iota^* \widehat{\omega}= \omega$ we say that the symmetry $G$ is {anomaly-free}.
In this case we can use the pushforward construction \cite[Section~4]{OFK} along $\lambda$ to get a $G$-equivariant field theory gauging the symmetry. We review this construction in the concrete example we are interested in.

Let $\Sigma$ be an $n{-}1$-dimensional closed manifold.
The group homomorphism $\lambda$ induces an extension functor 
\begin{align}
\lambda_* \colon \mathsf{Bun}_{\widehat{G}}(\Sigma) \longrightarrow \mathsf{Bun}_G(\Sigma) \ . 
\end{align}  
This functor acts on classifying maps by post-composition with the map $B\widehat{G}\longrightarrow B G$ induced by $\lambda \colon \widehat{G}\longrightarrow G$, which by a slight abuse of notation we denote again by $\lambda$. 
Equivalently this functor associates to a $\widehat{G}$-bundle $\widehat{\varphi}$ on $\Sigma$ the induced $G$-bundle 
\[
\widehat{\varphi}\times_\lambda G = \widehat{\varphi}\times G \, / \sim \ ,
\] 
where $(p,g)\sim (p\triangleleft \widehat{g},\lambda(\widehat{g}^{\,-1})\, g)$ for all $p\in \widehat{\varphi}$, $g\in G$ and $\widehat{g} \in \widehat{G}$.
For a bundle $\varphi \in \BunG(\Sigma)$ we denote by $\lambda_*^{-1}[\varphi]$ the homotopy fibre 
\begin{equation}
\begin{tikzcd}
\lambda_*^{-1}[\varphi] \ar[r] \ar[d] & \mathsf{Bun}_{\widehat{G}}(\Sigma) \ar[d, "\lambda_*"] \\
* \ar[r,swap, "\varphi"] & \BunG(\Sigma)
\end{tikzcd}
\end{equation}
Concretely, objects of $\lambda_*^{-1}[\varphi]$ are pairs $(\widehat{\varphi},h)$ of a $\widehat{G}$-bundle $\widehat{\varphi}$ and a gauge transformation $h \colon \lambda_* \widehat{\varphi} \longrightarrow \varphi $. Morphisms are gauge transformations $\widehat{h} \colon \widehat{\varphi} \longrightarrow \widehat{\varphi}\,'$ such that the diagram
\begin{equation}
\begin{tikzcd}
\lambda_* \widehat{\varphi} \ar[rr, "\lambda_* \widehat{h}"] \ar[rd, " h ", swap] & & \lambda_* \widehat{\varphi}\,' \ar[ld, " h' "] \\
 & \varphi &
\end{tikzcd}
\end{equation}
commutes.
The pushforward of $L_{\widehat{\omega}}$ can now be defined on an object $(\Sigma, \varphi \colon \Sigma \longrightarrow BG)$ as
\begin{align}
\label{EQ:EDW on Objects}
\lambda_* L_{\widehat{\omega}}(\Sigma, \varphi)= \int_{\lambda_*^{-1}[\varphi]} \, \Big( \lambda_*^{-1}[\varphi] \longrightarrow \mathsf{Bun}_{\widehat{G}}(\Sigma) \xrightarrow{L_{\widehat{\omega}}} \fvs \Big) \ .
\end{align}
This should be regarded as a quantization of the $D$-gauge fields
while leaving the $G$-sector classical. 
We realise this limit again as the space of parallel sections. In this case a parallel section $f\in \lambda_* L_{\widehat{\omega}}(\Sigma, \varphi)$ consists of an element $f(\widehat{\varphi},h)\in L_{\widehat{\omega}}(\Sigma,\widehat{\varphi})$ for all $(\widehat{\varphi},h) \in \lambda_*^{-1}[\varphi] $. 

Let $(M, \psi)\colon (\Sigma_a, \varphi_a)\longrightarrow (\Sigma_b,\varphi_b)$ be a morphism in $\GCob$.
To define the pushforward on a parallel section $f(\,\cdot\,)\in \lambda_* L_{\widehat{\omega}}(\Sigma_a, \varphi_a)$ we fix fundamental cycles $\sigma_a$ and $\sigma_b$ of $\Sigma_a$ and $\Sigma_b$, respectively, and write $f$ as $f(\,\cdot\,)=\mathbf{f}(\,\cdot\,)\,[\sigma_a]$ as in Section~\ref{sec:quantumDWtheories}. 
We define
\begin{align}
\label{EQ:EDW on morphisms}
\hspace{-10mm} \lambda_* L_{\widehat{\omega}}(M, \psi)[f](\widehat{\varphi}_b,h_b) = \Big(\int_{(\widehat{\psi},h,\widehat{h}\,)\in \lambda_*^{-1}[\psi]|_{(\widehat{\varphi}_b,h_b)}} \, \big\langle \, \widehat{h}^* \widehat{\omega},[0,1] \times \sigma_b \big \rangle \, \big\langle  \widehat{\psi}^* \widehat{\omega}, \sigma_M \big \rangle \, \mathbf{f}\big(\widehat{\psi}|_{\Sigma_a}, h|_{\Sigma_a}\big)\Big) \ [\sigma_b]
\end{align}
with $\sigma_M \in \Fund^{\sigma_b}_{\sigma_a}(M)$; here the homotopy pullback
$\lambda_*^{-1}[\psi]|_{(\widehat{\varphi}_b,h_b)}$ is the groupoid
with objects $(\widehat{\psi},h,\widehat{h}\,)$ where $
\widehat{\psi} \colon M \longrightarrow B\widehat{G}$ is a
$\widehat{G}$-bundle, $h \colon \lambda_* \widehat{\psi}
\longrightarrow \psi$ is a gauge transformation, and
$\widehat{h}\colon \widehat{\psi}\,|_{\Sigma_b} \longrightarrow
\widehat{\varphi}_b$ is a gauge transformation
such that the diagram
\begin{equation}
\begin{tikzcd}
\lambda_*\widehat{\psi}\,\big|_{\Sigma_b} \ar[rr, "\lambda_*\widehat{h}"] \ar[rd,"h|_{\Sigma_b}",swap] & & \lambda_*\widehat{\varphi}_b \ar[ld, "h_b"] \\
 & \varphi_b &
\end{tikzcd}
\end{equation}
commutes. 
This defines a functorial field theory by~\cite[Proposition~4.2]{OFK}. The field theories constructed here are called $G$-equivariant Dijkgraaf-Witten theories. In the case of trivial 2-cocycle they have been studied as extended field theories in \cite{NMS}. Extended $G$-equivariant Dijkgraaf-Witten theories with non-trivial 2-cocycles are constructed in \cite{EHQFT}. 

Now we can formulate the central statement of this section.
   
\begin{theorem}\label{Theorem: Gauging}
Let $Z_{\omega}$ be a discrete gauge theory with topological action $\omega\in Z^n(BD;U(1))$ and kinematical $G$-symmetry described by an extension
\begin{align}
1\longrightarrow D \overset{\iota}{\longrightarrow} \widehat{G} \overset{\lambda }{\longrightarrow} G \longrightarrow 1
\end{align}
such that there exists $\widehat{\omega} \in Z^n(B\widehat{G};U(1))$ satisfying $\omega =\iota^* \widehat{\omega}$. Then the $G$-equivariant Dijkgraaf-Witten theory $\lambda_* L_{\widehat{\omega}} \colon \GCob \longrightarrow \fvs$ gauges this symmetry. 
\end{theorem}
\begin{proof}
First we show that the trivial sector of $\lambda_* L_{\widehat{\omega}}$ is $Z_{\omega}$. From the exact sequence of groups we get a fibration
\begin{align}
BD \overset{\iota}{\longrightarrow} B\widehat{G} \overset{\lambda}{\longrightarrow} BG
\end{align}
of classifying spaces. Let $N$ be a manifold of dimension $n$ or $n-1$. We have to evaluate the homotopy fibre $\lambda_*^{-1}[\star]$ of the trivial bundle $\star \colon N \longrightarrow BG$. Using the homotopy lifting property we can restrict ourselves to the full subgroupoid with objects $(\widehat{\varphi} \colon N \longrightarrow \lambda^{-1}(*)=\iota (BD), \id)$, where $\lambda^{-1}(*)$ is the preimage of the base point of $BG$. Morphisms in this subgroupoid are homotopies which are trivial after applying $\lambda$. Using again the homotopy lifting property of a fibration we see that they must be homotopic to a homotopy supported in $BD$. This shows that we can replace limits and integration over the homotopy fibre of the trivial bundle with limits and integration over $\BunD(N)$ for every manifold $N$. For this reason \eqref{EQ:EDW on Objects} and \eqref{EQ:EDW on morphisms} reduce to \eqref{Definition on Objects} and \eqref{Definition on Morphisms} in this case, since $\widehat{\omega}$ pulls back to $\omega$.

Next we show that this gauges the symmetry, see~\eqref{Eq: Concrete form of the quantum symmetry}:
\begin{align}
\lambda_*L_{\widehat{\omega}}(\Sigma,g)f(\iota_* \varphi_D,\id)=f\big(\iota_* \mathcal{R}(\alpha(g^{-1}))\varphi_D, \id\big) \, \langle \varphi_D^*\Phi_g , \sigma_\Sigma \rangle
\end{align}  
for all closed $n{-}1$-dimensional manifolds $\Sigma$, $f(\,\cdot\,) \in \lambda_*L_{\widehat{\omega}}(\Sigma, \star \colon \Sigma \longrightarrow BG)$ and $\varphi_D \colon \Sigma \longrightarrow BD$, where we interpret $g\in G$ as a homotopy from the constant map $\star$ to itself. 
By~\cite[Proposition~4.2~(b)]{OFK} we have 
\begin{align}
\lambda_*L_{\widehat{\omega}}(\Sigma,g)f(\iota_* \varphi_D,\id)= f(\iota_* \varphi_D,g^{-1}) \ .
\end{align}  
We have to calculate a lift for the homotopy $g^{-1}$, i.e. a gauge transformation $\widehat{g}^{\,-1} \colon \iota_* \varphi_D \longrightarrow \iota_* \varphi'_D$ such that $\lambda (\widehat{g}^{\,-1}) = g^{-1}$. We use the concrete description of $\widehat{G}$-bundles as elements of the functor category $[\Pi_1(\Sigma),* \Ds \widehat G\,]$, where $\Pi_1(\Sigma)$ is the fundamental groupoid of $\Sigma$. A lift of the gauge transformation is then given by conjugation with $(g^{-1},1)\in \widehat G$. We calculate its action on the image $d\in D \subset \widehat{G}$ of a path in $\Sigma$. The inverse is given by \cite{Blanco2005}
\begin{equation}
\left( g \,,\, \sigma_\alpha(g,g^{-1})^{-1} \right ) \ .
\end{equation} 
Then
\begin{align}
\big(g^{-1},1\big)\,\big(1,d\big)\,\big( g\,,\, \sigma_\alpha(g,g^{-1})^{-1}\big) &= \big(g^{-1},\alpha(g^{-1})[d]\big)\,\big( g\,,\, \sigma_\alpha(g,g^{-1})^{-1}\big) \\[4pt]
&= \big( 1\,,\,\alpha(g^{-1})[d]\,\alpha(g^{-1})[\sigma_\alpha(g,g^{-1})^{-1}]\,\sigma_\alpha(g^{-1},g)\big) \\[4pt]
&= \big( 1\,,\,\alpha(g^{-1})[d]\, \sigma_\alpha (g^{-1},g)^{-1} \, \sigma_\alpha (g^{-1},g) \big) \\[4pt]
&= \big( 1,\alpha(g^{-1})[d]\big) \ ,
\end{align}
where in the third equality we used $\sigma_\alpha(g^{-1},g)^{-1}= \alpha(g^{-1})[\sigma_\alpha(g,g^{-1})^{-1}]$, which follows from \eqref{Eq: Non abelian cocycle} with $g_1=g^{-1}$, $g_2=g$ and $g_3= g^{-1}$ using $\sigma_\alpha({1,g})=\sigma_\alpha({g,1})=1$ for all $g\in G$. This shows that $\varphi_D'= \mathcal{R}(\alpha(g^{-1}))\varphi_D$. 
That $f(\,\cdot\,)$ is a parallel section implies
\begin{align}
f(\iota_* \varphi_D,g^{-1})= L_{\widehat{\omega}}^{-1}(\Sigma,\widehat{g})f\big(\mathcal{R}(\alpha(g^{-1}))\varphi_D , \id \big) \ . \label{Eq: Proof Main theorem}
\end{align}  

We define an $n{-}1$-cochain $\Phi_g$ on $D$ as follows: 
Let $\chi \colon \Delta^{n-1}\longrightarrow BD$ be an $n{-}1$-chain which we can include into $B\widehat{G}$ along $\iota$. Putting $\widehat{g}$ on the interval we get a map $[0,1]\times \Delta^{n-1} \longrightarrow B\widehat{G}$. Integration of the pullback of $\widehat{\omega}$ over $[0,1]\times \Delta^{n-1}$ gives the inverse of the value of $\Phi_g$ evaluated on the $n{-}1$-simplex. 
The value of $-\delta \Phi_g$ on an $n$-simplex $(d_1,\dots ,d_n) \colon \Delta^n \longrightarrow BD$ is given by 
\begin{align}
\big\langle [\widehat{g}\times(d_1,\dots , d_n)]^*\widehat{\omega} \,,\, [0,1]\times \partial \Delta^n \big\rangle & = \big\langle [\widehat{g}\times(d_1,\dots ,d_n)]^*\widehat{\omega} \,,\, \partial[0,1]\times \Delta^n - \partial([0,1]\times \Delta^n) \big\rangle \\[4pt]
&= \big\langle [\widehat{g}\times(d_1,\dots ,d_n)]^*\widehat{\omega} \,,\, (\{1\}-\{0\})\times \Delta^n \big\rangle \\[4pt]
&= \alpha(g^{-1})^*\omega(d_1, \dots ,d_n)- \omega(d_1, \dots , d_n) \ .
\end{align}  
By definition $L^{-1}_\omega(\Sigma,\widehat{g})= \langle \varphi_D^* \Phi_g, \sigma_\Sigma \rangle$. Inserting this into \eqref{Eq: Proof Main theorem} gives \eqref{Eq: Concrete form of the quantum symmetry} where $\Phi_g$ provide the homotopy fixed point structure.
\end{proof}

\begin{remark}
Theorem~\ref{Theorem: Gauging} provides a general mechanism for the gauging of
symmetries. However, we cannot show that it is impossible to gauge the
symmetry when the conditions of Theorem~\ref{Theorem: Gauging} are not satisfied, i.e.
when no such $\widehat{\omega}$ exists.\end{remark}

\begin{example}\label{Example: Gauging1}
We describe a discrete two-dimensional gauge theory with gauge group $D=\Z_N\times \Z_N$ and topological action $\omega_k \in H^2(\Z_N\times \Z_N; U(1))$ as defined in~\eqref{EQ:Def 2 cocycle}. The action of the symmetry group $G$ on $D$ can be encoded in a short exact sequence
\begin{align*}
1\longrightarrow D \longrightarrow \widehat{G} \longrightarrow G \longrightarrow 1 \ .
\end{align*}  
Set $G= \Z_M \times \Z_M$ and consider the extension
\begin{align}
(0,0)\longrightarrow \Z_N\times \Z_N \xrightarrow{(M , M)\,\cdot}
  \Z_{N\,M}\times \Z_{N\,M}\longrightarrow \Z_M \times \Z_M
  \longrightarrow (0,0) \ .
\end{align} 
In this case we can gauge the symmetry in the manner of Theorem~\ref{Theorem: Gauging} for the topological action
$\omega_k$ with $k\in \lbrace 0,1,\dots, N-1 \rbrace$ if and only if
$k$ is divisible by $M$ modulo $N$, i.e. there exists $k'\in \Z$ such that $k'\,M = k \text{ mod}\ N$. Concretely, $\widehat{\omega}\in H^2(\Z_{N\,M}\times \Z_{N\,M}; U(1))$ is given by $\omega_{k'}\in Z^2(\Z_{N\,M}\times \Z_{N\,M};U(1))$. 
This simple example already shows that we cannot gauge every symmetry using Theorem~\ref{Theorem: Gauging};
it is discussed in~\cite{Kapustin:2014gua,Gaiotto2014} in the context of
$0$-form and $1$-form global symmetries. We will discuss obstructions
to finding an appropriate lift $\widehat{\omega}$ in more detail and generality in Section~\ref{Sec: Obstructions}. 
\end{example}

\begin{example}\label{Ex: D8 Z2}
The cyclic group $\Z_2$ acts on the dihedral group $D_8$ by conjugation with the generator $a$. Since this is an action via inner automorphisms it preserves the non-trivial 2-cocycle $\omega\in H^2(BD_8;U(1))$ from Example~\ref{Ex: 2D coycles}.b. This action defines a non-abelian 2-cocycle with trivial map $\sigma$. The corresponding extension is
\begin{align}
1 \longrightarrow D_8 \longrightarrow D_8 \rtimes \Z_2 \longrightarrow \Z_2 \longrightarrow 1 \ . 
\end{align} 
The Pauli group is the subgroup 
\begin{align}
P_1 = \lbrace \pm\, \mathds{1}_2 , \pm \,\iu\, \mathds{1}_2 , \pm\, \sigma_x , \pm \,\iu\, \sigma_x , \pm\, \sigma_y , \pm \,\iu\, \sigma_y , \pm\, \sigma_z , \pm \,\iu\, \sigma_z \rbrace
\end{align}
of the unitary group $U(2)$ with the Pauli spin matrices
\begin{align}
\sigma_x=  \bigg( \begin{matrix}
0 & 1 \\ 
1 & 0
\end{matrix} 
\bigg) \ , \quad
\sigma_y= \bigg(
\begin{matrix}
0 & -\,\iu\, \\ 
\,\iu\, & 0
\end{matrix} 
\bigg) \qquad \mbox{and} \qquad
\sigma_z=  \bigg(
\begin{matrix}
1 & 0 \\ 
0 & -1
\end{matrix} 
\bigg) \ . 
\end{align}
There is an equivalence of extensions 
\begin{equation}
\begin{tikzcd}
 &  & D_8\rtimes \Z_2 \ar[dd, "\vartheta"] \ar[rd] & & \\
1 \ar[r] & D_8 \ar[rd] \ar[ru] & & \Z_2 \ar[r] & 1\\
 &  & P_1 \ar[ru] & &
\end{tikzcd}
\end{equation}
given by $\vartheta (a^i\,b^j,k)= (\,\iu\, \sigma_x)^i\,\sigma_y^j\,\sigma_x^k$, showing that this extension is non-trivial even though it comes from an inner automorphism. The intuitive reason for this is that conjugation by $a^2$ is the identity even though $a^2$ itself is not. 
We will show in Example~\ref{Ex: D8 Z2 v2} that this symmetry cannot be gauged in the manner of Theorem~\ref{Theorem: Gauging}.
\end{example}

\begin{example}\label{Example: Gauging2}
In three dimensions we can look at the extension 
\begin{align}
0\longrightarrow \Z_N \xrightarrow{M\,\cdot } \Z_{N\,M}
  \longrightarrow \Z_M \longrightarrow 0 \ .
\end{align}
The 3-cocycle $\omega_k$ defined in~\eqref{Eq: Zn 3-cocycle} can
always be gauged by the 3-cocycle $\widehat{\omega}_k \in
H^3(B\Z_{N\,M};U(1))$ corresponding to the same value of $k$. 
\end{example}

\subsection{Obstructions to gauging of symmetries}\label{Sec: Obstructions}

We shall now work with the group cohomology $H^n(G;U(1))$ which can be identified with the cohomology of $BG$ with coefficients in $U(1)$.
There are obstructions for $\widehat{\omega}$ to exist which follow
from the Lyndon-Hochschild-Serre spectral sequence associated to the extension \eqref{Group extension}. We briefly review these obstructions. For a physical perspective on these obstructions, see~\cite{Thorngren2015}.
 
There is an action of $G$ on $H^n(D;U(1))$ induced by conjugation in $\widehat G$. Every cocycle on $\widehat{G}$ is invariant under conjugation and hence the first obstruction for $\widehat \omega$ to exist is 
\begin{align}
\omega \in H^n\big(D;U(1)\big)^G \ .
\end{align}
By definition, the obstruction is always satisfied if the extension
corresponds to a kinematical symmetry. The first quadrant
Lyndon-Hochschild-Serre spectral sequence corresponding to the exact
sequence~\eqref{Group extension} takes the form  
\begin{align}
E^{p,q}_2=H^p\big(G;H^q(D;U(1))\big)\Longrightarrow H^{p+q}\big(\widehat G ; U(1)\big) 
\end{align}
with edge maps $H^{n}(\widehat G ; U(1)) \twoheadrightarrow E^{ 0,n}_\infty = E^{0,n}_{n+2} \hookrightarrow H^n(D;U(1))^G $ given by the restriction to $D$ (see e.g.~\cite[Section~6.8]{Weibel}). Hence we see that $\omega \in \im(\iota^*)=E^{0,n}_{n+2}$ if and only if 
\begin{align}\label{Eq: Obstruction}
{\rm d}^{0,n}_i \omega = 0 \ \in \ E^{i,n+1-i}_i 
\end{align}
for all $i \in \{ 2, \dots , n+1 \}$. Note that $\dd^{0,n}_i \omega$ is only well-defined if $\dd^{0,n}_{i-1}\omega = 0$ and $E^{i,n+1-i}_i$ is a sub-quotient of $H^i\big(G;H^{n+1-i}(D;U(1))\big)$. 

To understand these obstructions in more detail we introduce the algebraic model for the spectral sequence~\cite[Section 2]{HS53}. The group cohomology of $\widehat{G}$ can be computed from the normalised cochain complex $C^\bullet(\widehat{G};U(1))$:
\begin{align}
0 \longrightarrow C^0\big(\widehat{G};U(1)\big) \longrightarrow C^{1}\big(\widehat{G};U(1)\big) \longrightarrow \cdots \ .
\end{align}
We introduce a filtration 
\begin{align}
 C^\bullet\big(\widehat{G};U(1)\big)=F^0C^\bullet\big(\widehat{G};U(1)\big) \supseteq F^1 C^\bullet\big(\widehat{G};U(1)\big)\supseteq F^2 C^\bullet\big(\widehat{G};U(1)\big) \supseteq \cdots 
\end{align}
where $F^i C^n(\widehat{G};U(1))$ is $0$ for $i>n$ and otherwise consists of all normalized $n$-cochains which are $0$ as soon as $n-i+1$ entries are in the image of $D$.
This filtration is compatible with the coboundary operator $\delta$ and hence induces a spectral sequence, which is the Lyndon-Hochschild-Serre spectral sequence. Concretely we set
\begin{align*}
Z^{p,q}_r & \coloneqq \text{ker} \Big(F^p
            C^{p+q}\big(\widehat{G};U(1)\big)\overset{\delta}{\longrightarrow}
            C^{p+q+1}\big(\widehat{G};U(1)\big)\big/F^{p+r}
            C^{p+q+1}\big(\widehat{G};U(1)\big) \Big) \ , \\[4pt]
B^{p,q}_r & \coloneqq \delta\Big(
            F^{p-r+1}C^{p+q-1}\big(\widehat{G};U(1)\big)\Big)\cap F^p
            C^{p+q}\big(\widehat{G};U(1)\big) \ , \\[4pt]
E^{p,q}_r &\coloneqq Z^{p,q}_r \big/ \big(B^{p,q}_r + Z_{r-1}^{p+1,q-1}\big) \ .
\end{align*} 
The differential $\delta\colon  C^{p+q}(\widehat{G};U(1)) \longrightarrow C^{(p+r)+(q-r+1)}(\widehat{G};U(1))$ induces the corresponding differentials
\begin{align}
\dd^{p,q}_r \colon E^{p,q}_r \longrightarrow E^{p+r,q-r+1}_r 
\end{align}
in the spectral sequence.

We consider the two-dimensional case as a warm-up. We fix $\omega \in H^2(D;U(1))$. 
The corresponding element in $E^{0,2}_2$ is the 2-cochain 
\begin{align}
\tilde{\omega} \colon \widehat{G}\times \widehat{G} &\longrightarrow U(1) \\
\big((d,g)\,,\, (d',g')\big) & \longmapsto \omega(d,d') \ .
\end{align}  
This is not generally a cocycle, since the multiplication in $\widehat{G}$ is twisted by the corresponding non-abelian 2-cocycle. This cochain obviously pulls back to $\omega$. The ensuing calculation can be understood as trying to find a 2-cochain on $\widehat{G}$ which is $0$ when pulled back to $D$ such that its sum with $\tilde{\omega}$ is closed. 

The first obstruction $\dd^{0,2}_2 \tilde{\omega} = 0$ is equivalent to $\delta \tilde{\omega} \in B^{2,1}_2 + Z^{3,0}_1$. This implies that there exists $\gamma_1 \in F^1C^2(\widehat{G};U(1))$ such that 
\begin{align}
\delta \gamma_1 \in F^2 C^3\big(\widehat{G};U(1)\big) \qquad \mbox{and} \qquad \delta (\tilde{\omega}-\gamma_1) \in Z^{3,0}_1  \ . 
\end{align} 
This means that we can consider $\tilde{\omega}$ as an element of
$E^{0,2}_3 \cong \ker\,\dd^{0,2}_2$. 
Note that the identification is not the identity, rather we have to map $\tilde{\omega}$ to $\tilde{\omega}-\gamma_1$. 
We have thus shown that if the first obstruction vanishes, then there exists $\theta\in Z^{3,0}_1 = Z^{3}(G;U(1))$ and a cochain $\omega'=\tilde{\omega}-\gamma_1$ such that $\delta \omega' = \lambda^* \theta$ and $\iota^* \omega'= \omega$. 
 
The next obstruction is $\dd^{0,2}_3 \tilde{\omega}= 0$. This is equivalent to $\delta (\tilde{\omega}-\gamma_1) \in B^{3,0}_3$, hence there exists $\gamma_2\in F^1C^2(\widehat{G};U(1))$ such that $\delta \gamma_2= \delta(\tilde{\omega}-\gamma_1)\in F^3C^2(\widehat{G};U(1))$. This implies $\delta(\tilde{\omega}-\gamma_1 -\gamma_2)= 0$ and $\iota^*(\tilde{\omega}-\gamma_1 -\gamma_2) = \omega$, since $\gamma_1$ and $\gamma_2$ are elements in $F^1C^2(\widehat{G};U(1))$. This gives the desired 2-cocycle $\widehat{\omega}=\tilde{\omega}-\gamma_1 -\gamma_2$.
 
The discussion above readily generalises to arbitrary dimension $n$. If the first obstruction vanishes then there exists $\gamma_1 \in  F^1C^n(\widehat{G};U(1))$ such that
\begin{align}
\delta(\tilde{\omega}-\gamma_1)\in F^3C^{n+1}\big(\widehat{G};U(1)\big) \ .
\end{align}
More generally if the first $m\leq n$ obstructions vanish, there are elements $\gamma_1,\dots , \gamma_m \in F^1C^n(\widehat{G};U(1)) $ such that 
\begin{align}
\delta \gamma_i \ &\in \ F^i C^{n+1}\big(\widehat{G};U(1)\big) \ , \\[4pt]
\delta\Big(\tilde{\omega} -\sum_{i=1}^k\, \gamma_i\Big) \  &\in \ F^{k+2} C^{n+1}\big(\widehat{G};U(1)\big) \ , 
\end{align} 
for all $i,k=1,\dots,m$.
In particular, if all obstructions vanish then
\begin{align}
\delta\Big(\tilde{\omega} -\sum_{i=1}^n\, \gamma_i\Big) = 0 
\end{align}
and 
\begin{align}
\iota^* \Big(\tilde{\omega} -\sum_{i=1}^n\, \gamma_i\Big) = \omega \ \in \ H^n\big(D;U(1)\big) \ . 
\end{align}

We are mostly interested in the case when all obstructions except the last one vanish. In this case 
\begin{align}
\delta\Big(\tilde{\omega} -\sum_{i=1}^{n-1}\, \gamma_i\Big) = \lambda^* \theta  
\end{align}
with $\theta \in Z^{n+1}(G;U(1))$, since closed elements of $F^{n+1}C^{n+1}(\widehat{G};U(1))$ are in one-to-one correspondence with $Z^{n+1}(G;U(1))$. 
We summarize the present discussion in 

\begin{proposition}\label{Prop: Obstructions}
Let 
\begin{align}
1 \longrightarrow D \overset{\iota}{\longrightarrow} \widehat{G}
  \overset{\lambda}{\longrightarrow} G \longrightarrow 1
\end{align}
be a short exact sequence of groups, $n$ a natural number and $\omega$ an $n$-cocycle on $D$ with values in $U(1)$. 
\begin{itemize}
\item[{\rm (a)}]
When all obstructions in \eqref{Eq: Obstruction} vanish, then there exists $\widehat{\omega}\in Z^n(\widehat{G};U(1))$ satisfying $\iota^* \widehat{\omega}= \omega$. 

\item[{\rm (b)}]
When the first $n-1$ obstructions in \eqref{Eq: Obstruction} vanish, then there exist $\omega'\in C^n(\widehat{G};U(1))$ and $\theta\in Z^{n+1}(G;U(1))$ satisfying $\iota^* \omega'= \omega$ and $\delta \omega'=\lambda^*\theta$. 
\end{itemize} 
\end{proposition}

\begin{remark}
If the first $n-1$ obstructions vanish we can realize the anomalous
field theory as a boundary state of a classical
$n{+}1$-dimensional Dijkgraaf-Witten theory with topological action
$\theta$. In Section~\ref{Sec: Realisation on boundary} we will
explain this point in more detail. 
\end{remark}

\begin{example}
We have seen in Example~\ref{Example: Gauging2} that for the extension 
\begin{align}
0\longrightarrow \Z_N \longrightarrow \Z_{N\,M}\longrightarrow \Z_M \longrightarrow 0
\end{align}
all 3-cocycles on $\Z_N$ arise as pullbacks of 3-cocycles on $\Z_{N\,M}$, hence all obstructions vanish in this case.   
\end{example}

\begin{example}\label{Ex: D8 Z2 v2}
Following up on Example~\ref{Ex: D8 Z2} we show that for the symmetry described by 
\begin{align}\label{Eq: Exact sequence D8 Z2}
1 \longrightarrow D_8 \longrightarrow P_1 \longrightarrow \Z_2 \longrightarrow 1
\end{align} 
the non-trivial 2-cocycle $\omega\in H^2(D_8;U(1))$ cannot be
gauged. The cohomology groups of the Pauli group $P_1$ can be computed using a
computer algebra package such as {\tt GAP}~\cite{Joyner08aprimer} and
the universal coefficient theorem to get
\begin{align}
H^0\big(P_1; U(1)\big)&= U(1) \ , \\[4pt]
H^1\big(P_1; U(1)\big)&= \Z_2 \times \Z_2 \times \Z_2 \ , \\[4pt]
H^2\big(P_1; U(1)\big)&= \Z_2\times \Z_2 \ , \\[4pt]
H^3\big(P_1; U(1)\big)&= \Z_2\times \Z_2\times \Z_8  \ .
\end{align}  
The $E_2$ page of the corresponding spectral sequence is 
\begin{center}
\begin{tikzpicture}
  \matrix (m) [matrix of math nodes,
    nodes in empty cells,nodes={minimum width=5ex,
    minimum height=5ex,outer sep=-5pt},
    column sep=1ex,row sep=1ex]{
                &         &          &         &           & \\
          2     &     \Z_2 &  \Z_2   &   \Z_2  &   \Z_2    &  \\
          1     &  \Z_2\times \Z_2 &   \Z_2\times \Z_2  &  \Z_2\times \Z_2 &  \Z_2\times \Z_2       &  \\
          0     &  U(1)  & \Z_2 &  0  &  \Z_2     &  \\
    \quad\strut &   0  &  1  &  2  &   3   & \strut \\};
  \draw[-stealth] (m-2-2.south east) -- (m-3-4.north west);
  \draw[-stealth] (m-3-3.south east) -- (m-4-5.north west);
\draw[thick] (m-1-1.east) -- (m-5-1.east) ;
\draw[thick] (m-5-1.north) -- (m-5-6.north) ;
\end{tikzpicture}
\end{center}
The two differentials drawn are $0$ as can be checked by using the concrete description of the differentials in~\cite{HUEBSCHMANN1981296} and the fact that \eqref{Eq: Exact sequence D8 Z2} is the extension of $\Z_2$ by $D_8$ corresponding to the inner automorphism of $D_8$ given by conjugation with $a\in D_8$.  
Hence the $E_3$ page is given by 
\begin{center}
\begin{tikzpicture}
  \matrix (m) [matrix of math nodes,
    nodes in empty cells,nodes={minimum width=5ex,
    minimum height=5ex,outer sep=-5pt},
    column sep=1ex,row sep=1ex]{
                &         &          &         &           & \\
          2     &     \Z_2 &     &     &       &  \\
          1     &  \Z_2\times \Z_2 &   \Z_2\times \Z_2  &   &   &  \\
          0     &  U(1)  & \Z_2 &  0  &  \Z_2     &  \\
    \quad\strut &   0  &  1  &  2  &   3   & \strut \\};
\draw[thick] (m-1-1.east) -- (m-5-1.east) ;
\draw[thick] (m-5-1.north) -- (m-5-6.north) ;
\end{tikzpicture}
\end{center}
From $E_3^{1,1}=\Z_2\times \Z_2=E_\infty^{1,1}$ and $H^2(P_1;U(1))=
\Z_2\times \Z_2$ we deduce that the differential $\dd_3^{0,2}\colon
\Z_2 \longrightarrow \Z_2$ is an isomorphism. 
This implies that the symmetry corresponding to \eqref{Eq: Exact sequence D8 Z2} of the non-trivial topological action $\omega \in H^2(D_8;U(1))$ cannot be gauged using Theorem~\ref{Theorem: Gauging}, since the second obstruction corresponding to $\dd^{0,2}_3$ does not vanish. However, since the first obstruction vanishes we can gauge the symmetry using the relative field theory constructed in Section~\ref{Sec: Realisation on boundary}.   
\end{example}

\begin{example}
We have seen in Example~\ref{Example: Gauging1} that for the extension 
\begin{align}
(0,0) \longrightarrow \Z_2\times \Z_2 \longrightarrow \Z_4\times \Z_4
  \longrightarrow \Z_2\times \Z_2 \longrightarrow (0,0) 
\end{align}
the 2-cocycle $\omega_1 \in H^2(\Z_2\times \Z_2;U(1))$ cannot be obtained as the pullback of a 2-cocycle on $\Z_4\times \Z_4$. 
The corresponding 2-cochain is given by 
\begin{align}
\tilde{\omega}_1 \colon (\Z_4\times \Z_4)^2 & \longrightarrow U(1)\\
\big((a_1,b_1)\,,\, (a_2,b_2)\big) & \longmapsto \exp\Big(\pi \,\iu\,
                                     \Big\lfloor\frac{a_1}{2}\Big\rfloor\,
                                     \Big\lfloor\frac{b_2}{2}\Big\rfloor\Big)
                                     \ .
\end{align}
To find the corresponding obstructions we calculate using $\big\lfloor
\frac{a+b}{2} \big\rfloor = a\,b +\big\lfloor \frac{a}{2}
\big\rfloor+\big\lfloor \frac{b}{2} \big\rfloor$ mod~$2$ to get
\begin{align}
\delta \tilde{\omega}_1 \big((a_1,b_1)\,,\,
  (a_2,b_2)\,,\,(a_3,b_3)\big) = \exp\bigg(\pi \,\iu\, \Big(
  a_1\,a_2\,\Big\lfloor\frac{b_3}{2}\Big\rfloor
  +\Big\lfloor\frac{a_1}{2}\Big\rfloor \, b_2\,b_3 \Big) \bigg) \ .
\end{align}
Using the computer algebra program Maple~\cite{Maple} we verified 
by checking all possibilities that there are no solutions to the equation
\begin{align}
\delta (\tilde{\omega}_1-\gamma_1)=\lambda^*\theta
\end{align}
with $\gamma_1 \in F^1C^2(\Z_4\times\Z_4;U(1))$ and $\theta\in Z^3(\Z_2\times\Z_2;U(1))$. Hence the first obstruction $\dd^{0,2}_2\omega_1$ does not vanish. 
\end{example}

\section{Bulk-boundary correspondence}\label{Sec: Realisation on boundary}

In this section we realise anomalous gauged Dijkgraaf-Witten theories
on the boundaries of higher-dimensional Dijkgraaf-Witten theories. We
start by reviewing the general description of boundary field theories
and anomaly inflow through the bulk-boundary correspondence in the functorial framework. Afterwards we make the general theory explicit in the case of Dijkgraaf-Witten theories. 

\subsection{Anomaly inflow in functorial field theories}\label{Sec: Boundary theories}

We work with the geometric bicategory $G\text{-}\Cob_{n,n-1,n-2}$ for definiteness, but the description generalises to arbitrary background fields.
The description of boundary field theories and anomalies in the
framework of extended field theories is worked out
in~\cite{FreedAnomalies,RelativeQFT,
  MonnierHamiltionianAnomalies,Parity} (see also~\cite{Johnson-Freyd:2017ykw,BoundaryConditionsTQFT} for an
$(\infty,n)$-categorical discussion in the case of topological field
theories). This description is closely related to the twisted field
theories discussed in~\cite{Stolz:2011zj}. Here we follow~\cite[Section 3.3]{Parity}. 
Let $E\colon \EGCob \longrightarrow \Tvs$ be an extended field theory. We define its {truncation} $\text{tr}\, E$ to be the restriction of $E$ to the sub-bicategory of $\EGCob$ containing only invertible 2-morphisms.

\begin{definition}\label{Def: E. anomaly}
An $n{-}1$-dimensional $G$-equivariant anomalous field theory with anomaly
described by an invertible $n$-dimensional extended field theory $E\colon
\EGCob\longrightarrow \Tvs$ is a natural symmetric monoidal 2-transformation
\begin{align}
Z \colon 1 \Longrightarrow \tr\, E \ ,
\end{align} 
where $1$ is the trivial theory assigning $\fvs$ to every object, the identity functor to 1-morphisms and the identity natural transformation to 2-morphisms.
\end{definition} 
\begin{remark}
There are different definitions for natural symmetric monoidal 2-transformations
corresponding to different levels of strictness. Here we use
\cite[Definition~B.13]{Parity}, which seems to be best suited for
physical applications.
\end{remark}
Concretely $Z \colon 1 \Longrightarrow \tr\, E$ consists of a linear functor $Z(S, \xi)\colon \fvs \longrightarrow E(S,\xi) $ for all objects $(S,\xi)\in \EGCob$, and a natural transformation $Z(\Sigma, \varphi)\colon E(\Sigma, \varphi)\circ Z(S_1,\xi_1) \Longrightarrow Z(S_2,\xi_2) $ for all 1-morphisms $\big((\Sigma,\varphi) \colon (S_1,\xi_1) \longrightarrow (S_2,\xi_2)\big) \in \EGCob$.
The functor $Z(S,\xi)\colon \fvs \longrightarrow E(S,\xi)$ can be
described by an object $Z(S,\xi)[\C]\in E(S,\xi)$ which by a slight
abuse of notation we denote again by $Z(S,\xi)$. The natural
transformation $Z(\Sigma,\varphi)\colon E(\Sigma,\varphi) \circ Z(S_1,
\xi_1) \Longrightarrow Z(S_2, \xi_2)$ can be described by a morphism
$Z(\Sigma,\varphi)\colon
E(\Sigma,\varphi)[Z(S_1,\xi_1)]\longrightarrow Z(S_2,\xi_2)$ in
$E(S_2,\xi_2)$. Requiring $Z$ to be a natural 2-transformation reduces explicitly
to the following:
Let $(S,\xi)$, $(S_1,\xi_1)$, $(S_2,\xi_2)$ and $(S_3,\xi_3)$ be
objects of $\EGCob$, and $(\Sigma_a,\varphi_a)\colon (S_1,\xi_1)
\longrightarrow (S_2, \xi_2)$ and $(\Sigma_b,\varphi_b) \colon (S_2,
\xi_2) \longrightarrow (S_3,\xi_3)$ 1-morphisms in $\EGCob$. Then the diagrams
\begin{equation}\label{Condition twisted functoriality}
\begin{tikzcd}
E(\Sigma_b,\varphi_b)\circ E(\Sigma_a,\varphi_a)[Z(S_1,\xi_1)] \ar[dd,
"{E(\Sigma_b,\varphi_b)[Z(\Sigma_a,\varphi_a)]}",swap] \ar[r] &
E(\Sigma_b \circ \Sigma_a, \varphi_b \cup \varphi_a)[Z(S_1,\xi_1)]
\ar[dd, "{Z(\Sigma_b \circ \Sigma_a, \varphi_b \cup \varphi_a)}"] \\ & \\
E(\Sigma_b, \varphi_b)[Z(S_2,\xi_2)] \ar[r, swap,"{Z(\Sigma_b, \varphi_b)}"] & Z(S_3,\xi_3)
\end{tikzcd} 
\end{equation}
and
\begin{equation}\label{Condition twisted preservation of identities}
\begin{tikzcd}[column sep=tiny]
Z(S,\xi)\ar[rr] \ar[rd, "\id", swap] &  & E(\id_{(S,\xi)})[Z(S,\xi)] \ar[ld, "{Z(\id_{(S,\xi)})}"] \\
 &\ \ \ Z(S,\xi)
\end{tikzcd}
\end{equation}
commute, where the unlabelled morphisms are part of the structure of the extended field theory $E$. 
The symmetric monoidal structure on $Z$ is described explicitly by specifying natural morphisms 
\[ 
M^{-1}\colon Z(\varnothing) \longrightarrow \iota_E (\C)  
\] 
in $E (\varnothing)$ and 
\[
\Piit_{Z}\big((S_1,\xi_1)\,,\,(S_2,\xi_2\big)\colon \chi_{E}[ Z(S_1,\xi_1)\boxtimes Z(S_2,\xi_2)] \longrightarrow Z(S_1 \sqcup S_2,\xi_1\sqcup\xi_2) \ 
\]
in $E(S_1 \sqcup S_2,\xi_1\sqcup\xi_2)$, where $\iota_E\colon \fvs \longrightarrow
E(\varnothing)$ and $\chi_{E}\colon E(\,\cdot\,)\boxtimes
E(\,\cdot\,)\Longrightarrow E(\,\cdot\, \sqcup \,\cdot\,)$ are part of the
structure corresponding to a symmetric monoidal
2-functor~\cite[Definition B.12]{Parity}. There is a long list of
coherence and compatibility conditions that these morphisms have to satisfy, see~\cite[Proposition~3.14]{Parity} for details. 

\begin{remark}
If the extended field theory $E$ is trivial, it follows from this discussion that a field theory relative to $E$ is the same as an ordinary $n{-}1$-dimensional field theory.
\end{remark}
It is instructive to decategorify Definition \ref{Def: E. anomaly}~\cite[Section 2]{Parity}.
We denote by $\tr\, \GCob$ the maximal subgroupoid of $\GCob$. For a functor $L\colon \GCob \longrightarrow \fvs$ we denote by $\tr\, L \colon \tr\, \GCob \longrightarrow \fvs$ its restriction to $\tr\, \GCob$.
\begin{definition}\label{Def: nE. anomaly}
A {partition function} $Z$ with anomaly described by an invertible
$n$-dimensional field theory
$L\colon \GCob \longrightarrow \fvs$ is a natural symmetric monoidal transformation 
\begin{align}
Z\colon 1 \Longrightarrow \tr\, L \ . 
\end{align}
\end{definition}  

\begin{remark}
Restricting Definition~\ref{Def: E. anomaly} to the endomorphisms of
$\varnothing$ induces a natural transformation $\tr\, L \Longrightarrow
1$. Since in most physically relevant examples all vector spaces are
Hilbert spaces, this discrepancy can be resolved by taking the
adjoint. Here we stick to Definition~\ref{Def: nE. anomaly}, because it has a natural geometric interpretation in terms of line bundles.     
\end{remark}
Unpacking Definition~\ref{Def: nE. anomaly}, we get for every object $(\Sigma,\varphi)\in \GCob$ a linear map $Z(\Sigma,\varphi)\colon \C
 \longrightarrow L(\Sigma,\varphi) $, such that for all invertible
 morphisms $\phi \colon (\Sigma_a,\varphi_a) \longrightarrow
 (\Sigma_b,\varphi_b) $ in $ \GCob$ the diagram
\begin{equation}
\begin{tikzcd}
\C \ar[d, swap, "{\id}"] \ar{rrr}{Z(\Sigma_a,\varphi_a)} & & & L(\Sigma_a,\varphi_a) \ar{d}{L(\phi)} \\
 \C \ar[rrr, swap, "{Z(\Sigma_b,\varphi_b)}"] & & & L(\Sigma_b,\varphi_b)
\end{tikzcd}
\label{Equation: Decategorified case}
\end{equation} 
commutes. We think of $\phi$ as a symmetry of the classical background fields. Since
$L$ is an invertible field theory, $L(\Sigma, \varphi)$ is
a one-dimensional vector space isomorphic to $\C$, though not
necessarily in a canonical way.
Picking such an isomorphism for all $(\Sigma, \varphi)$ induces a $\C^\times$-valued 1-cocycle of the groupoid of symmetries $\tr\, \GCob$,
since we can then identify the linear map $L(\phi)$ with a non-zero complex number.    

Definition~\ref{Def: nE. anomaly} naturally encodes properties of field theories with anomalies.
For example, in physically relevant theories one should also require that the vector spaces $L(\Sigma, \,\cdot\,)$ form
a line bundle over the space of field configurations as in the case of smooth field theories \cite{Stolz:2011zj}. 
In the case of $\GCob$ the space of field configurations is a discrete groupoid and every invertible functorial field theory gives a flat line bundle over the groupoid of field configurations
in the sense of \cite{TwistedDWandGerbs}. 
The partition function is now a parallel section of this line bundle. 
This reproduces the more geometric description of anomalies as the non-triviality of a line bundle over the space
of field configurations~\cite{NashBook}. 

In the extended framework of Definition \ref{Def: E. anomaly}, evaluating $E$ on manifolds of dimension $n-2$ can be regarded as a 2-line bundle over the 
groupoid of field configurations as defined in
e.g.~\cite{SWParallelSections}. In particular, an anomalous field theory defines a section of this 2-line bundle. This reproduces the description of Hamiltonian anomalies in terms of line bundle gerbes~\cite{Carey1997}.
The extended functorial framework thus naturally combines the Hamiltonian and Lagrangian descriptions of anomalies. 
It further formulates the field theory on all possible spacetimes
simultaneously, and also requires compatibility with gluing and
cutting of manifolds. Hence it is a mathematically concrete formulation
of the requirements from~\cite{WittenString15} of local quantum field theory. 

Of particular interest in the Hamiltonian description of anomalies is the projective representation of the symmetry group
on the Hilbert space of the field theory. 
To describe these in the functorial framework 
we briefly recall the 2-categorical description of projective representations. 
An element of $H^2(BG; \C^\times)$ 
corresponds to a homotopy class of maps from $BG$ 
to the Eilenberg-MacLane space $K(\C^\times ,2)$. The classifying space
$BG$ is the nerve of a category, namely the action groupoid $* \DS G$. The
Eilenberg-MacLane space $K(\C^\times ,2)$ 
is the nerve of a 2-category $B^2\C^\times $ with one object, 
one 1-morphism and the group $\C^\times$ as the set of 2-morphisms. 
Hence we can describe 2-cocycles alternatively as 2-functors 
up to natural isomorphisms.

\begin{definition}
Let $\Gscr$ be a groupoid and $A$ an abelian group. A {2-cocycle on $\Gscr$ with values in $A$} is a 2-functor 
\begin{align}
\alpha \colon \underline{\Gscr} \longrightarrow B^2 A \ .
\end{align} 
\end{definition}  

\begin{remark}
Spelling out the coherence isomorphisms corresponding to a 2-functor reproduces the usual definition
of groupoid cohomology~\cite[Section 3.4]{Parity}. 
\end{remark}

There is a natural 2-functor $B^2\C^\times \hookrightarrow \Tvs$ sending the unique object to the category of 
vector spaces, the unique 1-morphism to the identity functor and a non-zero complex number $\lambda$ to the natural transformation of the
identity functor induced by scalar multiplication with $\lambda$. 
Using this embedding a 2-cocycle $\alpha$ on $\Gscr$ with values in
$\C^\times$ induces a 2-functor $\alpha \colon \underline{\Gscr} \longrightarrow \Tvs$. Using this 
2-functor we can give a categorical definition of projective representations (see \cite{BoundaryConditionsTQFT} for the higher categorical framework).

\begin{definition}\label{Def: Projective representation}
A {projective representation of a groupoid $\Gscr$ twisted by a 2-cocycle $\alpha$ with values in $\C^\times$} is a natural
2-transformation 
\begin{equation}
\begin{tikzcd} & \  & \\
\underline{\Gscr} \ar[rr, "\alpha" ,bend left=50] \ar[rr,"1" ,bend right=50, swap]& & \Tvs \\
&\  \ar[uu,Rightarrow ,"\,\rho", swap, shorten <= 15, shorten >= 15] &  
\end{tikzcd}
\end{equation}
\end{definition}

\begin{remark}
Spelling out the definition reproduces the usual definition
of a projective representation~\cite[Section 3.4]{Parity}. 
\end{remark}      
The assumption that the anomaly field theory $E$ is invertible implies
that for every $n{-}2$-dimensional manifold $S$ there exists a
2-cocycle $\alpha \colon \underline{\BunG(S)}\longrightarrow B^2 \C^\times$ such
that the diagram
\begin{equation}
\begin{tikzcd}
\underline{\BunG(S)} \ar[rd, swap,"\alpha"]\ar[rr, "E"]& & \Tvs \\
 & B^2\C^\times \ar[ru]&
\end{tikzcd}
\end{equation}
commutes up to a non-canonical isomorphism. A quantum field theory with anomaly $Z$ induces a projective representation
of $\BunG(S)$ with respect to the 2-cocycle $\alpha$. Morphisms in
$\BunG(S)$ are gauge transformations, which are symmetries of the
field theory. Hence we shall
sometimes refer to $\BunG(S)$ as the symmetry groupoid. 
Projective representations corresponding to 
different choices of 
$\alpha \colon \underline{\BunG(S)}\longrightarrow B^2 \C^\times$ are equivalent. 
The projective representation can be concretely 
constructed by picking trivialisations of the categories that $E$ assigns to field configurations. 
The 2-cocycle twisting the projective representation is completely determined by $E$ and independent of the explicit form of~$Z$.

We shall now discuss in more detail how to couple the bulk field
theory $(E,L)$ and boundary field theory $Z$ to construct an anomaly-free theory. 
We start with the unextended framework corresponding to Definition~\ref{Def: nE. anomaly}.
This involves the full quantum field theory $L\colon \GCob \longrightarrow \fvs$ and not just its truncation. 
Let $M$ be an $n$-dimensional manifold 
with boundary $\partial M = -\Sigma $, and $\psi \colon M \longrightarrow BG$ a $G$-bundle. 
The field theory with anomaly defines an element $Z(\Sigma,\psi|_\Sigma)\in L(\Sigma, \psi|_\Sigma)$. 
We can interpret
$(M,\psi)$ as a morphism $(M,\psi) \colon (\Sigma, \psi|_\Sigma) \longrightarrow  \varnothing$ in $\GCob$. The partition function of the composite system can now be defined as
\begin{align}\label{Def: Combined partition function}
Z_{\rm bb}(M,\psi, \Sigma)= L(M,\psi)[Z(\Sigma,\psi|_\Sigma)] \ \in \ L(\varnothing)\cong \C \ .
\end{align}     
This definition does not depend on any additional choices. Let $\psi'\colon M \longrightarrow BG$ be a principal $G$-bundle and $\nu\colon \psi \longrightarrow \psi'$ a gauge transformation.
We then calculate
\begin{align}
Z_{\rm bb}(M,\psi',\Sigma )&= L(M,\psi'\,)[Z(\Sigma,\psi'|_\Sigma)] \\[4pt]
&= L(M, \psi'\,)\circ L([0,1]\times \Sigma, \nu|_{\Sigma})[Z(\Sigma,\psi|_{\Sigma})]\\[4pt]
&=L(M, \psi)[Z(\Sigma,\psi|_{\Sigma})]\\[4pt]
&=Z_{\rm bb}(M,\psi, \Sigma)\ , 
\end{align}
where we used~\eqref{Equation: Decategorified case} in the second
equality and in the third equality the fact that $L$ is invariant
under gauge transformations relative to the boundary. 
This shows that the composite partition function is gauge-invariant. 

Definition \ref{Def: E. anomaly} also allows us to formulate the
composite system at the level of state spaces. Let $E\colon \EGCob
\longrightarrow \Tvs$ be an invertible extended field theory. Consider
an $n{-}1$-dimensional manifold $\Sigma$ with boundary $\partial
\Sigma = -S$ and a principal $G$-bundle $\varphi\colon \Sigma
\longrightarrow BG$. The anomalous field theory $Z$ defines an element $Z(S,\varphi|_S)\in E(S,\varphi|_S)$. The composite state space is given by
\begin{align}
Z_{\rm bb}(\Sigma,\varphi , S)= E(\Sigma,\varphi )[Z(S,\varphi|_S)] \
  \in \ E(\varnothing) \cong \fvs \ . \label{Eq: Def combined state space}
\end{align} 
This vector space does not depend on any additional choices. Let $\nu
\colon \varphi \longrightarrow \varphi'$ be a gauge
transformation. Then there is an induced linear map
\begin{align}
Z_{\rm bb}(\Sigma,\varphi , S)= E(\Sigma,\varphi )[Z(S,\varphi|_S)]&\xrightarrow{E(\nu)} E\big((\Sigma, \varphi'\,)\circ ([0,1]\times S,\nu|_S)\big)[Z(S,\varphi|_S)]\\[4pt]
&\xrightarrow{\hspace{0.65cm}} E(\Sigma, \varphi'\,)\circ E([0,1]\times S,\nu|_S)[Z(S,\varphi|_S)]\\[4pt]
&\xrightarrow{Z([0,1]\times S,\nu|_S)} E(\Sigma, \varphi'\,)[Z(S,
  \varphi'|_S)]=Z_{\rm bb}(\Sigma,\varphi' , S) \ . \label{Eq: Def action combined state space}
\end{align}
It follows from the coherence conditions that this defines an honest representation of the symmetry groupoid. Hence we have described a way of coupling bulk and boundary degrees of freedom to an anomaly-free state space. In condensed matter physics applications the invertible field theory $E$ arises as the low-energy effective theory of the bulk system. 

\subsection{Anomalous Dijkgraaf-Witten theories as boundary states}\label{Sec: State space}

We will now illuminate this construction in the case of anomalies of Dijkgraaf-Witten theories.  
 Let $Z_\omega\colon \Cob_{n-1}\longrightarrow \fvs$ be a Dijkgraaf-Witten 
theory with gauge group $D$, topological action 
$\omega \in Z^{n-1}(BD;U(1))$ and kinematical $G$-symmetry described by an extension
\begin{align}
1\longrightarrow D \overset{\iota}{\longrightarrow} \widehat{G}\overset{\lambda}{\longrightarrow} G \longrightarrow 1 \ .
\end{align}
Recall from Proposition~\ref{Prop: Obstructions} that when the first $n-2$ 
obstructions in the spectral sequence vanish,
a cochain $\omega' \in C^{n-1}(B\widehat{G}; U(1))$ and a cocycle 
$\theta \in Z^{n}(BG; U(1))$ exist such that $\omega= \iota^* \omega'$ and 
$\delta \omega' = \lambda^* \theta$. 
We realise the anomalous gauged theory as a field theory $Z_{\omega'}$ living on 
the boundary of a classical $n$-dimensional Dijkgraaf-Witten theory 
$E_\theta \colon \EGCob \longrightarrow \Tvs$ with 
topological action $\theta$.   

\subsubsection*{Partition function}

We shall first construct the state as an unextended relative field theory, i.e. a natural transformation $Z_{\omega'} \colon 1 \Longrightarrow \text{tr}\, L_\theta$ (see Definition~\ref{Def: nE. anomaly}).
The following construction is similar to the one 
in~\cite[Section~3.3]{Witten:2016cio}. However, we use the language of functorial field theories and homotopy fibres to describe the construction. The approach to boundary field theories in the present
paper is to some extent the reverse of the approach in 
\cite{Witten:2016cio}, where anomalous boundary field theories are
constructed starting from a bulk
Dijkgraaf-Witten theory. 
Instead we start from a field theory with anomaly and show how 
to realize this theory as a boundary field theory.
The construction of the state space below
is not given in~\cite{Witten:2016cio}. 

Following the general theory outlined in Section~\ref{Sec: Boundary theories} we have to specify an element $Z_{\omega'}(\Sigma, \varphi\colon M \longrightarrow BG)$ of $L_\theta(\Sigma,\varphi)$ for all objects $(\Sigma,\varphi)\in \GCob$. Let $\sigma_\Sigma$ be a representative for the fundamental class of $\Sigma$.
We set
\begin{align}\label{Eq: Partition function relative field theory}
Z_{\omega'}(\Sigma,\varphi)= \Big(\int_{(\widehat{\varphi},h)\in \lambda_*^{-1}[\varphi]} \, \langle {\widehat{\varphi}}^{\,*} \omega', \sigma_\Sigma \rangle \, \langle  h^* \theta, [0,1] \times \sigma_\Sigma \rangle\Big) \ [\sigma_\Sigma] \ \in \ L_\theta(\Sigma,\varphi) \ .
\end{align}
\begin{proposition}
$Z_{\omega'}$ is a partition function with anomaly
$L_\theta:\GCob\longrightarrow \fvs$.
\end{proposition}     
\begin{proof}
We have to show that $Z_{\omega'}$ is a well-defined natural transformation. This is an immediate consequence of Theorem~\ref{Theorem: Relative field theory} below. To become acquainted with the constructions involved, we present here part of the proof.
We start by showing that $\langle  {\widehat{\varphi}}\,^* \omega', \sigma_\Sigma \rangle \, \langle  h^* \theta, [0,1] \times \sigma_\Sigma \rangle$ is well-defined on isomorphism classes of $\lambda_*^{-1}[\varphi]$.
Let $\widehat{h}\colon (\widehat{\varphi}_1,h_1) \longrightarrow (\widehat{\varphi}_2,h_2)$ be a morphism in $\lambda_*^{-1}[\varphi]$,
i.e.~a homotopy $\widehat{h}\colon \widehat{\varphi}_1 \longrightarrow \widehat{\varphi}_2$ such that the diagram
\begin{equation}
\begin{tikzcd}
\lambda_* \widehat{\varphi}_1 \ar[rd,"h_1", swap] \ar[rr,"\lambda_* \widehat{h}"] & & \lambda_*\widehat{\varphi}_2 \ar[ld,"h_2"] \\
 & \varphi & 
\end{tikzcd}
\end{equation}
commutes.  The homotopy induces a chain homotopy $H\colon \widehat{\varphi}_{1*} \longrightarrow \widehat{\varphi}_{2*}$ between the induced maps on singular chains given by $H(c)= \widehat{h}_* ( [0,1] \times c)$ for all chains $c \in C_\bullet (\Sigma)$. Hence, writing $U(1)=\R/\Z$ additively for the calculation, we find
\begin{align}
\big\langle {\widehat{\varphi}_2}^{\,*} \omega', \sigma_\Sigma \big\rangle - \big\langle  {\widehat{\varphi}_1}^{\,*} \omega', \sigma_\Sigma \big\rangle &= \big\langle \omega', \partial H(\sigma_\Sigma) - H (\partial \sigma_\Sigma) \big\rangle\\[4pt]
& = \big\langle \omega' , \partial H (\sigma_\Sigma) \rangle\\[4pt]
& = \big\langle \, \widehat{h}^* \lambda^* \theta, [0,1]  \times \sigma_\Sigma \big\rangle \\[4pt]
&= \big\langle  h_1^* \theta - h_2^* \theta, [0,1] \times \sigma_\Sigma \big\rangle \ .
\end{align} 
This shows that the integration in \eqref{Eq: Partition function relative field theory} is well-defined. 

Let $\sigma_\Sigma'$ be a different representative for the fundamental class of $\Sigma$ and $\chi$ an $n$-chain satisfying $\partial \chi = \sigma_\Sigma' - \sigma_\Sigma $.
To show that \eqref{Eq: Partition function relative field theory} is an element of $L_\theta(\Sigma, \varphi)$ we calculate 
\begin{align}
\langle {\widehat{\varphi}}^{\,*} \omega',\sigma_\Sigma' - \sigma_\Sigma \rangle \, \langle {h}^* \theta, [0,1] \times (\sigma_\Sigma' - \sigma_\Sigma) \rangle &= \langle  {\widehat{\varphi}}^{\,*} \omega',\partial \chi \rangle \, \langle {h}^* \theta, [0,1] \times \partial \chi \rangle \\[4pt]
& = \langle {\widehat{\varphi}}^{\,*} \lambda^* \theta, \chi  \rangle \, \langle {h}^* \theta, -\lbrace 0 \rbrace \times  \chi + \lbrace 1 \rbrace \times \chi   \rangle \\[4pt]
&= \langle  \varphi^* \theta, \chi    \rangle \ .
\end{align}
This is exactly the required transformation behaviour. We leave the verification of naturality to the reader. 
\end{proof}

\begin{remark}
Before extending the field theory we give the precise form of the composite partition function \eqref{Def: Combined partition function}. We fix an $n$-dimensional manifold $M$ with boundary $\partial M = -\Sigma$ and a principal $G$-bundle $\psi \colon M \longrightarrow BG$. Evaluating $L_\theta$ on $(M,\psi)$ gives a linear map $L_\theta(M,\psi)\colon L_\theta(\Sigma,\psi|_\Sigma)\longrightarrow \C$. The composite partition function is then
\begin{align}
Z_{\omega'\,{\rm bb}}(M,\psi, \Sigma)&= L_\theta(M,
                                       \psi)[Z_{\omega'}(\Sigma,\psi|_\Sigma)]
  \\[4pt] &= \Big(\int_{(\widehat{\varphi},h)\in
            \lambda_*^{-1}[\psi|_\Sigma]} \, \langle
            {\widehat{\varphi}}^{\,*} \omega', \partial \sigma_M
            \rangle^{-1} \, \langle  h^* \theta, [0,1] \times \partial
            \sigma_M \rangle^{-1}\Big) \ \langle \psi^* \theta ,
            \sigma_M \rangle \ ,
\end{align}  
which is gauge-invariant according to the 
general theory outlined in Section~\ref{Sec: Boundary theories}. 
\end{remark}

\subsubsection*{State space}

Let $(S,\xi)$ be an object of $\EGCob$ and denote by
$\Fund_\theta(S,\xi)$ the full subcategory of simple objects of
$E_\theta(S,\xi)$. Recall from Section~\ref{Sec: Calssical gauge
  theory} that objects of $\Fund_\theta(S,\xi)$ are representatives $\sigma_S$ of the fundamental class of $S$ and morphism spaces are given by 
\[
\Hom_{E_\theta(S,\xi)}(\sigma_S,\sigma_S')= \C\big[\lbrace \Lambda \in C_{n-1}(S) \mid \partial \Lambda = \sigma_S'-\sigma_S \rbrace\big] \, \big/ \sim_\theta \ .
\]
We construct a functor 
\begin{align}
\tilde{Z}_{\omega'}^{(S,\xi)}\colon \Fund_\theta(S,\xi)^{\text{op}} \longrightarrow \fvs \ .
\end{align} 
To this end we first construct a functor $L_{\xi, \omega'}(\sigma_S)
\colon \lambda_*^{-1}[\xi] \longrightarrow \fvs$ as follows: To every
object $(\widehat{\xi},h)$ we assign $\C$ and to a homotopy $\widehat{h}\colon (\widehat\xi,h)\longrightarrow (\widehat\xi\,',h')$ we assign the complex number $\langle\,\widehat{h}^* \omega', [0,1] \times \sigma_S \rangle$. 
Let $\sigma_S$ and $\sigma_S'$ be representatives for the fundamental class of $S$ and $\Lambda \in C_{n-1}(S) $ an $n{-}1$-chain satisfying $\partial \Lambda = \sigma_S'-\sigma_S$, i.e. a morphism in $\Fund_\theta(S,\xi)$. 
We construct a natural transformation
\begin{align}
\label{Def: L natural}
L_{\xi, \omega'}(\Lambda) \colon L_{\xi, \omega'}(\sigma_S') &\Longrightarrow L_{\xi, \omega'}(\sigma_S) \\
{L_{\xi, \omega'}}(\Lambda)_{ (\widehat\xi,h)} \colon \C
                                                             &\longrightarrow
                                                               \C \ ,
                                                               \quad 1
                                                               \longmapsto
                                                               \big\langle
                                                               {\widehat{\xi}}\,^*
                                                               \omega',\Lambda
                                                               \big\rangle^{-1}
                                                               \,
                                                               \big\langle
                                                               h^*
                                                               \theta
                                                               , [0,1]
                                                               \times
                                                               \Lambda
                                                               \big\rangle^{-1}
                                                               \ .
\end{align}
We show in Appendix~A that $L_{\xi, \omega'} \colon \Fund_\theta(S,\xi)^{\text{op}} \longrightarrow \big[\lambda_*^{-1}[\xi], \fvs\big]$
is a well-defined functor to the category $\big[\lambda_*^{-1}[\xi], \fvs\big]$ of functors from $\lambda_*^{-1}[\xi]$ to $\fvs$ (Lemma~\ref{Lemma well-defined on morphisms}).
Composition with the limit functor $\int_{\lambda_*^{-1}[\xi]} \colon
\big[\lambda_*^{-1}[\xi], \fvs\big] \longrightarrow \fvs$ then constructs the desired functor
\begin{align} 
\tilde{Z}^{(S,\xi)}_{\omega'}= \int_{\lambda_*^{-1}[\xi]} \, L_{\xi, \omega'} \ .
\end{align}

The limit can again be realised by parallel sections.
This allows us to define 
\begin{align}\label{Def on Objects}
Z_{\omega'}(S,\xi) = \int^{\sigma\in \Fund_\theta(S,\xi)} \,
  \tilde{Z}^{(S,\xi)}_{\omega'}(\sigma)*\sigma = \int^{\sigma\in
  \Fund_\theta(S,\xi)} \ \int_{\lambda_*^{-1}[\xi]} \, L_{\xi,
  \omega'}(\sigma)*\sigma \ \in \ E_\theta(S,\xi) \ . 
\end{align}
Let $(\Sigma, \varphi) \colon (S_1,\xi_1) \longrightarrow (S_2,\xi_2)$ be a 1-morphism in $\EGCob$. We construct a morphism 
\begin{align}
Z_{\omega'}(\Sigma,\varphi) \colon E_\theta(\Sigma,\varphi)[Z_{\omega'}(S_1,\xi_1)] & = 
 \int^{\sigma_2\in \Fund_\theta(S_2,\xi_2)} \Big(\int^{\sigma_1\in
                                                                                      \Fund_\theta(S_1,\xi_1)}
                                                                                      \,
                                                                                      \Sigma^\varphi(\sigma_2,\sigma_1)\otimes_\C
                                                                                      \tilde{Z}^{(S_1,\xi_1)}_{\omega'}(\sigma_1)\Big)*\sigma_2\\
                                                                                    &
                                                                                      \longrightarrow \int^{\sigma_2\in \Fund_\theta(S_2,\xi_2)} \, \tilde{Z}^{(S_2,\xi_2)}_{\omega'}(\sigma_2)*\sigma_2 
= Z_{\omega'}(S_2,\xi_2)
\end{align}
in $E_\theta(S_2,\xi_2)$ from the universal property of the coend by realising $\tilde{Z}^{(S_2,\xi_2)}_{\omega'}(\sigma_2)$ as a cowedge. We define the required linear maps for the concrete description of the limit as parallel sections by
\begin{align}
Z_{\omega'}(\Sigma,\varphi)_{\sigma_1} \colon \Sigma^\varphi(\sigma_2,\sigma_1)\otimes_\C \tilde Z_{\omega'}^{(S_1,\xi_1)}(\sigma_1) & \longrightarrow \tilde Z_{\omega'}^{(S_2,\xi_2)}(\sigma_2) \\
\Lambda \otimes_\C f( \,\cdot\, ) & \longmapsto
                                    Z_{\omega'}(\Sigma,\varphi)_{\sigma_1}(f,\Lambda)(\,\cdot\,)
                                    \ ,  
\end{align}
with 
\begin{align}
Z_{\omega'}(\Sigma,\varphi)_{\sigma_1}(f,\Lambda
  )\big(\widehat{\xi}_2,h_2\big) &=
                           \int_{(\widehat{\varphi},g,\widehat{h}\,)\in
                           \lambda_*^{-1}[\varphi]|_{(\widehat{\xi}_2,h_2)}}
                           \, \big\langle \widehat{\varphi}^{\,*} \omega'
                           , \Lambda \big\rangle \, \big\langle g^* \theta ,
                           [0,1] \times \Lambda \big\rangle \, \big\langle\,\widehat{h}^* \omega' , [0,1] \times
  \sigma_2 \big\rangle \\ &
  \hspace{7cm} \cdot f(\widehat{\varphi}\,|_{S_1},g|_{S_1}) \ \in \ \C
  \ .
\label{Def: Action on invariant functions} 
\end{align}
The domain of integration here is the groupoid with objects consisting of triples of a map
\begin{align}
\widehat{\varphi} \colon \Sigma \longrightarrow B\widehat{G} \ ,
\end{align}
a gauge transformation
\begin{align}
g \colon \lambda_* \widehat{\varphi} \longrightarrow \varphi \ ,
\end{align}
and a gauge transformation
\begin{align}
\widehat{h}\colon \widehat{\varphi}\,|_{S_2} \longrightarrow \widehat{\xi}_2
\end{align}
such that the diagram
\begin{equation}
\begin{tikzcd}
\lambda_*\widehat{\varphi}\,|_{S_2} \ar[rr, "\lambda_*\widehat{h}"] \ar[rd,"g|_{S_2}",swap] & & \lambda_*\widehat{\xi}_2 \ar[ld, "h_2"] \\
 & \xi_2 &
\end{tikzcd}
\end{equation}
commutes. We show in Appendix~A that this induces the 
desired morphism (Lemma~\ref{Lem: Def on Mor}).

We can now formulate and prove the main result of this section.

\begin{theorem}\label{Theorem: Relative field theory}
Let $Z_\omega \colon \Cob_{n-1}\longrightarrow \fvs $ be an $n{-}1$-dimensional discrete gauge theory with gauge group $D$, topological action $\omega \in Z^{n-1}(BD; U(1))$ and kinematical $G$-symmetry described by an extension
\begin{align}
1\longrightarrow D \overset{\iota}{\longrightarrow} \widehat{G}\overset{\lambda}{\longrightarrow} G \longrightarrow 1 \ .
\end{align}
Let $\omega'$ be an $n{-}1$-chain on $B\widehat{G}$ and $\theta\in Z^n(BG;U(1))$ an $n$-cocycle on $BG$ satisfying $\iota^*\omega'=\omega$ and $\delta \omega'= \lambda^* \theta$.
Then $Z_{\omega'}$ defined in~\eqref{Def on Objects} and \eqref{Def: Action on invariant functions} is an anomalous field theory with anomaly $E_\theta\colon \EGCob \longrightarrow \Tvs$. 
\end{theorem}

\begin{proof}
Let $(S_1,\xi_1)$ and $(S_2,\xi_2)$ be objects of $\EGCob$.
We first construct the missing structure corresponding to the compatibility with the monoidal structures. 
This involves a morphism 
\[ 
M^{-1}\colon Z_{\omega '}(\varnothing)= \sigma_\varnothing \longrightarrow \sigma_\varnothing  
\] 
in $E_\theta (\varnothing)$ which we choose to be the identity, and a natural isomorphism 
\begin{align}\label{Eq: Monoidal structure}
\Piit_{Z_{\omega '}}\big((S_1,\xi_1)\,,\,(S_2,\xi_2)\big)\colon \chi_{E_\theta} [ Z_{\omega '}(S_1,\xi_1)\boxtimes Z_{\omega '}(S_2,\xi_2)] \longrightarrow Z_{\omega '}(S_1 \sqcup S_2,\xi_1 \sqcup \xi_2) \ .
\end{align}
Spelling out \eqref{Def on Objects} we find that $\chi_{E_\theta}[ Z_{\omega '}(S_1,\xi_1)\boxtimes Z_{\omega '}(S_2,\xi_2)]$ is given explicitly by
\begin{align}
& \int^{\sigma_2 \in \Fund_\theta(S_2,\xi_2)} \ \int^{\sigma_1 \in \Fund_\theta(S_1,\xi_1)} \, \Big( \tilde{Z}^{(S_2,\xi_2)}_{\omega'}(\sigma_2) \otimes_\C \tilde{Z}^{(S_1,\xi_1)}_{\omega'}(\sigma_1) \Big) * [\sigma_{1}\sqcup \sigma_{2}] \\ & \hspace{4cm}
\cong \int^{\sigma_1 \sqcup \sigma_2 \in \Fund_\theta(S_1 \sqcup S_2,\xi_1 \sqcup \xi_2)} \, \Big( \tilde{Z}^{(S_2,\xi_2)}_{\omega'}(\sigma_2) \otimes_\C \tilde{Z}_{\omega'}^{(S_1,\xi_1)}(\sigma_1) \Big) * [\sigma_{1}\sqcup \sigma_{2}] \ ,
\end{align}
where we used Fubini's Theorem for coends together with the fact that we can naturally identify $\Fund_\theta(S_1,\xi_1) \times \Fund_\theta(S_2,\xi_2)$ with $\Fund_\theta(S_1 \sqcup S_2,\xi_1 \sqcup \xi_2)$. Concretely, $Z_{\omega '}(S_1 \sqcup S_2,\xi_1 \sqcup \xi_2)$ is given by
\begin{align}
\int^{\sigma_1 \sqcup \sigma_2 \in \Fund_\theta(S_1 \sqcup S_2,\xi_1 \sqcup \xi_2)}  \, \tilde{Z}^{(S_2 \sqcup S_1,\xi_2 \sqcup \xi_1)}_{\omega'}(\sigma_2 \sqcup \sigma_1)  * [\sigma_{1}\sqcup \sigma_{2}] \ .
\end{align} 
The isomorphism \eqref{Eq: Monoidal structure} is induced by the collection of natural linear isomorphisms 
\begin{align}
\tilde{Z}_{\omega'}^{(S_2,\xi_2)}(\sigma_2) \otimes_\C \tilde{Z}_{\omega'}^{(S_1,\xi_1)}(\sigma_1) \longrightarrow \tilde{Z}^{(S_2 \sqcup S_1,\xi_2 \sqcup \xi_1)}_{\omega'}(\sigma_2 \sqcup \sigma_1) 
\end{align}
given on parallel sections by
\begin{align}
f_2(\,\cdot\,) \otimes_\C f_1(\,\cdot\,) \longmapsto f_2(\,\cdot\, |_{S_2}) \, f_1(\,\cdot\, |_{S_1}) \ .
\end{align} 
A straightforward but tedious calculation shows that these definitions satisfy \cite[eqs.~(3.18)--(3.22)]{Parity}.

Next we show that $Z_{\omega '}$ is compatible with identities, i.e. that the diagram \eqref{Condition twisted preservation of identities} commutes. Since all constructions are natural it is enough to check this on the terms appearing in the coends of \eqref{Def on Objects} and \eqref{eqnnatisocohuni}. Explicitly, we have to show that the diagram
\begin{equation}
\begin{tikzcd}
\tilde{Z}^{(S,\xi)}_{\omega'}(\sigma) \ar[rr, "{(-1)^{n}\,\sigma \times [0,1] \otimes_\C \,\cdot\,} "] \ar[rd, "\id", swap] & & (S\times [0,1])^{\xi \times \id_{[0,1]}}(\sigma,\sigma) \otimes_\C \tilde{Z}^{(S,\xi)}_{\omega'}(\sigma) \ar[ld] \\
 & \tilde{Z}^{(S,\xi)}_{\omega'}(\sigma) &
\end{tikzcd}
\end{equation} 
commutes for all objects $(S,\xi) \in \EGCob$ and representatives $\sigma$ of the fundamental class of $S$.
Let $(\widehat{\xi},h)$ be an object of $ \lambda_*^{-1}[\xi]$. The groupoid $\lambda_*^{-1}[\xi \times \id_{[0,1]}]|_{(\widehat{\xi},h)}$ is contractible. For this reason we may fix the object $(\widehat{\xi}\times \id_{[0,1]}, h\times \id_{[0,1]}, \id) \in \lambda_*^{-1}[\xi \times \id_{[0,1]}]|_{(\widehat{\xi},h)} $ without loss of generality.
The upper composition evaluated on a parallel section $f$ at $(\widehat{\xi},h)$ is
\begin{equation}
\begin{aligned}
&  \big\langle (\widehat{\xi}\times \id_{[0,1]})^* \omega' \,,\, (-1)^{n}\,\sigma \times [0,1] \big\rangle \, \big\langle (\id_{[0,1]}\times \xi \times \id_{[0,1]})^* \theta \,,\, [0,1] \times (-1)^{n}\, (\sigma \times [0,1]) \big\rangle \\
& \hspace{1cm} \cdot \big\langle (\id_{[0,1]}\times \widehat{\xi} \ )^* \omega' \,,\, [0,1] \times \sigma \big\rangle \ f(\widehat{\xi},h) \\[4pt]
& \hspace{2cm} = \big\langle (\id_{[0,1]}\times \xi \times \id_{[0,1]})^* \theta \,,\, [0,1] \times (-1)^{n}\,(\sigma \times [0,1]) \big\rangle \ f(\widehat{\xi},h) \\[4pt]
& \hspace{3cm} = \big\langle \xi ^* \theta \,,\, \pr_2 \big([0,1] \times (-1)^{n}\,(\sigma \times [0,1])\big) \big\rangle \ f(\widehat{\xi},h) \\[4pt]
& \hspace{4cm} = f(\widehat{\xi},\id_{[0,1]}\times \xi) \ ,
\end{aligned}
\end{equation}
where in the first equality we used the fact that $f$ is a parallel section and in the last equality that the projection to the second factor $\pr_2 \big([0,1] \times (-1)^{n}\,(\sigma \times [0,1])\big)$ is a boundary, which follows from 
\[
\partial\,\pr_2 \big([0,1] \times (-1)^{n}\,(\sigma \times [0,1])\big) = 0
\]
and by dimensional reasons.

We now show that the diagram \eqref{Condition twisted functoriality} commutes. Again it is enough to check the commutativity on elements of the coends. Let $(\Sigma_a,\varphi_a)\colon (S_1,\xi_1)\longrightarrow (S_2,\xi_2)$ and $(\Sigma_b,\varphi_b)\colon (S_2,\xi_2)\longrightarrow (S_3,\xi_3)$ be 1-morphisms in $\EGCob$. We fix representatives $\sigma_1$, $\sigma_2$ and $\sigma_3$ for the fundamental classes of $S_1$, $S_2$ and $S_3$, respectively. 
The upper composition in \eqref{Condition twisted functoriality} corresponds to the linear map
\begin{align}
\Sigma_b^{\varphi_b}(\sigma_2,\sigma_3)\otimes_\C \Sigma_a^{\varphi_a}(\sigma_1,\sigma_2)\otimes_\C \tilde{Z}^{(S_1,\xi_1)}_{\omega'}(\sigma_1) & \longrightarrow \tilde{Z}^{(S_3,\xi_3)}_{\omega'}(\sigma_3) \\
\Lambda_b\otimes_\C \Lambda_a \otimes_\C f(\,\cdot\,) & \longmapsto \tilde{f}(\,\cdot\,)
\end{align}
with 
\begin{align}
\tilde{f}(\widehat{\xi}_3,h_3) &= \int_{(\widehat{\varphi},g,\widehat{h}\,)\in \lambda_*^{-1}[\varphi_b \cup \varphi_a]|_{(\widehat{\xi}_3,h_3)}} \, \big\langle  \widehat{\varphi}^{\,*} \omega' \,,\, \Lambda_a +\Lambda_b \big\rangle \, \big\langle g^* \theta \,,\, [0,1]\times (\Lambda_a+\Lambda_b) \big\rangle \, \big\langle\,\widehat{h}^* \omega' \,,\, [0,1] \times \sigma_3 \big\rangle \\ & \hspace{11cm} \cdot  f(\widehat{\varphi}\,|_{S_1},\widehat{h}|_{S_1}) \ . \label{Eq: Uper composition}
\end{align}
The lower composition gives 
\begin{align}
\Sigma_b^{\varphi_b}(\sigma_2,\sigma_3)\otimes_\C \Sigma_a^{\varphi_a}(\sigma_1,\sigma_2)\otimes_\C \tilde{Z}^{(S_1,\xi_1)}_{\omega'}(\sigma_1) & \longrightarrow \tilde{Z}^{(S_3,\xi_3)}_{\omega'}(\sigma_3) \\
\Lambda_b\otimes_\C \Lambda_a \otimes_\C f (\,\cdot\,) & \longmapsto \underline{f}(\,\cdot\,)
\end{align}
with 
\begin{align}
\underline{f}(\widehat{\xi}_3,h_3) &= \int_{(\widehat{\varphi}_b,g_b,\widehat{h}_b)\in \lambda_*^{-1}[\varphi_b]|_{(\widehat{\xi}_3,h_3)}} \ \int_{(\widehat{\varphi}_a,g_a,\widehat{h}_a)\in \lambda_*^{-1}[\varphi_a]|_{(\widehat{\varphi}_b,\widehat{h}_b)|_{S_2}}} \,
\big\langle \widehat{\varphi}_b^{\,*} \omega' \,,\, \Lambda_b \big\rangle \, \big\langle g_b^* \theta \,,\, [0,1]\times \Lambda_b \big\rangle \label{Eq: Lower composition} \\
& \hspace{1cm} \cdot 
\big\langle\, \widehat{h}^* \omega' \,,\, [0,1] \times \sigma_3 \big\rangle \, \big\langle \widehat{\varphi}_a^{\,*} \omega' \,,\, \Lambda_a \big\rangle \, \big\langle g_a^* \theta \,,\, [0,1]\times \Lambda_a \big\rangle \, \big\langle\, \widehat{h}^* \omega' \,,\, [0,1] \times \sigma_2 \big\rangle \ f(\widehat{\varphi}_a|_{S_1},\widehat{h}_a|_{S_1}) \ . 
\end{align}
Using the descent property for the stack $\BunG$ of $G$-bundles we can write the domain of the first integral as a homotopy pullback\footnote{Here we use the homotopy invariance of the stack $\BunG$ to identify the evaluation on an open neighbourhood of $S_2$ with the evaluation on $S_2$.}
\begin{equation}
\begin{tikzcd}
\lambda_*^{-1}[\varphi_b \cup  \varphi_a]|_{(\widehat{\xi}_3,h_3)} \ar[d, "\pr_{\Sigma_a}",swap] \ar[r, "\pr_{\Sigma_b}"] & \lambda_*^{-1}[\varphi_b]|_{(\widehat{\xi}_3,h_3)}\ar[d,"\pr_{S_2}"] \\
\lambda_*^{-1}[\varphi_a] \ar[r,swap,"\pr_{S_2}"]& \lambda_*^{-1}[\varphi_a|_{S_2}]
\end{tikzcd}
\end{equation}
where $\pr_{\,\cdot\,}$ denotes the pullback functor over the
indicated submanifold. 
From the fibrewise characterization of homotopy
pullbacks~\cite{HomotopyPullback} it follows that 
\begin{equation}
\begin{tikzcd}
\pr_{\Sigma_b}^{-1}\big[(\widehat{\varphi}_b,\widehat{h}_b)\big]\ar[r] \ar[d, "\Xi", swap]& \lambda_*^{-1}[\varphi_b \cup  \varphi_a]|_{(\widehat{\xi}_3,h_3)} \ar[d, "\pr_{\Sigma_a}",swap] \ar[r, "\pr_{\Sigma_b}"] & \lambda_*^{-1}[\varphi_b]|_{(\widehat{\xi}_3,h_3)}\ar[d,"\pr_{S_2}"] \\
\pr_{S_2}^{-1}\big[(\widehat{\varphi}_b,\widehat{h}_b)|_{S_2}\big]\ar[r] & \lambda_*^{-1}[\varphi_a] \ar[r,swap,"\pr_{S_2}"]& \lambda_*^{-1}[\varphi_a|_{S_2}]
\end{tikzcd}
\end{equation}
is a homotopy commuting diagram containing an equivalence $\Xi$. The
equality of the integrals~\eqref{Eq: Uper composition} and \eqref{Eq:
  Lower composition} now follows from the generalized Cavalieri Principle \cite[Proposition~A.15]{OFK} applied to the functor 
\[\pr_{\Sigma_b} \colon \lambda_*^{-1}[\varphi_b \cup  \varphi_a]\big|_{(\widehat{\xi}_3,h_3)} \longrightarrow \lambda_*^{-1}[\varphi_b]\big|_{(\widehat{\xi}_3,h_3)} 
\] 
and the ordinary Cavalieri Principle~\cite[Proposition~A.14]{OFK} for $\Xi$.

Finally, the invariance with respect to gauge
transformations~\cite[eq.~(3.17)]{Parity} follows directly from the
invariances of the integrands which is part of the statement of Lemma~\ref{Lemma well-defined on morphisms}. 
\end{proof}
Similarly to the proof of Theorem~\ref{Theorem: Gauging}, it should be possible to show that the relative field theory gauges the $G$-symmetry. 
Let us explain in more detail what this means: The pullback
$i^*E_\theta$ along the inclusion $i\colon \ECob \longrightarrow \EGCob$ is naturally isomorphic to
the trivial theory $1\colon \ECob\longrightarrow \Tvs$. The pullback $i^* Z_{\omega'}\colon 1 
\Longrightarrow \tr i^*E_\theta\cong 1$ is a field theory
relative to the trivial theory. From \cite[Proposition 3.14]{Parity} it follows that $i^*Z_{\omega'}$ 
is an $n-1$ dimensional topological quantum field theory. This field theory comes with an internal
$G$ symmetry from the evaluation of $Z_{\omega'}$ on gauge transformations of the trivial bundle. 
Gauging the symmetry means that the field theory $i^*Z_{\omega'}$ recovers the Dijkgraaf-Witten theory
$Z_\omega$ together with its internal symmetry induced by the group extensions \begin{align}
1\longrightarrow D \overset{\iota}{\longrightarrow} \widehat{G}\overset{\lambda}{\longrightarrow} G \longrightarrow 1 \ .
\end{align}  

\begin{remark}
Let $S$ be a closed oriented $n{-}2$-dimensional manifold and $\sigma_S$ a representative of its fundamental class.
The general theory outlined in Section~\ref{Sec: Boundary theories}
implies that the vector spaces
$\tilde{Z}_\omega^{(S,\,\cdot\,)}(\sigma_S)$ form a projective
representation of $\BunG(S)$. The 2-cocycle $\alpha$ twisting the
projective representation is completely described by the coherence
isomorphisms for $E_\theta$. In~\cite[Theorem~4.5]{EHQFT} it is shown
that the class of this 2-cocycle is given by the transgression of
$\theta\in Z^n(BG;U(1))$ to the groupoid $\BunG(S)$, i.e. $\alpha$ is induced by the 2-cocycle $\tau_S\theta$ on the underlying mapping
space $|\BunG(S)|$ (with the compact-open topology) given by
\begin{align}
(\tau_S\theta)(\chi):=({\rm ev}^*\theta)(\chi\times\sigma_S)
\end{align}
for any 2-simplex $\chi:\Delta^2\longrightarrow |\BunG(S)|$, where
${\rm ev}:|\BunG(S)|\times S\longrightarrow BG$ is the evaluation
map. This generalizes the low-dimensional descriptions of anomalies
and projective representations on state spaces discussed
in~\cite[Section~2.1]{Tachikawa2016}: In the simplest
$n=1$ case, with $S=\{\ast\}$ the 2-cocycles $\alpha$ and $\theta$ may be identified,
and describe the same 2-cocycle specifying both the two-dimensional bulk $G$-symmetry protected
phase and the class of the
projective $G$-representation on the one-dimensional boundary state,
whereas for $n=3$ with $S=\mathbb{S}^1$ transgression induces a homomorphism
$H^3(BG;U(1))\longrightarrow H^2(BG;U(1))$ specifying the
two-dimensional $G$-symmetry protected phase on the
boundary of the three-dimensional $G$-symmetry protected phase.

In a more geometric language this means that the state spaces of the
gauged theory form a section of the transgression 2-line bundle of the
flat $n{-}1$-gerbe on the classifying space $BG$ described by $\theta$, as the classical
gauge theory corresponding to $\theta$ describes the parallel
transport for the $n{-}1$-gerbe. This 2-line bundle is trivial if and only if the corresponding 2-cocycle is a boundary. Hence the obstruction for the projective representation to form an honest representation is the non-triviality of the transgression 2-line bundle. 
\end{remark}

\begin{remark}
Let $(\Sigma,\varphi)\colon (S,\xi)\longrightarrow \varnothing$ be a 1-morphism in $\EGCob$. According to \eqref{Eq: Def combined state space} the state space of the composite system is given by
\begin{align}
Z_{\omega'\,{\rm bb}}(\Sigma,\varphi, S)= E_\theta(\Sigma,\varphi
  )[Z_{\omega'}(S,\varphi|_S)] \cong \Sigma^\varphi (\varnothing,
  \sigma_S)\otimes_\C \tilde{Z}^{(S,\varphi|_S)}_{\omega'}(\sigma_S) \ .
\end{align}  
It is independent of the choice of $\sigma_S$ up to unique isomorphism corresponding to the choice of a representative of the coend. The composite state space carries an honest representation of the gauge group $G$ described in~\eqref{Eq: Def action combined state space}. 
\end{remark}  

\appendix

\section{Lemmata}

\begin{lemma}\label{Lemma well-defined on morphisms}
$L_{\xi, \omega'} \colon \Fund_\theta(S,\xi)^{{\rm op}} \longrightarrow \big[\lambda_*^{-1}[\xi], \fvs\big]$ is a well-defined functor. 
\end{lemma}
\begin{proof}
Let $(S,\xi)$ be an object of $\EGCob$, $\sigma_S$ and $\sigma_S'$ representatives of the fundamental class of $S$, $\Lambda \in C_{n-1}(S) $ an $n{-}1$-chain satisfying $\partial \Lambda = \sigma_S'-\sigma_S$ and $\widehat{h}\colon (\widehat{\xi},h)\longrightarrow (\widehat{\xi}\,',h')$ a morphism in $\lambda_*^{-1}[\xi]$.
The only subtle part is the naturality of $L_{\xi, \omega'}(\Lambda)$, i.e. commutativity of the diagram
\begin{equation}
\begin{tikzcd}
 L_{\xi, \omega'}(\sigma'_S)(\widehat{\xi},h)=\C \ar[rr, "{\langle \, \widehat{h}^* \omega', [0,1]\times \sigma_S' \rangle\,\cdot}"] \ar[dd, "{\langle  \widehat{\xi}\,^* \omega', \Lambda \rangle^{-1}\, \langle h^*\theta , [0,1]\times \Lambda \rangle^{-1}\,\cdot} ", swap]& \ \ \ \ \ &  L_{\xi, \omega'}(\sigma'_S)(\widehat{\xi}\,',h')=\C \ar[dd, "{\langle  \widehat{\xi}\,'{}^* \omega',\Lambda \rangle^{-1}\, \langle h'^*\theta, [0,1]\times \Lambda \rangle^{-1}\,\cdot} "] \\
  & & \\
 L_{\xi, \omega'}(\sigma_S)(\widehat{\xi},h)=\C \ar[rr,swap, "{\langle \, \widehat{h}^* \omega', [0,1]\times \sigma_S \rangle\,\cdot}"]& \ \ \ \ \ &  L_{\xi, \omega'}(\sigma_S)(\widehat{\xi}\,',h')=\C
\end{tikzcd}
\end{equation}
We check this by calculating the lower path to get
\begin{equation}
\begin{aligned}
& \langle \widehat{\xi}\,^* \omega',\Lambda  \rangle^{-1} \, \langle
h^*\theta , [0,1]\times \Lambda \rangle^{-1}\, \langle \, \widehat{h}^* \omega', [0,1]\times \sigma_S \rangle \\[4pt]
&\qquad = 
\langle \widehat{\xi}\,^* \omega' ,  \Lambda \rangle^{-1}\, \langle h^*\theta , [0,1]\times \Lambda \rangle^{-1}\, \langle \, \widehat{h}^* \omega' ,[0,1]\times \sigma_S' \rangle\, \langle \, \widehat{h}^* \omega' , [0,1]\times \partial \Lambda \rangle^{-1} \\[4pt]
&\qquad = 
\langle \widehat{\xi}\,^* \omega',\Lambda \rangle^{-1}\, \langle h^*\theta , [0,1]\times \Lambda \rangle^{-1}\, \langle \, \widehat{h}^* \omega' ,[0,1]\times \sigma_S' \rangle\, \langle \, \widehat{h}^* \omega', -\partial([0,1]\times \Lambda )-\lbrace 0 \rbrace \times \Lambda + \lbrace 1 \rbrace \times \Lambda \rangle^{-1}\\[4pt]
&\qquad = 
\langle h^*\theta , [0,1]\times \Lambda \rangle^{-1}\, \langle \,
\widehat{h}^* \omega' , [0,1]\times \sigma_S' \rangle\, \langle \,
\widehat{h}^* \omega' ,  -\partial([0,1]\times \Lambda )
\rangle^{-1}\, \langle \widehat{\xi}\,'{}^* \omega' , \Lambda \rangle^{-1}\\[4pt]
&\qquad = 
\langle \widehat{\xi}\,'{}^* \omega',\Lambda \rangle^{-1}\, \langle h'^*\theta , [0,1]\times \Lambda \rangle^{-1}\, \langle \, \widehat{h}^* \omega' , [0,1]\times \sigma_S' \rangle \ ,
\end{aligned}
\end{equation}
where in the first step we used $\sigma_S=\sigma_S'-\partial \Lambda$,
in the second step the graded product rule for $\times$, in the third step that $\widehat{h}$ is a homotopy from $\widehat{\xi}$ to $\widehat{\xi}\,'$, and in the last step $\delta \omega'=\lambda^*\theta$ and $\lambda_*\widehat{h}= h'{}^{-1}\,h$.
\end{proof}

\begin{lemma}\label{Lem: Def on Mor}
$Z_{\omega'}(\Sigma,\varphi)_{\sigma_1}(f,\Lambda)(\,\cdot\,)$ is a well-defined parallel section and the collection $
Z_{\omega'}(\Sigma,\varphi)_{\sigma_1}$ defines a cowedge inducing the
desired maps, which again form a cowedge defining $Z_{\omega'}(\Sigma,\varphi)$.
\end{lemma}

\begin{proof}
We first check that the integrand of \eqref{Def: Action on invariant functions} is gauge-invariant, i.e. the integral is well-defined.
Let $\widehat{H}\colon (\widehat{\varphi},g,\widehat{h}\,)
\longrightarrow (\widehat{\varphi}\,',g',\widehat{h}') $ be an
isomorphism in $\lambda_*^{-1}[\varphi]|_{(\widehat{\xi}_2,h_2)}$,
i.e. a homotopy $\widehat{H}\colon \widehat{\varphi} \longrightarrow
\widehat{\varphi}\,'$ such that the diagrams
\begin{equation}
\begin{tikzcd}
\lambda_*\widehat{\varphi}\ar[rd,"g",swap] \ar[rr,"\lambda_*\widehat{H}"] & & \lambda_*\widehat{\varphi}\,' \ar[ld,"g'"] \\
 & \varphi & 
\end{tikzcd}\label{EQ: Condition1}
\end{equation}
and 
\begin{equation}
\begin{tikzcd}
\widehat{\varphi}\,|_{S_2}\ar[rd,"\widehat{h}",swap] \ar[rr,"\widehat{H}|_{S_2}"] & & \widehat{\varphi}\,'|_{S_2} \ar[ld,"\widehat{h}'"] \\
 & \widehat{\xi_2} & 
\end{tikzcd}\label{EQ: Condition2}
\end{equation}
commute.
We study the transformation of the first term in \eqref{Def: Action on invariant functions}.
The homotopy $\widehat{H}\colon [0,1]\times \Sigma \longrightarrow B\widehat{G}$ induces a chain homotopy between the chain maps $\widehat{\varphi}_*$ and $\widehat{\varphi}_*^{\,\prime}$ given by
\begin{align} 
{{\widehat{\varphi}}^{\,\prime}_*}{}_{p}- {\widehat{\varphi}_*}{}_{p}
  = \partial \circ H_p + H_{p-1} \circ \partial
\end{align}
for all $p\in\mathbb{Z}$, where $H_p := {\widehat{H}_*}{}_{p+1} \circ D_p$ and 
\begin{align} 
D_p : C_p(\Sigma) \longrightarrow C_{p+1}([0,1]\times \Sigma) \ , \quad c \longmapsto  [0,1] \times c \ .
\end{align} 
Hence
\begin{align}
 \langle  {\widehat{\varphi}\,'}{}^* \omega' ,\Lambda \rangle \, \langle  {\widehat{\varphi}}\,^* \omega',\Lambda \rangle^{-1} & = \langle  \omega' , \partial H_{n-1}(\Lambda) + H_{n-2}(\partial \Lambda) \rangle \\[4pt]
& = \langle (\lambda_*\widehat{H})^* \theta , [0,1]\times \Lambda \rangle \, \langle   \widehat{H}^* \omega' , [0,1]\times (\sigma_2-\sigma_1) \rangle \\[4pt]
&= \langle g^* \theta,  [0,1]\times \Lambda \rangle \, \langle g'^*
\theta , [0,1]\times \Lambda \rangle^{-1} \, \langle \, \widehat{h}^* \omega' , [0,1]\times \sigma_2 \rangle \, \langle \, {\widehat{h}'}{}^* \omega',[0,1]\times \sigma_2 \rangle^{-1} \\
& \hspace{6cm}  \cdot \langle\widehat{H}|_{S_1}^* \omega'  , [0,1]\times \sigma_1 \rangle^{-1} \\[4pt]
&= \langle  g^* \theta,  [0,1]\times \Lambda \rangle \, \langle   g'^*
\theta , [0,1]\times \Lambda \rangle^{-1} \, \langle \, \widehat{h}^* \omega' , [0,1]\times \sigma_2 \rangle \, \langle\, {\widehat{h}'}{}^* \omega',[0,1]\times \sigma_2 \rangle^{-1} \\
& \hspace{6cm} \cdot f(\widehat{\varphi}\,'|_{S_1},g'|_{S_1}) \, f(\widehat{\varphi}\,|_{S_1},g|_{S_1})^{-1} \ ,
\label{EQ: First Term}
\end{align}
where we used \eqref{EQ: Condition1} and \eqref{EQ: Condition2} in the
third step, and that $f(\,\cdot\,)$ is a parallel section in the last
step. Comparing \eqref{EQ: First Term} with the transformation
behaviours of the other terms in \eqref{Def: Action on invariant functions} shows that the integrand is gauge-invariant.

We now check that \eqref{Def: Action on invariant functions} is a
parallel section of $L_{\xi_2,\omega'}(\sigma_2)\colon
\lambda_*^{-1}[\xi_2]\longrightarrow \fvs$. Let $\widehat{h}' \colon
(\widehat{\xi}_2,h_2) \longrightarrow
(\widehat{\xi}_2^{\,\prime},h_2')$ be a morphism in
$\lambda_*^{-1}[\xi_2]$. The map $\widehat{h}'$ induces an equivalence of groupoids 
\begin{align}
\widehat{h}'_* \colon \lambda_*^{-1}[\varphi]|_{(\widehat{\xi}_2,h_2)} &\longrightarrow \lambda_*^{-1}[\varphi]|_{(\widehat{\xi}_2^{\,\prime},h_2')} \\
(\widehat{\varphi},g,\widehat{h}\,) & \longmapsto (\widehat{\varphi},g,\widehat{h}'\circ \widehat{h}\,) \ .
\end{align}
Pulling back the integrand of \eqref{Def: Action on invariant
  functions} along this equivalence, reparametrization invariance of
the integral over groupoids and $ \big\langle (\,\widehat{h}' \circ
\widehat{h}\,)^* \omega', [0,1]\times\sigma_2 \big\rangle = \big\langle\,
\widehat{h}^* \omega' , [0,1]\times \sigma_2 \big\rangle + \big\langle\,
{\widehat{h}'}{}^* \omega' , [0,1]\times \sigma_2 \big\rangle$ then shows that \eqref{Def: Action on invariant functions} is a parallel section.  

Next we show that this defines a cowedge with respect to $\sigma_1$,
i.e. the diagram
\begin{equation}
\begin{tikzcd}
\Sigma^\varphi(\sigma_2,\sigma_1)\otimes_\C \tilde{Z}^{(S_1,\xi_1)}_{\omega'}(\sigma_1') \ar[rr,"\Sigma^\varphi(\zeta)\otimes_\C \id"] \ar[dd,"\id \otimes_\C \tilde{Z}^{(S_1,\xi_1)}_{\omega'}(\zeta)", swap] & &
\Sigma^\varphi(\sigma_2,\sigma_1')\otimes_\C
\tilde{Z}^{(S_1,\xi_1)}_{\omega'}(\sigma_1') \ar[dd] \\ & & \\
\Sigma^\varphi(\sigma_2,\sigma_1)\otimes_\C \tilde{Z}^{(S_1,\xi_1)}_{\omega'}(\sigma_1)  \ar[rr]& &  \tilde{Z}^{(S_2,\xi_2)}_{\omega'}(\sigma_2) 
\end{tikzcd}
\end{equation}
commutes for all morphisms $(\zeta\colon \sigma_1 \longrightarrow
\sigma_1') \in \Fund_\theta(S_1,\xi_1)$. The upper composition
evaluated on an element $\Lambda \otimes_\C f( \,\cdot\, ) \in
\Sigma^\varphi(\sigma_2,\sigma_1)\otimes_\C
\tilde{Z}^{(S_1,\xi_1)}_{\omega'}(\sigma_1')$ gives
\begin{align*}
&\int_{(\widehat{\varphi},g,\widehat{h}\,)\in \lambda_*^{-1}[\varphi]|_{(\widehat{\xi}_2,h_2)}} \, \langle  \widehat{\varphi}\,^* \omega',  \Lambda-\zeta \rangle \, \langle  g^* \theta, [0,1] \times (\Lambda-\zeta) \rangle \, \langle\, \widehat{h}^* \omega', [0,1] \times \sigma_2 \rangle \ f(\widehat \varphi\,|_{S_1}, g|_{S_1}) \\[4pt]
& \quad =
\int_{(\widehat{\varphi},g,\widehat{h}\,)\in \lambda_*^{-1}[\varphi]|_{(\widehat{\xi}_2,h_2)}} \, \langle g^* \theta , [0,1]\times \zeta \rangle^{-1} \, \langle \widehat{\varphi}\,^* \omega' ,  \zeta  \rangle^{-1} \, \langle  \widehat{\varphi}\,^* \omega' , \Lambda \rangle \, \langle g^* \theta , [0,1]\times \Lambda  \rangle \, \langle\, \widehat{h}^* \omega' , [0,1] \times \sigma_2 \rangle\\
& \hspace{14cm} \cdot f(\widehat \varphi\,|_{S_1}, g|_{S_1})
\\[4pt]
& \quad =
\int_{(\widehat{\varphi},g,\widehat{h}\,)\in \lambda_*^{-1}[\varphi]|_{(\widehat{\xi}_2,h_2)}} \, { \langle g|_{S_1}^* \theta , [0,1] \times \zeta \rangle^{-1} \, \langle \widehat{\varphi}\,|_{S_1}^* \omega' , \zeta \rangle^{-1}} \, \langle \widehat{\varphi}\,^* \omega' ,\Lambda \rangle \, \langle g^* \theta , [0,1]\times \Lambda  \rangle \, \langle\, \widehat{h}^* \omega' , [0,1] \times \sigma_2 \rangle \\
& \hspace{14cm} \cdot f(\widehat \varphi\,|_{S_1}, g|_{S_1}) \ .
\end{align*}
Comparing the term $\langle g|_{S_1}^* \theta , [0,1] \times \zeta \rangle^{-1} \, \langle \widehat{\varphi}\,|_{S_1}^* \omega' , \zeta \rangle^{-1}$ here with \eqref{Def: L natural} makes it clear that this is the same as the lower composition. 

Finally, we have to show naturality with respect to $\sigma_1$. This follows from the same calculation as above if we replace $\zeta$ with $-\zeta$ everywhere.
\end{proof}



\begin{thebibliography}{BDSPV15}

\bibitem[Ati88]{Atiyah1988}
M.~F. Atiyah.
\newblock Topological quantum field theories.
\newblock {\em Publ. Math. IHÉS}, 68:175--186, 1988.

\bibitem[BBF05]{Blanco2005}
V.~Blanco, M.~Bullejos, and E.~Faro.
\newblock {Categorical non-abelian cohomology and the Schreier theory of
  groupoids}.
\newblock {\em Math. Z.}, 251(1):41--59, 2005,
  \href{https://arxiv.org/abs/math/0410202}{arXiv:math.CT/0410202}.

\bibitem[BCH19]{Benini2018}
F.~Benini, C.~C{\'o}rdova, and P.-S. Hsin.
\newblock {On 2-group global symmetries and their anomalies}.
\newblock {\em J. High Energy Phys.}, 03:118, 2019, \href{https://arxiv.org/abs/1803.09336}{arXiv:1803.09336
  [hep-th]}.

\bibitem[BDSPV15]{bartlett2015modular}
B.~Bartlett, C.~L. Douglas, C.~J. Schommer-Pries, and J.~Vicary.
\newblock Modular categories as representations of the three-dimensional
  bordism 2-category.
\newblock 2015, \href{https://arxiv.org/abs/1509.06811}{arXiv:1509.06811
  [math.AT]}.

\bibitem[BHW10]{groupoidfication}
J.~C. Baez, A.~E. Hoffnung, and C.~D. Walker.
\newblock {Higher-dimensional algebra VII: Groupoidification}.
\newblock {\em Theory Appl. Categ.}, 24(18):489--553, 2010,
  \href{https://arxiv.org/abs/0908.4305}{arXiv:0908.4305 [math.QA]}.

\bibitem[BS06]{nCat=Cohomology}
J.{\,}C. {Baez} and M.~{Shulman}.
\newblock {Lectures on $n$-categories and cohomology}.
\newblock In {\em Towards Higher Categories}, pages 1--68. Springer, New York,
  2006, \href{https://arxiv.org/abs/math/0608420}{arXiv:math.CT/0608420}.

\bibitem[BS17]{Bunk2017}
S.~Bunk and R.~J. Szabo.
\newblock {Topological insulators and the Kane-Mele invariant: Obstruction and
  localization theory}.
\newblock 2017, \href{https://arxiv.org/abs/1712.02991}{arXiv:1712.02991
  [math-ph]}.

\bibitem[CCW17a]{Cong2017}
I.~Cong, M.~Cheng, and Z.~Wang.
\newblock {Defects between gapped boundaries in two-dimensional topological
  phases of matter}.
\newblock {\em Phys. Rev. B}, 96(19):195129, 2017,
  \href{https://arxiv.org/abs/1703.03564}{arXiv:1703.03564 [cond-mat.str-el]}.

\bibitem[CCW17b]{Cong:2017ffh}
I.~Cong, M.~Cheng, and Z.~Wang.
\newblock {Hamiltonian and algebraic theories of gapped boundaries in
  topological phases of matter}.
\newblock {\em Commun. Math. Phys.}, 355:645--689, 2017,
  \href{https://arxiv.org/abs/1707.04564}{arXiv:1707.04564 [cond-mat.str-el]}.

\bibitem[CGLW13]{Chen:2011pg}
X.~Chen, Z.-C. Gu, Z.-X. Liu, and X.-G. Wen.
\newblock {Symmetry-protected topological orders and the group cohomology of
  their symmetry group}.
\newblock {\em Phys. Rev. B}, 87(15):155114, 2013,
  \href{https://arxiv.org/abs/1106.4772}{arXiv:1106.4772 [cond-mat.str-el]}.

\bibitem[CGW10]{Chen:2010gda}
X.~Chen, Z.-C. Gu, and X.-G. Wen.
\newblock {Local unitary transformation, long-range quantum entanglement,
  wavefunction renormalization, and topological order}.
\newblock {\em Phys. Rev. B}, 82:155138, 2010,
  \href{https://arxiv.org/abs/1004.3835}{arXiv:1004.3835 [cond-mat.str-el]}.

\bibitem[CMM97]{Carey1997}
A.~L. Carey, M.~K. Murray, and J.~Mickelsson.
\newblock {Index theory, gerbes, and Hamiltonian quantization}.
\newblock {\em Commun. Math. Phys.}, 183(3):707--722, 1997,
  \href{https://arxiv.org/abs/hep-th/9511151}{arXiv:hep-th/9511151}.

\bibitem[CPS06]{HomotopyPullback}
W.~{Chacholski}, W.~{Pitsch}, and J.~{Scherer}.
\newblock {Homotopy pullback squares up to localization}.
\newblock {\em Contemp. Math.}, 399:55--72, 2006,
  \href{https://arxiv.org/abs/math/0501250}{arXiv:math.AT/0501250}.

\bibitem[Del17]{Delcamp2017}
C.~Delcamp.
\newblock {Excitation basis for $(3{+}1)D$ topological phases}.
\newblock {\em J. High Energy Phys.}, 12:128, 2017,
  \href{https://arxiv.org/abs/1709.04924}{arXiv:1709.04924 [hep-th]}.

\bibitem[DPR90]{Twisted_Drinfeld_double}
R.~Dijkgraaf, V.~Pasquier, and P.~Roche.
\newblock {Quasi-Hopf algebras, group cohomology and orbifold models}.
\newblock {\em Nucl. Phys. B Proc. Suppl.}, 18:60--72, 1990.

\bibitem[DT18]{Delcamp2018}
C.~Delcamp and A.~Tiwari.
\newblock {From gauge to higher gauge models of topological phases}.
\newblock {\em J. High Energy Phys.}, 10:049, 2018,
  \href{https://arxiv.org/abs/1802.10104}{arXiv:1802.10104 [cond-mat.str-el]}.

\bibitem[DW90]{DijkgraafWitten}
R.~Dijkgraaf and E.~Witten.
\newblock Topological gauge theories and group cohomology.
\newblock {\em Commun. Math. Phys.}, 129(2):393--429, 1990.

\bibitem[ENOM10]{2009arXiv0909.3140E}
P.~I. {Etingof}, D.~{Nikshych}, V.~{Ostrik}, and E.~{Meir}.
\newblock {Fusion categories and homotopy theory}.
\newblock {\em Quantum Topol.}, 1(3):209--273, 2010,
  \href{https://arxiv.org/abs/0909.3140}{arXiv:0909.3140 [math.QA]}.

\bibitem[FH16]{Freed:2016rqq}
D.~S. Freed and M.~J. Hopkins.
\newblock {Reflection positivity and invertible topological phases}.
\newblock 2016, \href{https://arxiv.org/abs/1604.06527}{arXiv:1604.06527
  [hep-th]}.

\bibitem[FM13]{FreedMoore}
D.~S. Freed and G.~W. Moore.
\newblock Twisted equivariant matter.
\newblock {\em Ann. Henri Poincar{\'e}}, 14(8):1927--2023, 2013,
  \href{https://arxiv.org/abs/1208.5055}{arXiv:1208.5055 [hep-th]}.

\bibitem[FPSV15]{BrauerGroup}
J.~Fuchs, J.~Priel, C.~Schweigert, and A.~Valentino.
\newblock {On the Brauer groups of symmetries of abelian Dijkgraaf-Witten
  theories}.
\newblock {\em Commun. Math. Phys.}, 339(2):385--405, 2015,
  \href{https://arxiv.org/abs/1404.6646}{arXiv:1404.6646 [hep-th]}.

\bibitem[FQ93]{FreedQuinn}
D.~S. Freed and F.~Quinn.
\newblock {Chern-Simons theory with finite gauge group}.
\newblock {\em Commun. Math. Phys.}, 156(3):435--472, 1993,
  \href{https://arxiv.org/abs/hep-th/9111004}{arXiv:hep-th/9111004}.

\bibitem[Fre95]{FreedCS}
D.~S. Freed.
\newblock {Classical Chern-Simons theory 1}.
\newblock {\em Adv. Math.}, 113:237--303, 1995,
  \href{https://arxiv.org/abs/hep-th/9206021}{arXiv:hep-th/9206021}.

\bibitem[Fre14a]{FreedAnomalies}
D.~S. Freed.
\newblock Anomalies and invertible field theories.
\newblock {\em Proc. Symp. Pure Math.}, 88:25--46, 2014,
  \href{https://arxiv.org/abs/1404.7224}{arXiv:1404.7224 [hep-th]}.

\bibitem[Fre14b]{Freed:2014eja}
D.~S. Freed.
\newblock {Short-range entanglement and invertible field theories}.
\newblock 2014, \href{https://arxiv.org/abs/1406.7278}{ arXiv:1406.7278
  [cond-mat.str-el]}.

\bibitem[FT14]{RelativeQFT}
D.~S. Freed and C.~Teleman.
\newblock Relative quantum field theory.
\newblock {\em Commun. Math. Phys.}, 326(2):459--476, 2014,
  \href{https://arxiv.org/abs/1212.1692}{arXiv:1212.1692 [hep-th]}.

\bibitem[FT18]{Freed:2018cec}
D.~S. Freed and C.~Teleman.
\newblock {Topological dualities in the Ising model}.
\newblock 2018, \href{https://arxiv.org/abs/1806.00008}{arXiv:1806.00008
  [math.AT]}.

\bibitem[FV15]{BoundaryConditionsTQFT}
D.~Fiorenza and A.~Valentino.
\newblock Boundary conditions for topological quantum field theories, anomalies
  and projective modular functors.
\newblock {\em Commun. Math. Phys.}, 338(3):1043--1074, 2015,
  \href{https://arxiv.org/abs/1409.5723}{arXiv:1409.5723 [math.QA]}.

\bibitem[GJF19]{Gaiotto:2017zba}
D.~Gaiotto and T.~Johnson-Freyd.
\newblock {Symmetry-protected topological phases and generalized cohomology}.
\newblock  {\em J. High Energy Phys.}, 05:007, 2019, \href{https://arxiv.org/abs/1712.07950}{arXiv:1712.07950
  [hep-th]}.

\bibitem[GKSW15]{Gaiotto2014}
D.~Gaiotto, A.~Kapustin, N.~Seiberg, and B.~Willett.
\newblock {Generalized global symmetries}.
\newblock {\em J. High Energy Phys.}, 02:172, 2015,
  \href{https://arxiv.org/1412.5148}{arXiv:1412.5148 [hep-th]}.

\bibitem[HLY14]{Huang2014}
H.-L. Huang, G.~Liu, and Y.~Ye.
\newblock The braided monoidal structures on a class of linear {Gr}-categories.
\newblock {\em Algebr. Represent. Theor.}, 17(4):1249--1265, 2014,
  \href{https://arxiv.org/abs/1206.5402}{arXiv:1206.5402 [math.QA]}.

\bibitem[HS53]{HS53}
G.~Hochschild and J.-P. Serre.
\newblock {Cohomology of group extensions}.
\newblock {\em Trans. Amer. Math. Soc.}, 74(1):110--134, 1953.

\bibitem[Hue81]{HUEBSCHMANN1981296}
J.~Huebschmann.
\newblock {Automorphisms of group extensions and differentials in the
  Lyndon-Hochschild-Serre spectral sequence}.
\newblock {\em J. Algebra}, 72(2):296--334, 1981.

\bibitem[HZvK17]{He:2016xpi}
H.~He, Y.~Zheng, and C.~von Keyserlingk.
\newblock {Field theories for gauged symmetry-protected topological phases:
  Non-abelian anyons with abelian gauge group $\mathbb Z_2^{ 3}$}.
\newblock {\em Phys. Rev. B}, 95(3):035131, 2017,
  \href{https://arxiv.org/abs/1608.05393}{arXiv:1608.05393 [cond-mat.str-el]}.

\bibitem[JFS17]{Johnson-Freyd:2017ykw}
T.~Johnson-Freyd and C.~Scheimbauer.
\newblock {(Op)lax natural transformations, twisted quantum field theories, and
  “even higher” Morita categories}.
\newblock {\em Adv. Math.}, 307:147--223, 2017,
  \href{https://arxiv.org/abs/1502.06526}{arXiv:1502.06526 [math.CT]}.

\bibitem[Joy08]{Joyner08aprimer}
D.~Joyner.
\newblock {A primer on computational group homology and cohomology using {\tt
  GAP} and {\tt SAGE}}.
\newblock {\em Algebr. Discr. Math.}, 1:159--191, 2008,
  \href{https://arxiv.org/abs/0706.0549}{arXiv:0706.0549 [math.GR]}.

\bibitem[Kap14]{Kapustin2014}
A.~Kapustin.
\newblock {Symmetry-protected topological phases, anomalies, and cobordisms:
  Beyond group cohomology}.
\newblock 2014, \href{https://arxiv.org/abs/1403.1467}{arXiv:1403.1467
  [cond-mat.str-el]}.

\bibitem[Kar85]{ProjectiveRep}
G.~Karpilovsky.
\newblock {\em Projective Representations of Finite Groups}.
\newblock Marcel Dekker, 1985.

\bibitem[KS14]{Kapustin:2014gua}
A.~Kapustin and N.~Seiberg.
\newblock {Coupling a QFT to a TQFT and duality}.
\newblock {\em J. High Energy Phys.}, 04:001, 2014,
  \href{https://arxiv.org/1401.0740}{arXiv:1401.0740 [hep-th]}.

\bibitem[KT14]{Kapustin:Symmetries}
A.~Kapustin and R.~Thorngren.
\newblock {Anomalies of discrete symmetries in various dimensions and group
  cohomology}.
\newblock 2014, \href{https://arxiv.org/abs/1404.3230}{arXiv:1404.3230
  [hep-th]}.

\bibitem[KT17a]{Kapustin2013}
A.~Kapustin and R.~Thorngren.
\newblock {Higher symmetry and gapped phases of gauge theories}.
\newblock {\em {Progr. Math.}}, {324}:{177--202}, 2017,
  \href{https://arxiv.org/abs/1309.4721}{arXiv:1309.4721 [hep-th]}.

\bibitem[KT17b]{Kapustin2017}
A.~Kapustin and A.~Turzillo.
\newblock Equivariant topological quantum field theory and symmetry-protected
  topological phases.
\newblock {\em J. High Energy Phys.}, 03:006, 2017,
  \href{https://arxiv.org/abs/1504.01830}{arXiv:1504.01830 [cond-mat.str-el]}.

\bibitem[KTTW15]{Kapustin2015}
A.~Kapustin, R.~Thorngren, A.~Turzillo, and Z.~Wang.
\newblock Fermionic symmetry-protected topological phases and cobordisms.
\newblock {\em J. High Energy Phys.}, 12:052, 2015,
  \href{https://arxiv.org/abs/1406.7329}{arXiv:1406.7329 [cond-mat.str-el]}.

\bibitem[KV94]{KV94}
M.~Kapranov and V.~Voevodsky.
\newblock {Braided monoidal 2-categories and Manin-Schechtman higher braid
  groups}.
\newblock {\em J. Pure Appl. Algebra}, 92:241{\textendash}267, 1994.

\bibitem[LG12]{Levin2012}
M.~Levin and Z.-C. Gu.
\newblock {Braiding statistics approach to symmetry-protected topological
  phases}.
\newblock {\em Phys. Rev. B}, 86:115109, 2012,
  \href{https://arxiv.org/abs/1202.3120}{arXiv:1202.3120 [cond-mat.str-el]}.

\bibitem[LP17a]{Lentner:2015pla}
S.~D. Lentner and J.~Priel.
\newblock {A decomposition of the Brauer-Picard group of the representation
  category of a finite group}.
\newblock {\em J. Algebra}, 489:264--309, 2017,
  \href{https://arxiv.org/abs/1506.07832}{arXiv:1506.07832 [math.QA]}.

\bibitem[LP17b]{2017arXiv170205133L}
S.~D. {Lentner} and J.~{Priel}.
\newblock {Three natural subgroups of the Brauer-Picard group of a Hopf algebra
  with applications}.
\newblock {\em Bull. Belg. Math. Soc. Simon Stevin}, 24:1--34, 2017,
  \href{https://arxiv.org/abs/1702.05133}{arXiv:1702.05133 [math.QA]}.

\bibitem[MNS12]{NMS}
J.~{Maier}, T.~{Nikolaus}, and C.~{Schweigert}.
\newblock {Equivariant modular categories via Dijkgraaf-Witten theory}.
\newblock {\em Adv. Theor. Math. Phys.}, 16(1):289--358, 2012,
  \href{https://arxiv.org/abs/1103.2963}{arXiv:1103.2963 [math.QA]}.

\bibitem[Map]{Maple}
Maple 18. 
\newblock{\em Maplesoft, a division of Waterloo Maple Inc., Waterloo, Ontario.}


\bibitem[Mon15]{MonnierHamiltionianAnomalies}
S.~Monnier.
\newblock Hamiltonian anomalies from extended field theories.
\newblock {\em Commun. Math. Phys.}, 338(3):1327--1361, 2015,
  \href{https://arxiv.org/abs/1410.7442v2}{arXiv:1410.7442 [hep-th]}.

\bibitem[Mor15]{Morton}
J.~C. Morton.
\newblock {Cohomological twisting of 2-linearization and extended TQFT}.
\newblock {\em J. Homotopy Relat. Struct.}, 10:127--187, 2015,
  \href{https://arxiv.org/abs/1003.5603}{arXiv:1003.5603 [math.QA]}.

\bibitem[MS18]{Parity}
L.~Müller and R.~J. Szabo.
\newblock {Extended quantum field theory, index theory and the parity anomaly}.
\newblock {\em Commun. Math. Phys.}, 362(3):1049--1109, 2018,
  \href{https://arxiv.org/abs/1709.03860}{arXiv:1709.03860 [hep-th]}.

\bibitem[MW18]{EHQFT}
L.~{M{\"u}ller} and L.~Woike.
\newblock {Parallel transport of higher flat gerbes as an extended homotopy
  quantum field theory}.
\newblock 2018, \href{https://arxiv.org/abs/1802.10455}{arXiv:1802.10455
  [math.QA]}.

\bibitem[Nas91]{NashBook}
C.~Nash.
\newblock {\em Differential Topology and Quantum Field Theory}.
\newblock Academic Press, London, 1991.

\bibitem[Seg88]{Segal1988}
G.~B. Segal.
\newblock The definition of conformal field theory.
\newblock In {\em Differential Geometrical Methods in Theoretical Physics},
  pages 165--171. Springer Netherlands, 1988.

\bibitem[Seg11]{FelixKleinSegal}
G.~B. Segal.
\newblock {Three roles of quantum field theory.}
\newblock {\em {Felix Klein Lectures}}, 2011.
\newblock Available at
  {\small\href{http://www.mpim-bonn.mpg.de/node/3372/abstracts}{\tt
  http://www.mpim-bonn.mpg.de/node/3372/abstracts}}.

\bibitem[ST11]{Stolz:2011zj}
S.~Stolz and P.~Teichner.
\newblock {Supersymmetric field theories and generalized cohomology}.
\newblock {\em Proc. Symp. Pure Math.}, 83:279--340, 2011,
  \href{https://arxiv.org/abs/1108.0189}{arXiv:1108.0189 [math.AT]}.

\bibitem[SW19]{OFK}
C.~{Schweigert} and L.~{Woike}.
\newblock {Orbifold construction for topological field theories}.
\newblock {\em J. Pure Appl. Algebra}, 223:1167-1192, 2019, \href{https://arxiv.org/abs/1705.05171}{arXiv:1705.05171
  [math.QA]}.

\bibitem[SW18a]{EOFK}
C.~{Schweigert} and L.~{Woike}.
\newblock {Extended homotopy quantum field theories and their orbifoldization}.
\newblock 2018, \href{https://arxiv.org/abs/1802.08512}{arXiv:1802.08512
  [math.QA]}.

\bibitem[SW18b]{SWParallelSections}
C.~Schweigert and L.~Woike.
\newblock A parallel section functor for 2-vector bundles.
\newblock {\em Theory Appl. Categ.}, 33(23):644--690, 2018,
  \href{https://arxiv.org/abs/1711.08639}{arXiv:1711.08639 [math.CT]}.

\bibitem[Tac17]{Tachikawa:2017gyf}
Y.~Tachikawa.
\newblock {On gauging finite subgroups}.
\newblock 2017, \href{https://arxiv.org/abs/1712.09542}{arXiv:1712.09542
  [hep-th]}.

\bibitem[TCSR18]{Tiwari2017}
A.~Tiwari, X.~Chen, K.~Shiozaki, and S.~Ryu.
\newblock {Bosonic topological phases of matter: Bulk-boundary correspondence,
  symmetry-protected topological invariants, and gauging}.
\newblock {\em Phys. Rev. B}, 97(24):245133, 2018,
  \href{https://arxiv.org/abs/1710.04730}{arXiv:1710.04730 [cond-mat.str-el]}.

\bibitem[TE18]{Thorngren:2018wwt}
R.~Thorngren and D.~V. Else.
\newblock {Gauging spatial symmetries and the classification of topological
  crystalline phases}.
\newblock {\em Phys. Rev. X}, 8(1):011040, 2018,
  \href{https://arxiv.org/1612.00846}{arXiv:1612.00846 [cond-mat.str-el]}.

\bibitem[tH80]{tHooft1979}
G.~'t~Hooft.
\newblock {Naturalness, chiral symmetry, and spontaneous chiral symmetry
  breaking}.
\newblock {\em NATO Sci. Ser. B}, 59:135--157, 1980.

\bibitem[Tho17]{Thorngren:2017vzn}
R.~Thorngren.
\newblock {Topological terms and phases of sigma-models}.
\newblock 2017, \href{https://arxiv.org/abs/1710.02545}{arXiv:1710.02545
  [cond-mat.str-el]}.

\bibitem[Tur10]{turaev2010homotopy}
V.~G Turaev.
\newblock {\em {Homotopy Quantum Field Theory}}.
\newblock European Mathematical Society, 2010.

\bibitem[TV14]{TV14}
V.~G. Turaev and A.~Virelizier.
\newblock {On three-dimensional homotopy quantum field theory~{II}: The surgery
  approach}.
\newblock {\em Intern. J. Math.}, 25(04):1450027, 2014.

\bibitem[TvK15]{Thorngren2015}
R.~Thorngren and C.~von Keyserlingk.
\newblock {Higher SPT's and a generalization of anomaly inflow}.
\newblock 2015, \href{https://arxiv.org/abs/1511.02929}{arXiv:1511.02929
  [cond-mat.str-el]}.

\bibitem[TY17]{Tachikawa2016}
Y.~Tachikawa and K.~Yonekura.
\newblock {On time-reversal anomaly of $(2{+}1)D$ topological phases}.
\newblock {\em Progr. Theor. Exp. Phys.}, 2017(3):033B04, 2017,
  \href{https://arxiv.org/abs/1610.07010}{arXiv:1610.07010 [hep-th]}.

\bibitem[Wei95]{Weibel}
C.~A. Weibel.
\newblock {\em An Introduction to Homological Algebra}.
\newblock Cambridge University Press, 1995.

\bibitem[Wen13]{Wen:2013oza}
X.-G. Wen.
\newblock {Classifying gauge anomalies through symmetry-protected trivial
  orders and classifying gravitational anomalies through topological orders}.
\newblock {\em Phys. Rev. D}, 88(4):045013, 2013,
  \href{https://arxiv.org/abs/1303.1803}{arXiv:1303.1803 [hep-th]}.

\bibitem[WG18]{PhysRevX.8.011055}
Q.-R. Wang and Z.-C. Gu.
\newblock Towards a complete classification of symmetry-protected topological
  phases for interacting fermions in three dimensions and a general group
  supercohomology theory.
\newblock {\em Phys. Rev. X}, 8:011055, 2018,
  \href{https://arxiv.org/abs/1703.10937}{arXiv:1703.10937 [cond-mat.str-el]}.

\bibitem[WHT{\etalchar{+}}18]{Wen2017}
X.~Wen, H.~He, A.~Tiwari, Y.~Zheng, and P.~Ye.
\newblock {Entanglement entropy for $(3{+}1)$-dimensional topological order
  with excitations}.
\newblock {\em Phys. Rev. B}, 97(8):085147, 2018,
  \href{https://arxiv.org/abs/1710.11168}{arXiv:1710.11168 [cond-mat.str-el]}.

\bibitem[Wil08]{TwistedDWandGerbs}
S.~Willerton.
\newblock The twisted {D}rinfeld double of a finite group via gerbes and finite
  groupoids.
\newblock {\em Algebr. Geom. Topol.}, 8(3):1419--1457, 2008,
  \href{https://arxiv.org/abs/math/0503266}{arXiv:math.QA/0503266}.

\bibitem[Wit15]{WittenString15}
E.~Witten.
\newblock Anomalies revisited.
\newblock {\em Lecture At Strings 2015}, 2015.
\newblock {Available at
  {\small\href{https://strings2015.icts.res.in/talkDocuments/6-2.00-2.30-Edward-Witten.pdf}{\tt
  https://strings2015.icts.res.in/talkDocuments/6-2.00-2.30-Edward-Witten.pdf}}}.

\bibitem[Wit16]{Witten:2016cio}
E.~Witten.
\newblock {The ``parity" anomaly on an unorientable manifold}.
\newblock {\em Phys. Rev. B}, 94(19):195150, 2016,
  \href{https://arxiv.org/abs/1605.02391}{arXiv:1605.02391 [hep-th]}.

\bibitem[WLHW18]{Wang2018}
H.~Wang, Y.~Li, Y.~Hu, and Y.~Wan.
\newblock {Gapped boundary theory of the twisted gauge theory model of
  three-dimensional topological orders}.
\newblock {\em J. High Energy Phys.}, 10:114, 2018,
  \href{https://arxiv.org/abs/1807.11083}{arXiv:1807.11083 [cond-mat.str-el]}.

\bibitem[WSW15]{Wang:2014tia}
J.~C. Wang, L.~H. Santos, and X.-G. Wen.
\newblock {Bosonic anomalies, induced fractional quantum numbers and degenerate
  zero modes: The anomalous edge physics of symmetry-protected topological
  states}.
\newblock {\em Phys. Rev. B}, 91(19):195134, 2015,
  \href{https://arxiv.org/1403.5256}{arXiv:1403.5256 [cond-mat.str-el]}.

\bibitem[WW15]{Wang:2014oya}
J.~C. Wang and X.-G. Wen.
\newblock {Non-abelian string and particle braiding in topological order:
  Modular $SL(3,\mathbb{Z})$ representation and $(3{+}1)$-dimensional twisted
  gauge theory}.
\newblock {\em Phys. Rev. B}, 91(3):035134, 2015,
  \href{https://arxiv.org/abs/1404.7854}{arXiv:1404.7854 [cond-mat.str-el]}.

\bibitem[WWH15]{PhysRevB.92.045101}
Y.~Wan, J.~C. Wang, and H.~He.
\newblock Twisted gauge theory model of topological phases in three dimensions.
\newblock {\em Phys. Rev. B}, 92:045101, 2015,
  \href{https://arxiv.org/abs/1409.3216}{arXiv:1409.3216 [cond-mat.str-el]}.

\bibitem[WWW18]{Wang:2017loc}
J.~C. Wang, X.-G. Wen, and E.~Witten.
\newblock {Symmetric gapped interfaces of SPT and SET states: Systematic
  constructions}.
\newblock {\em Phys. Rev. X}, 8(3):031048, 2018,
  \href{https://arxiv.org/1705.06728}{arXiv:1705.06728 [cond-mat.str-el]}.

\bibitem[Yos17]{Yoshida:2017xqa}
B.~Yoshida.
\newblock {Gapped boundaries, group cohomology and fault-tolerant logical
  gates}.
\newblock {\em Ann. Phys.}, 377:387--413, 2017,
  \href{https://arxiv.org/1509.03626}{arXiv:1509.03626 [cond-mat.str-el]}.

\bibitem[You18]{Young2018}
M.~B. Young.
\newblock {Orientation twisted homotopy field theories and twisted unoriented
  Dijkgraaf-Witten theory}.
\newblock 2018, \href{https://arxiv.org/abs/1810.04612}{arXiv:1810.04612
  [math.QA]}.

\end{thebibliography}

\newcommand{\etalchar}[1]{$^{#1}$}

\end{document}